\documentclass[12pt]{amsart}

\usepackage{upgreek, amssymb, amsmath}
\usepackage[margin=3.5 cm]{geometry}
\usepackage{graphicx}
\usepackage{color}
\usepackage{subfigure}
\newtheorem{theorem}{Theorem}[section]
\newtheorem{fact}{Fact}[section]
\newtheorem{claim}{Claim}[section]
\newtheorem{prop}{Proposition}[section]

\newtheorem{lemma}[theorem]{Lemma}
\newtheorem{conj}{Conjecture}[section]
\newtheorem{obs}{Observation}[section]

\theoremstyle{definition}
\newtheorem{definition}[theorem]{Definition}

\theoremstyle{remark}
\newtheorem{remark}[theorem]{Remark}

\numberwithin{equation}{section}

\newcommand{\abs}[1]{\lvert#1\rvert}
\newcommand{\norm}[1]{\lVert#1\rVert}
\DeclareMathOperator{\re}{Re}
\DeclareMathOperator{\im}{Im}
\newcommand{\ud}{\mathrm{d}}

\newcommand{\vertiii}[1]{{\left\vert\kern-0.25ex\left\vert\kern-0.25ex\left\vert #1 
    \right\vert\kern-0.25ex\right\vert\kern-0.25ex\right\vert}}
    
\def\be{\begin{equation}}
\def\ee{\end{equation}}
\def\Z{\mathbb{Z}}

\begin{document}

\title[]{Spectral theory of extended Harper's model and a question by Erd\H{o}s and Szekeres}


%
%
\author{ A. Avila}
\address{CNRS, IMJ-PRG, UMR 7586, Univ Paris Diderot, Sorbonne Paris
Cite, Sorbonnes Universit\'es, UPMC Univ Paris 06, F-75013, Paris, France \& IMPA, Estrada Dona Castorina 110, Rio de Janeiro, Brasil}
\email{artur@math.univ-paris-diderot.fr}
\author{S. Jitomirskaya}
\address{University of California, Department of Mathematics, Irvine CA-92717, USA.}
\email{szhitomi@uci.edu}
\author{C. A. Marx}
\address{Department of Mathematics, Oberlin College, Oberlin OH, 44074}
\email{cmarx@oberlin.edu}

\thanks{The work of A.A. has been partially supported by the ERC Starting Grant ``Quasiperiodic,'' by the Balzan project of Jacob Palis, and the grant ANR-15-CE40-0001. S. J. was a 2014-15 Simons Fellow.
This work was partially supported by NSF Grants DMS -1101578 and DMS-1401204.}




\begin{abstract}
The extended Harper's model, proposed by D.J. Thouless in 1983,
generalizes the famous almost Mathieu operator, allowing for a wider range
of lattice geometries (parametrized by three {\it coupling parameters})
by permitting 2D electrons to hop to both nearest and
next nearest neighboring (NNN) lattice sites, while still exhibiting its
characteristic symmetry (Aubry-Andr\'e duality).  Previous understanding of the
spectral theory of this model was restricted to two dual regions of the
parameter space, one of which is characterized by the positivity of the
Lyapunov exponent.  In this paper, we complete the picture with a
description of the spectral measures over the entire
remaining (self-dual) region, for all irrational values of the frequency
parameter (the magnetic flux in the model).  Most notably, we prove that
in the entire interior of this regime, the model
exhibits a collapse from purely ac spectrum to purely sc spectrum when the
NNN interaction becomes symmetric.
In physics literature, extensive numerical analysis had indicated
such ``spectral collapse,'' however so far not even a heuristic
argument for this phenomenon could be provided.  On the other hand,
in the remaining part of the self-dual region, the spectral measures are
singular continuous irrespective of such symmetry.
The analysis requires some rather delicate number theoretic estimates, which
ultimately depend on the solution of a problem posed by
Erd\H{o}s and Szekeres in \cite{es}.
\end{abstract}

\maketitle

\section{Introduction} \label{sec_intro}

One-dimensional quasiperiodic Schr\"odinger operators with analytic potentials have traditionally been studied by perturbative KAM schemes in two distinct regimes: ``large'' and ``small'' potential. Some 15 years ago it has become understood \cite{j,bg} that the ``large'' regime can be described in a non-perturbative way through a purely dynamical property: positivity of the Lyapunov exponent. Recently, a full nonperturbative (and  purely dynamical) characterization of the entire ``small'' regime for the case of one-frequency operators has also been established \cite{Avila_prep_ARC_1,Avila_prep_ARC_2}, thus leading to the division of energies in the spectrum into
\begin{enumerate}
\item {\it supercritical}, characterized by positive Lyapunov exponent (thus non-uniform hyperbolicity)
\item {\it subcritical}, characterized by the Lyapunov exponent vanishing in a strip of complexified phases (leading to almost-reducibility \cite{Avila_prep_ARC_1,Avila_prep_ARC_2})
\item {\it critical}, characterized as being neither of the two above.
\end{enumerate}

The first regime leads to Anderson localization for a.e. frequency, and the second to absolutely continuous spectrum for all frequencies. Various other interesting aspects of the first two regimes, each of which  holds on an open set, are also well 
understood. The third regime, which is the boundary of the first two,  cannot support absolutely continuous spectrum \cite{AvilaFayadKrikorian_2011} but otherwise largely remains a mystery, even  (or especially) in the most well studied case of the critical almost Mathieu operator. Even though measure-theoretically typical operators are acritical (so have no critical energies in the spectrum) \cite{global}, the most interesting/important operators from the point of view of physics turn out to be entirely critical, due to certain underlying symmetries! For example, such are the extended Harper's (and also the original Harper's) model for the - most physically relevant - case of  isotropic interactions.  Indeed, the duality transform, acting on the family of (long-range) quasiperiodic operators, often maps the first two regimes into each other, allowing for duality based conclusions, while mapping the third one into itself, making it self-dual and thus not allowing to use either localization or reducibility methods.  
 
In this paper we provide the first mechanism for exclusion of point spectrum and thus proof of singular continuous spectrum in the {\it critical} regime that works for all frequencies and a.e. phase\footnote{So far the existing arguments for exclusion of point spectrum have had nothing to do with criticality and have been limited to measure zero sets of frequency/phase \cite{Gordon_1976,as,js}.}. This allows us to prove singular continuity of the spectrum of extended Harper's model through its entire critical regime, for all frequencies, describing the spectral theory of the region that has resisted even heuristic explanations in physics literature.

While our argument is specific to extended Harper's model, we believe that certain features of it will be extendable to the general critical case. A simple particular case of the argument proves singular continuity of the spectrum of the critical almost Mathieu operator\footnote{This has been open since the proof in \cite{gjls} has a gap.}. We also are able to describe spectral theory of the extended Harper's model for all other values of the couplings in its three-dimensional parameter space, largely by putting together the facts proved in several other recent papers.      

The {\em{extended Harper's model}}  is a model from solid state physics defined by the following quasi-periodic Jacobi operator acting on $\mathit{l}^2(\mathbb{Z})$,
\begin{eqnarray} \label{eq_hamiltonian}
& (H_{\theta;\lambda, \alpha} \psi)_k := v(\theta + \alpha k) \psi_{k} + c_{\lambda}(\theta + \alpha k) \psi_{k+1} + \overline{c_{\lambda}(\theta + \alpha (k-1))} \psi_{k-1} ~\mbox{.}
\end{eqnarray}
Here, $\alpha$ is a fixed irrational, $\theta$ varies in $\mathbb{T}:=\mathbb{R}/\mathbb{Z}$, and
\begin{equation} \label{eq_hamiltonian1}
c_{\lambda}(\theta) := \lambda_{1} \mathrm{e}^{-2\pi i (\theta+\frac{\alpha}{2})} + \lambda_{2} + \lambda_{3} \mathrm{e}^{2 \pi i (\theta+\frac{\alpha}{2})} ~\mbox{,}
~ v(\theta)  := 2 \cos(2 \pi \theta) ~\mbox{.}
\end{equation}
We will generally understand $\mathbb{T}$ to be equipped with its Haar probability measure, denoted by $\mu$.

Physically, extended Harper's model describes the influence of a transversal magnetic field of flux $\alpha$ on a single tight-binding electron in a 2-dimensional crystal layer. In this context, the coupling triple $\lambda:=(\lambda_1, ~\lambda_2, ~\lambda_3) \in \mathbb{R}^3$ allows for {\em{nearest}} (expressed through $\lambda_2$) and {\em{next-nearest neighbor}} (NNN) interaction between lattice sites (expressed through $\lambda_1$ and $\lambda_3$). Without loss of generality, one may assume $0\leq \lambda_{2} ~\mbox{,} ~0 \leq \lambda_{1} + \lambda_{3}$ and at least one of $\lambda_{1} \mbox{,} ~\lambda_{2} \mbox{,} ~\lambda_{3}$ to be positive. 

Assuming a Bloch wave in one direction of the lattice plane with quasi momentum $\theta$, the conductivity properties in the transversal direction are governed by (\ref{eq_hamiltonian}). As common, we will refer to $\alpha$ as the {\em{frequency}} and $\theta$ as the {\em{phase}}. Proposed by D. J. Thouless in 1983 in context with the integer quantum Hall effect \cite{Thouless_1983}, extended Harper's model attracted significant attention in physics literature, and has been studied rigorously by Bellissard (e.g. \cite{bel}), Helffer et al (e.g. \cite{hel}), Shubin \cite{shu}, and others. It unifies various interesting special cases. We mention especially the triangular lattice, obtained by letting one of $\lambda_1$, $\lambda_3$ equal zero and {\em{Harper's model}}, in mathematics better known as the {\em{almost Mathieu operator}}, which arises when switching off NNN interactions, i.e. letting $\lambda_1 = \lambda_3 = 0$.

In this article, we provide a complete spectral analysis of extended Harper's model, valid for all values of $\lambda$ and a full measure set of irrational frequencies  (see Theorem \ref{thm_ehmspectral}, below), which, so far, has escaped rigorous mathematical treatment. Even in physics literature, despite extensive, mostly numerical studies of its spectral properties \cite{LiuGhoshChong_2015, ChangIkezawaKohmoto_1997, DreseHolthaus_1997, GongTong_2008,  HanThoulessHiramotoKohmoto_1994, HatsugaiKohmoto_1990,  KetojaSatija_1995, KetojaSatija_1995_2, Thouless_1983}, a fully analytical treatment  of extended Harper's model covering the full range of $\lambda$ has so far been missing. In view of Theorem \ref{thm_ehmspectral}, we mention however \cite{KetojaSatija_1997}, one of the very few heuristic treatments of extended Harper's model whose results indicate a difference in the spectral properties between isotropic ($\lambda_1 = \lambda_3$) and anisotropic ($\lambda_1 \neq \lambda_3$) NNN interactions. 

Our analysis relies on earlier work \cite{JitomirskayaMarx_2012, JitomirskayaMarx_2013_erratum}, in which a formula for the {\em{complexified Lyapunov exponent}} of extended Harper's model was proven, valid for all $\lambda$ and all irrational $\alpha$. In particular, underlying this formula is a partitioning of the parameter space into the following three regions
\begin{description}
\item[Region I] $0 \leq \lambda_{1}+\lambda_{3} \leq 1, ~0 < \lambda_{2} \leq 1$ ~\mbox{,}
\item[Region II] $0 \leq \lambda_{1}+\lambda_{3} \leq \lambda_{2}, ~1 \leq \lambda_{2} $ ~\mbox{,}
\item[Region III] $\max\{1,\lambda_{2}\} \leq \lambda_{1}+\lambda_{3}$, $\lambda_2>0$ ~\mbox{,}
\end{description}
which we illustrate pictorially in Fig. \ref{figure_1}. As shown in \cite{JitomirskayaMarx_2012}, this partitioning is a result of the duality transform for extended Harper's model, which for {\em{non-zero}} nearest neighbor coupling $\lambda_2$ is given by the following map acting on the space of coupling parameters\footnote{For completeness, the duality map for case $\lambda_2 = 0$ is given in (\ref{eq_dualityapp1}) and discussed in Appendix \ref{app_zeronn}.} is :
\begin{equation} \label{eq_sigma}
\sigma(\lambda):=\frac{1}{\lambda_{2}}(\lambda_{3},1,\lambda_{1}) \mbox{.}
\end{equation}
The precise action of the duality map is summarized in Observation \ref{obs_dualitymap}, to whose end, we define the line segments (see also Fig. \ref{figure_1}):
\begin{eqnarray}
\mathrm{L}_\mathrm{I} &:=& \{\lambda_1 + \lambda_3=1, 0 < \lambda_2 \leq 1\} \\
\mathrm{L}_{\mathrm{II}} &:=& \{0 \leq \lambda_1 + \lambda_3 \leq 1, \lambda_2 = 1\} \\
\mathrm{L}_{\mathrm{III}} &:=& \{1 \leq \lambda_1 + \lambda_3 = \lambda_2\}
\end{eqnarray}
One then easily verifies the following:
\begin{obs} \label{obs_dualitymap}
$\sigma$ is bijective on $\{\lambda_1 + \lambda_3 \geq 0, ~\lambda_2>0\}$ and one has:
\begin{itemize}
\item[(i)] $\sigma(\mathrm{I}^\circ) = \mathrm{II}^\circ$, $\sigma(\mathrm{III}^\circ) = \sigma(\mathrm{III}^\circ)$
\item[(ii)] $\sigma(\mathrm{L}_\mathrm{I}) = \mathrm{L}_{\mathrm{III}}$ and $\sigma(\mathrm{L}_{\mathrm{II}})= \mathrm{L}_{\mathrm{II}}$
\end{itemize}
\end{obs}
Observation \ref{obs_dualitymap} identifies the interior of the regions $I$ and $II$ as dual regions. Prior to this work, it had already been known that the Lyapunov exponent is positive in $I^\circ$ accompanied by Anderson localization for a.e. $\theta$ at all Diophantine $\alpha$ \cite{JitomirskayaKosloverSchulteis_2005}. Known duality-based arguments then allow to conclude purely absolutely continuous spectrum for a.e. $\theta$ and all Diophantine $\alpha$ in the dual regime $II^\circ$; see Theorem \ref{thm_dualregime} below. On the other hand the regime of couplings defined by
\begin{equation}
\mathcal{SD}:= \mathrm{III} \cup \mathrm{L}_{\mathrm{II}} ~\mbox{,}
\end{equation}
is characterized throughout by zero Lyapunov exponent \cite{JitomirskayaMarx_2012}, thus escaping traditional duality-based arguments. Since $\sigma$ bijectively maps $\mathcal{SD}$ onto itself, the literature refers to $\mathcal{SD}$ as the {\em{self-dual regime}}. To avoid confusion, we emphasize that the points in $\mathcal{SD}$ are not necessarily fixed points of $\sigma$; in fact, only points along  $\mathrm{L}_{\mathrm{II}}$ are fixed by $\sigma$. As mentioned earlier, the self-dual regime has so far posed the biggest challenge to both heuristic and rigorous treatments.

\begin{figure}[ht] \label{figure_1}
\includegraphics[width=0.5\textwidth]{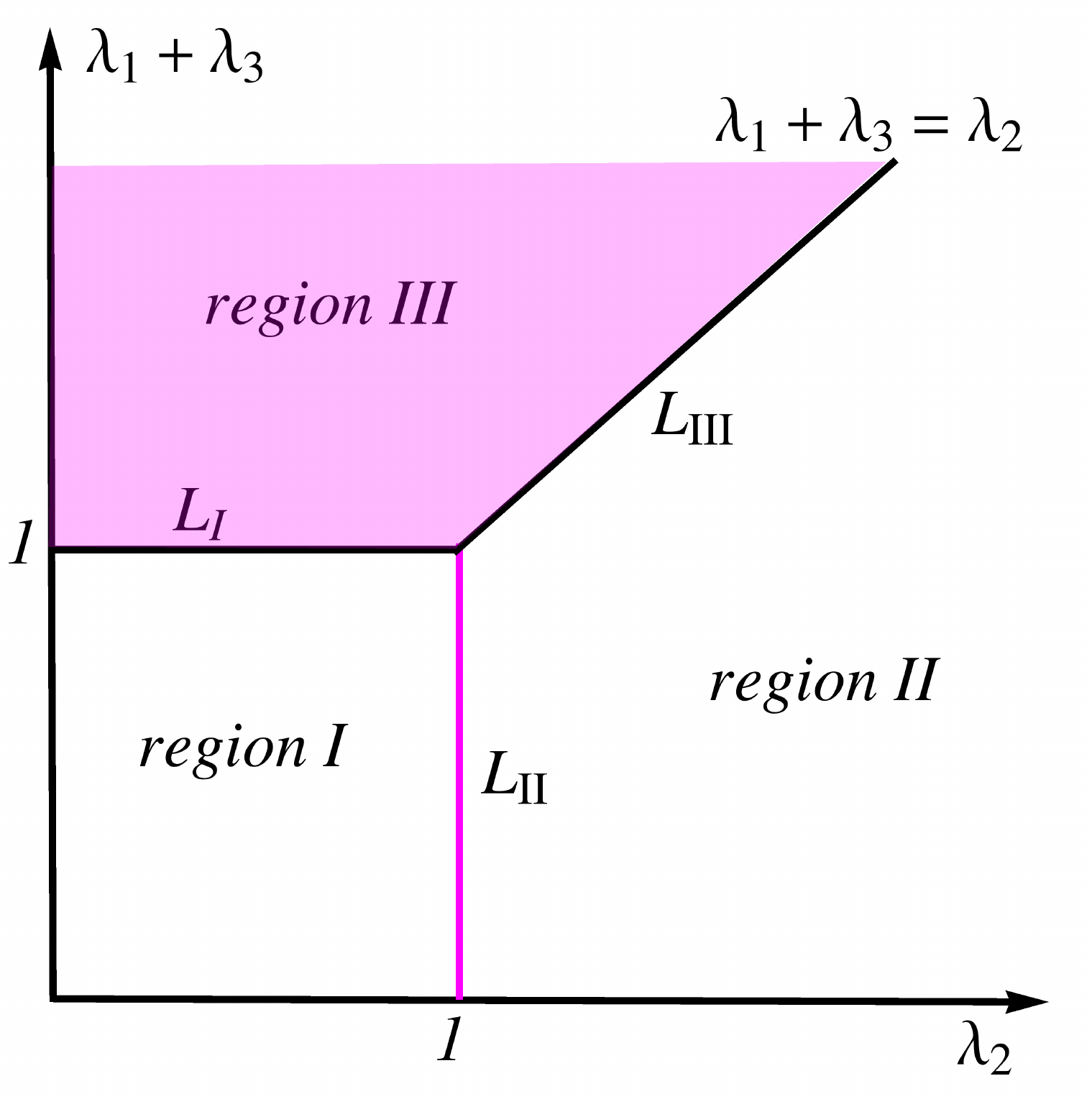} 
\caption{Partitioning of the space of coupling constants $\lambda = (\lambda_1,\lambda_2,\lambda_3)$ for extended Harper's model. The interesting self-dual regime is colored in red.}
\end{figure}

As will be explained, the missing link between \cite{JitomirskayaMarx_2012, JitomirskayaMarx_2013_erratum} and a complete understanding of the spectral properties of extended Harper's model is the following theorem which excludes eigenvalues in the self-dual regime; it constitutes the main result of this paper.
\begin{theorem} \label{thm_point} 
For all irrational $\alpha$ and all $\lambda \in \mathcal{SD}$, $H_{\theta; \lambda, \alpha}$ has empty point spectrum for $\mu$-a.e. $\theta$.
\end{theorem}

For each $\lambda$, the set of excluded phases $\theta$ can be described precisely through arithmetic conditions, see Sec. \ref{sec_proofofthmpoint} for details. Here, we only mention that for each $\lambda \in \mathcal{SD}$, one contribution to this excluded zero-measure set of phases is given by $\alpha$-rational $\theta$, defined as the following countable set:
\begin{definition} \label{def_alpharational}
$\theta \in \mathbb{T}$ is called {\em{$\alpha$-rational}} if $(\mathbb{Z} \alpha + 2 \theta) \cap \mathbb{Z} \neq \emptyset$ and {\em{non-$\alpha$-rational}}, otherwise.
\end{definition}

In particular, since the critical almost Mathieu operator arises from extended Harper's model by letting $\lambda_1=\lambda_3=0$ and $\lambda_2=1$ (therefore corresponding to $\lambda \in \mathrm{L}_{\mathrm{II}}$), we obtain the following important consequence of Theorem \ref{thm_point}:
\begin{theorem} \label{AM} 
For all irrational $\alpha$, the critical almost Mathieu operator has purely singular continuous spectrum for all non-$\alpha$-rational $\theta$.
\end{theorem}

\begin{remark} \label{rem_criticalAMO_sc}
\begin{enumerate}
\item For pedagogical reasons, we will prove Theorem \ref{thm_point} first for the special case of the critical almost Mathieu operator, which will imply Theorem \ref{AM} directly. This special case of Theorem \ref{thm_point} for the critical almost Mathieu operator is discussed in Sec. \ref{sec_selfdual_criticalamo}, Theorem \ref{thm_AMO} therein.
\item  While the spectrum of the critical almost Mathieu operator is known to have zero Lebesgue measure \cite{last,AvilaKrikorian_2006} (a fact actually not used in the present proof), the absence of eigenvalues (and thus purely singular continuous nature of the spectrum) has been a longstanding open question. Delyon \cite{del} proved that there are no eigenvectors belonging to $\ell^1$ and Chojnacki \cite{choi} established presence of some continuous spectrum for a.e. $\theta.$  A measure theoretic version of Theorem \ref{AM} was the main corollary of \cite{gjls}. However, the corresponding part of the argument in \cite{gjls} has a gap, thus Theorem \ref{AM} has been open, except for certain topologically generic but measure zero sets of $\alpha$ or $\theta$ where more general arguments apply \cite{as,js}. It should also be mentioned that other than for these measure zero sets,  the entire region III for the extended Harper's model has been completely open.

\item This paper incorporates two preprints, \cite{Avila_preprint_2008_2} and \cite{JM}, both of which were not intended for publication. In particular, Theorem \ref{AM} appeared in the preprint \cite{Avila_preprint_2008_2}, and the a.e. $\alpha$ version of Theorems \ref{thm_point} and \ref{thm_ehmspectral} appeared in  the preprint \cite{JM}. 

\item Theorem \ref{AM} excludes only a countable set of phases $\theta.$ The question whether the statement of Theorem \ref{AM} extends to {\em{all}} phases is one of the few open problems from the spectral theory of the almost Mathieu operator. Based on Sec. \ref{sec_Aubry}, this question relates to whether the exclusion of the $\alpha$-rational phases in Proposition \ref{prop_det} is really necessary. While we conjecture that this exclusion in Theorem \ref{AM} is not needed, in the general case of Theorem \ref{thm_point} some phases do lead to some point spectrum, see Proposition \ref{prop_ehmexcludephase}. 
\end{enumerate}
\end{remark}

The gap between Theorem \ref{thm_point} and a complete understanding of the spectral properties of extended Harper's model is bridged by the {\em{global theory of quasi-periodic, analytic Schr\"odinger operators}} developed in \cite{global} and partially extended to the Jacobi case in \cite{JitomirskayaMarx_2012, JitomirskayaMarx_2013_erratum}; subsequently, the global theory will be referred to as GT. The GT relies on an understanding of the complexified Lyapunov exponent, defined in (\ref{eq_defcomplexle}) of Sec. \ref{sec_avilasglobal}. To keep the paper as self-contained as possible, we will summarize some relevant aspects of the GT in Sec \ref{sec_avilasglobal}. For further details we refer the reader to the recent survey article on the dynamics and spectral theory of quasi-periodic Schr\"odinger-type operators in \cite{JitomirskayaMarx_ETDS_2016_review}. 

Based on the GT, Theorem \ref{thm_point} will be shown to imply the spectral resolution of extended Harper's model in the entire regime of zero Lyapunov exponents. 
The contents of Theorem \ref{thm_ehmspectral} are illustrated in Fig. \ref{figure_5}.
\begin{theorem} \label{thm_ehmspectral}
\begin{itemize}
\item[(i)] For all irrational $\alpha$, $\mu$-a.e. $\theta$, and $\lambda_1 \neq \lambda_3$, the spectrum is purely absolutely continuous in $II^\circ \cup III^\circ$ and purely singular continuous on the union of line segments $L_I \cup L_{II} \cup L_{III}$.
\item[(ii)] For all irrational $\alpha$, $\mu$-a.e. $\theta$, and $\lambda_1 = \lambda_3$, the spectrum is purely absolutely continuous in $II^\circ$ and purely singular continuous on $\mathcal{SD}$.
\end{itemize}
\begin{figure}[ht] 
\centering
\subfigure[$\lambda_1 \neq \lambda_3$]{
\includegraphics[width=0.4\textwidth]{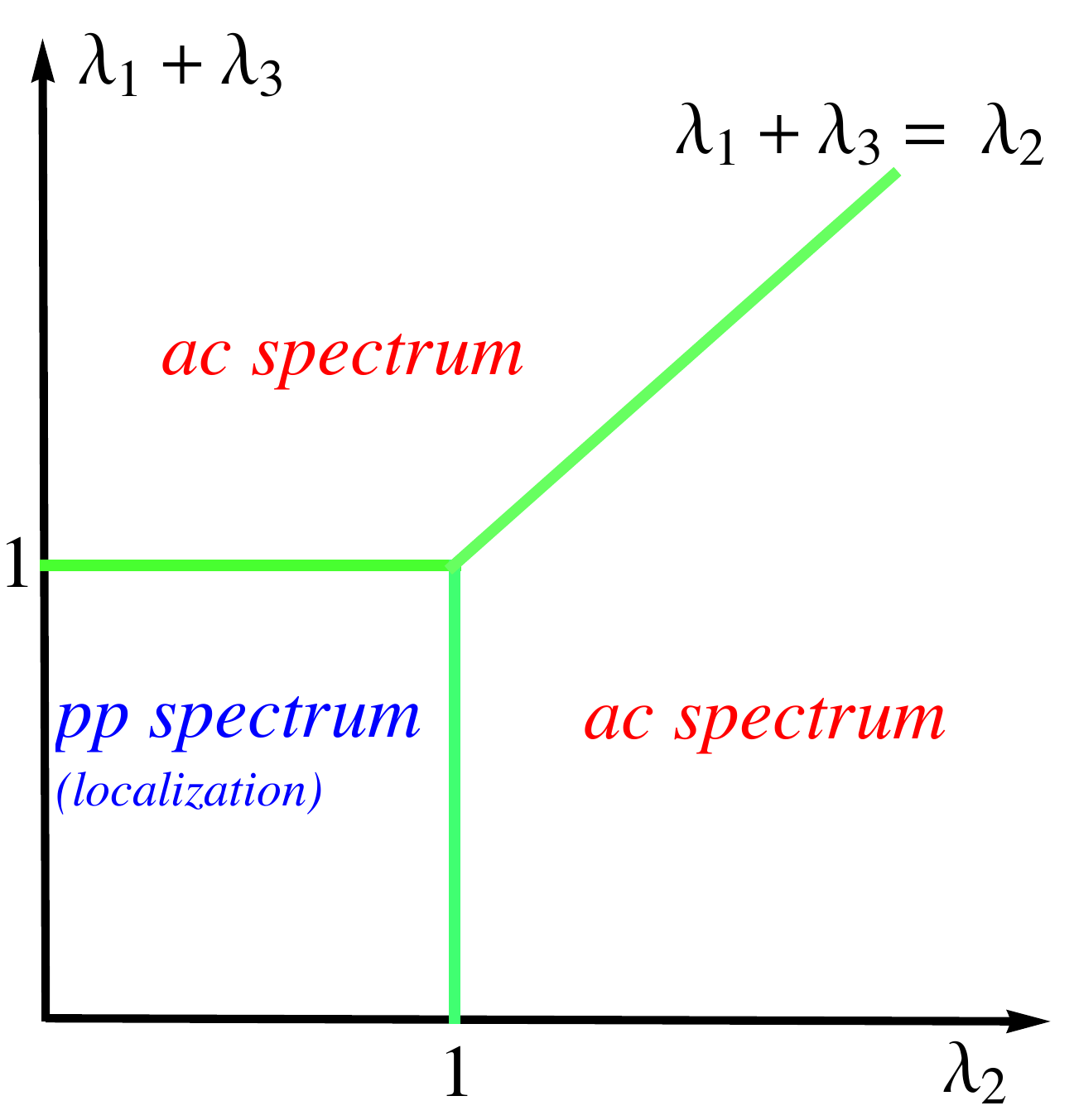} 
}
\subfigure[$\lambda_1 = \lambda_3$]{
\includegraphics[width=0.4\textwidth]{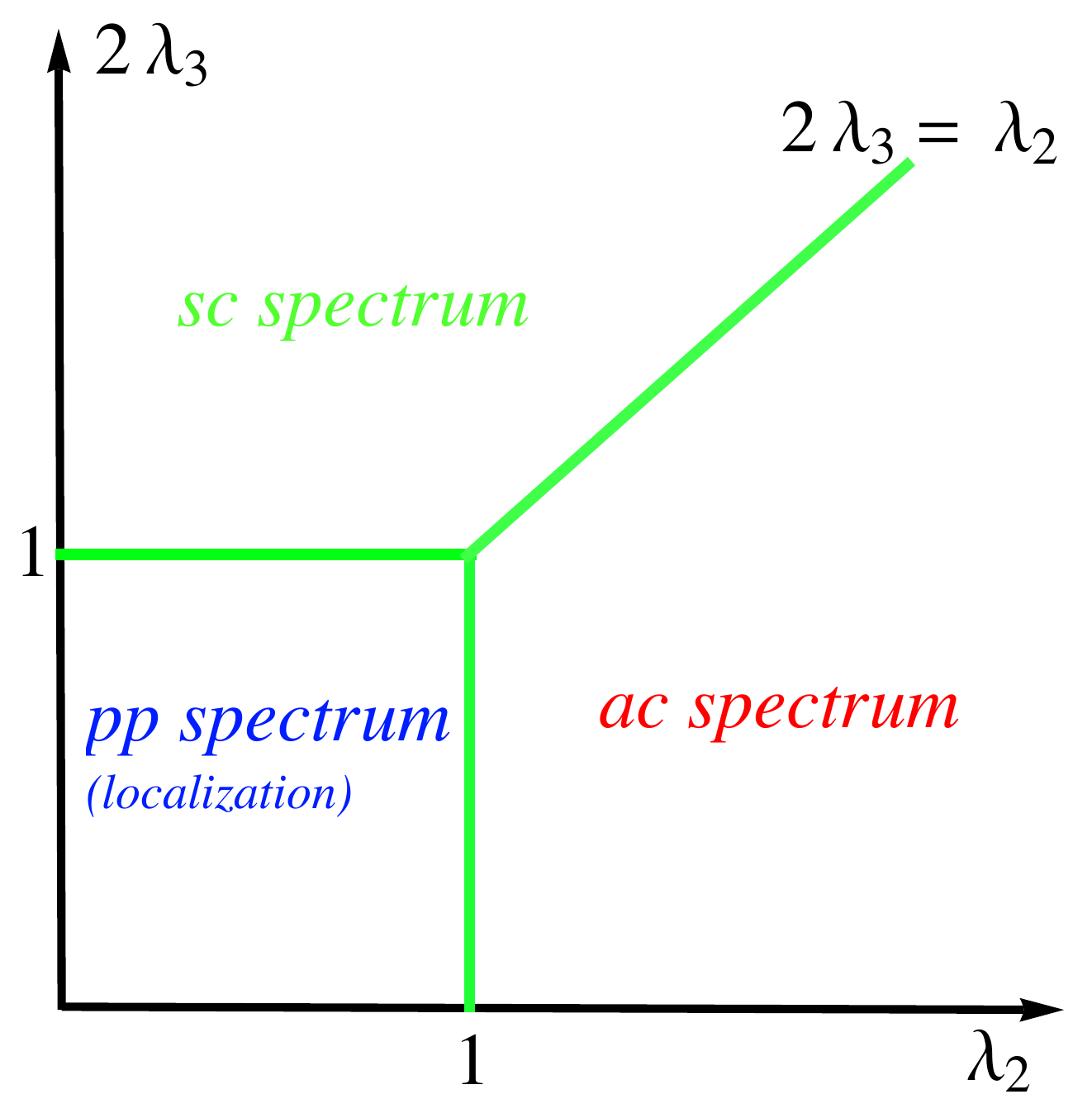} 
} \label{figure_5}
\caption{Spectral theory of extended Harper's model. Green indicates (purely) singular continuous spectrum. The spectral properties of extended Harper's model crucially depend on the symmetry of NNN interaction. Particularly noteworthy is the collapse in the self-dual regime, from purely absolutely continuous spectrum for $\lambda_1 \neq \lambda_3$ to purely singular continuous spectrum once $\lambda_1 = \lambda_3$. Anderson localization in region I had been proven before in \cite{JitomirskayaKosloverSchulteis_2005}.} \label{figure_5}
\end{figure}
\end{theorem}
\begin{remark}
We note, for completeness, that for $\lambda$ in the complementary region, $I^\circ,$ for all Diophantine $\alpha$ (defined in (\ref{eq_diophcond})) and $\mu$-a.e. $\theta$, the spectrum has been known to be purely point with exponentially decaying eigenfunctions \cite{JitomirskayaKosloverSchulteis_2005}. Thus Theorem \ref{thm_ehmspectral} completes spectral picture of the extended Harper's model for all couplings and a.e. $\alpha,\theta$. Moreover, it was recently shown that in $I^\circ$ there is a sharp arithmetic transition between pure point and singular continuous spectrum for $\mu$-a.e. $\theta$ at $\alpha$ with $\beta(\alpha)$ equal to the Lyapunov exponent, where $\beta (\alpha) =\limsup \frac{\ln q_{n+1}}{q_n}$ (see (\ref{21})) is the upper exponential growth rate of the continued fraction expansion of  $\alpha$ \cite{hj}. The spectral picture in the supercritical region $I^\circ$ remains unclear only for the tiny set of $\alpha,$ the ``second critical line'' where $\beta(\alpha)$ coincides with the Lyapunov exponent. We note that Theorem \ref{thm_ehmspectral} holds for all irrational $\alpha.$
\end{remark}
The most noteworthy conclusion of Theorem \ref{thm_ehmspectral} is that the symmetry of the NNN interaction triggers a collapse in the interior of region III from purely absolutely continuous (ac)  ({\em{anisotropic}} NNN interaction, $\lambda_1 \neq \lambda_3$) to purely singular continuous (sc) spectrum ({\em{isotropic}} NNN interaction, $\lambda_1 = \lambda_3$). Such {\em{spectral collapse}} has not yet been observed for any other known quasi-periodic operator.

Theorem \ref{thm_point} is also interesting from a more general point of view, which we formulate as the {\em{critical energy conjecture}} (CEC) in Conjecture \ref{conj_cec}; we comment more on the context of the CEC in Sec. \ref{sec_avilasglobal}, see in particular Remark \ref{rem_CEC}. In essence, the CEC claims that critical behavior in the sense of the GT is the signature of purely sc spectrum. Establishing the CEC in general would thus provide the long sought-after {\em{direct}} criterion for sc spectrum for quasi-periodic Jacobi operators with analytic coefficients. As detailed in Sec. \ref{sec_avilasglobal}, Theorem \ref{thm_point} verifies the CEC for the special case of extended Harper's model.

Even though some aspects of the proof of Theorem \ref{thm_point} rely on the specifics of extended Harper's model, we believe that the overall strategy should be extendable. Indeed, the method of this paper has already been implemented in establishing the CEC in another important model, a one-dimensional coined quantum walk with $n$-th coin defined by the rotation by the angle $\theta + n\alpha$, dubbed the unitary almost Mathieu operator, in \cite{foz} (which in particular directly uses the main number theoretical estimate of this paper, the solution of a conjecture of Erd\H{o}s-Szekeres, see below). This model appears in physics literature \cite{phys1} as the most natural next step from periodic quantum random walks studied in \cite{phys2}. 

A key ingredient in the proof of Theorem \ref{thm_point} is Theorem \ref{prop_prozero}, an estimate on the {\em{upper bound in the ergodic theorem}} for $\log \vert f(\theta) \vert$ under irrational rotations for complex analytic $f: \mathbb{T} \to \mathbb{C}$. Presence of zeros in the  function $f$ complicates the matter quite substantially. Indeed, the main accomplishment here is to obtain an upper bound without imposing restrictions on the arithmetic properties of the rotational frequency. Aside from its important role for the present paper, we also expect Theorem \ref{prop_prozero} to become crucial when establishing the critical energy conjecture for general quasi-periodic operators\footnote{Indeed, it has already played a crucial role in the above mentioned proof of absence of point spectrum in \cite{foz}.}, and to be of interest in its own right. Theorem \ref{prop_prozero} is proven in Sec. \ref{sec_rate}.  It essentially boils down to answering a question of Erd\H{o}s and Szekeres \cite{es} about certain trigonometric products. The interest in questions of this type has been renewed lately, see e.g. \cite{bc} where some other problems posed in \cite{es} were addressed/answered, but the one which plays a role in our analysis had remained open.  Its solution (Theorem \ref{thm_rate_main}) is the main content of Section \ref{sec_rate}.

The rest of the paper is structured as follows. Sec. \ref{sec_avilasglobal}-\ref{sec_avilasglobal_applicehm} embed Theorem \ref{thm_ehmspectral} into the context of the global theory for quasi-periodic, analytic Jacobi operators. In particular, the spectral consequences of the GT will reduce Theorem \ref{thm_ehmspectral} to our main result, Theorem \ref{thm_point}. The point is, that, {\em{while critical behavior in the sense of the GT already implies singular (sc+pp) spectrum, it does not a priori exclude eigenvalues.}}

Theorem \ref{thm_point} is proven by contradiction in Sec. \ref{sec_Aubry} - \ref{sec_selfdual}. To illustrate the general idea, we start with the special case of the critical almost Mathieu operator (Theorem \ref{thm_AMO}) and prove absence of eigenvalues for all non-$\alpha$-rational, i.e., for all but countably many phases. For all such phases, the latter implies purely sc spectrum (Theorem \ref{AM}). The more complicated form of extended Harper's model, as well as the presence of zeros in $c_\lambda(\theta)$, however, leads to non-trivial changes in the argument, in particular, requiring the results of Sec. \ref{sec_rate}.

We note that even though the original argument for the critical almost Mathieu operator already excludes countably many phases, it is (still) not clear whether this exclusion is indeed necessary. For extended Harper's model, however, it is shown in Proposition \ref{prop_ehmexcludephase} that the zeros in $c_\lambda(\theta)$ {\em{necessitate}} the exclusion of countably many phases in Theorem \ref{thm_point}. It is interesting, though, that for extended Harper's model with isotropic NNN ($\lambda_1 = \lambda_3$) an additional zero measure set of phases has to be excluded in our proof. Origin of this additional zero measure set is a general fact on almost uniqueness of rational approximation, which we prove in Sec. \ref{sec_rationalapprox}. The authors note that it has meanwhile been shown by R. Han in \cite{RHan_IMRN_2017} that exclusion of this additional zero measure set of phases is indeed an artefact of our proof which can be avoided using the simplifications done in \cite{RHan_IMRN_2017}.

The remaining two sections, Sec. \ref{sec_afk} and \ref{sec_almredimpliesac}, establish some ingredients needed for the spectral consequences of the GT, which are currently only available for Schr\"odinger but not for Jacobi operators.

Sec. \ref{sec_afk} is devoted to the proof of Theorem \ref{m5_thm_afk_jacobi}, an extension of the spectral dichotomy expressed in \cite{AvilaFayadKrikorian_2011} to non-singular Jacobi operators: for Lebesgue a.e. $E \in \mathbb{R}$, the Lyapunov exponent of the Jacobi operator is either strictly positive or (the analytically normalized Jacobi cocycle associated with) $E$ is analytically reducible to rotations (in the sense specified in Theorem \ref{m5_thm_afk_jacobi}). In final consequence, Theorem \ref{m5_thm_afk_jacobi} implies that the set of critical energies in the sense of the GT can only support singular (sc+pp) spectrum. Since the main result of \cite{AvilaFayadKrikorian_2011} is not specific to Schr\"odinger operators, Theorem \ref{m5_thm_afk_jacobi} essentially boils down to proving $L^2$-reducibility of the (normalized) Jacobi cocycle. For Schr\"odinger operators, the latter is a well known fact going back to \cite{DeiftSimon_1983}. 

Finally, Sec. \ref{sec_almredimpliesac} shows that almost reducibility implies purely absolutely continuous spectrum for $\mu$-a.e. phase (Theorem \ref{thm_arimpliesac}), which is necessary to draw the spectral theoretic conclusions about the set of subcritical energies. We mention that the proof we present here slightly shortens the argument given for Schr\"odinger operators in \cite{Avila_prep_ARC_1}.

{\bf{Acknowledgement:}} We are grateful to J. Bourgain for pointing out that Theorem \ref{thm_rate_main} in the previous version solves a conjecture from \cite{es} and to M.-C. Chang for sharing with us \cite{es}.

\section{Upper bound for analytic, quasi-periodic products: solution of a problem by Erd\H{o}s and Szekeres} \label{sec_rate}

Given $\alpha \in [0,1)$ irrational, denote by $\frac{p_n}{q_n}$ the $n$th approximant associated with the continued fraction expansion of $\alpha = [0; a_1, a_2, \dots]$, in particular, 
\begin{equation}\label{21}
q_n = a_n q_{n-1} + q_{n-2} ~\mbox{, $n \geq 2$ .} 
\end{equation}
Here, we use the conventions $q_0 = 1$ and $q_{-1} = 0$. 

Following, for $r \in \mathbb{R}$, we set 
\begin{equation}
\vertiii{r}:=\inf_{n \in \mathbb{Z}} \vert r - n \vert ~\mbox{,}
\end{equation}
which induces the usual norm on $\mathbb{T}$. Letting $\Delta_n:= \vert q_n \alpha - p_n \vert$, we recall the basic estimates
\begin{eqnarray} \label{eq_contifracbasic}
\frac{1}{q_n + q_{n+1}} < \Delta_n < \frac{1}{q_{n+1}} ~\mbox{,} \nonumber \\
\vertiii{k \alpha} > \Delta_{n-1} ~\mbox{, if } q_{n-1} + 1 \leq k \leq q_n -1 ~\mbox{.}
\end{eqnarray}

The following  question was asked in a paper by Erd\H{o}s and Szekeres \cite{es}: whether it is true that for all irrational $\alpha$, one has 
\begin{equation}\label{es1}
\liminf_{n \to \infty} \max_{|z|=1}\prod_{k=1}^n|z-e^{2\pi i k\alpha}|<\infty
\end{equation}

It was pointed out in \cite{es} that (\ref{es1}) holds for a.e. $\alpha$ with moreover a subsequence along which the limit is equal to $2$. 

Erd\H{o}s and Szekeres posed several conjectures in \cite{es}, and while there has been a number of partial results on some of those, in particular, on the one on pure product polynomials, e.g. \cite{bc,bdm},  we are not aware of further results towards (\ref{es1}).

Above-mentioned question by Erd\H{o}s and Szekeres will be important for studying {\em{quasi-periodic products}} of the form
\begin{equation} \label{eq_qpproducs}
 \frac{\left(\prod_{j=0}^{q_{n_k}} f(x+ j \alpha)\right)
}{\mathrm{exp}\left( q_{n_k}\int_{\mathbb{T}} \log \vert f(x) \vert \ud \mu(x) \right)} ~\mbox{,} \end{equation} 
for $f$  analytic in a neighborhood of $\mathbb{T}$. Here, the goal will be to obtain a subsequence $(q_{n_k})$ which allows for a {\em{uniform}} upper bound of order $\mathrm{exp} (\mathcal{O}(1/q_{n_k})).$ The main challenge in this endeavor is to allow for zeros of the function $f$ without imposing additional number theoretic conditions on $\alpha$. This section is devoted to the proof of the above conjecture by Erd\H{o}s and Szekeres and some related questions/corollaries.

First, we denote by
\begin{equation}
S(q_n,z) := \sum_{k=0}^{q_n - 1} \log \vert \mathrm{e}^{2 \pi i k \alpha} z - 1 \vert  ~\mbox{,}
\end{equation}
where, here and following, $z \in \mathbb{C}$ is assumed to satisfy $\vert z \vert =1$.  The conjecture  (\ref{es1}) is then established as a consequence of:
\begin{theorem} \label{thm_rate_main}
For each irrational $\alpha$, there exists $C>0$ such that
\begin{equation}\label{bound}
\liminf_{n \to \infty} \sup_{\vert z \vert =1} S(q_n,z) \leq C ~\mbox{.}
\end{equation}
\end{theorem}

As will follow from the proof below (which also had already been pointed out by Erd\H{o}s and Szekeres), for {\em{certain}} $\alpha$ one in fact has that $\sup_{\vert z \vert =1} S(q_n,z) \leq C$ for all $n \in \mathbb{N}$. Moreover, for {\em{all}} $\alpha$ and $n \in \mathbb{N}$, one has the bound $\sup_{\vert z \vert =1} S(q_n,z) \leq C\log q_n$ (e.g. \cite{AvilaJitomirskaya_2009}). In general however, the $\liminf$ in Theorem \ref{thm_rate_main} is indeed necessary, which is the subject of the following:
\begin{theorem} \label{counter}
There exist $\alpha$ such that \begin{equation}\label{sup}\limsup_{n \to \infty}\sup_{\vert z \vert =1} S(q_n,z)/\log q_n\geq 1.\end{equation}
\end{theorem}

As an immediate corollary of Theorem \ref{thm_rate_main} we obtain our main result about the rate of convergence of the quasi-periodic products in (\ref{eq_qpproducs}):
\begin{theorem} \label{prop_prozero}
Let $f$ be analytic in a neighborhood of $\mathbb{T}$ and $\alpha$ a fixed irrational. There exists $C>0$ and a subsequence $(q_{n_l})$ of $(q_n)$ such that {\em{uniformly}} in $x \in \mathbb{T}$:
\begin{equation} \label{eq_ergodicthmrate}
\frac{1}{q_{n_l}} \sum_{j=0}^{q_{n_l}-1} \log \abs{f(x+j \alpha)} - \int \log\abs{f} \ud \mu \leq \frac{C}{q_{n_l}} ~\mbox{.}
\end{equation}
\end{theorem}
In the context of extended Harper's model, Theorem \ref{prop_prozero} will later serve as a crucial ingredient in the proof of Theorem \ref{thm_point}.
\begin{remark}
It follows from Lemma \ref{lem_roc1} below that (\ref{eq_ergodicthmrate}) holds along the {\em{full}} sequence $(q_n)$ if $f(x)$ has no zeros on $\mathbb{T}$. The achievement of Theorem \ref{prop_prozero} is to account for possible zeros of $f$. It is shown in Theorem \ref{counter} that presence of zeros in general necessitates passing to a subsequence, which implies that Theorem \ref{prop_prozero} as stated is optimal.
\end{remark}

{\it Proof of Theorem \ref{thm_rate_main}.} For $q_n \geq q_m$, $l \geq 0$, $0 \leq r \leq q_m$, and $|z|=1$, we introduce
\be
C(q_n,q_m,l,r,z):=\left | \sum_{k \in J_{q_n,q_m,z}, l \leq k \leq l+r-1}
\log |e^{2 \pi i k \alpha} z-1|-\log |e^{2 \pi i k p_n/q_n} z-1| \right |,
\ee
where
\be \label{eq_rate_defn}
J_{q_n,q_m,z}:=\{k \in \Z,\, |e^{2 \pi i k
p_n/q_n} z-1| \geq \frac {10} {q_m}\}.
\ee
We set $C(q_n,q_m,z):=C(q_n,q_m,0,q_m,z)$ and
\be
C(q_n,q_m):=\sup_{|z|=1} C(q_n,q_m,z).
\ee

\begin{lemma} \label{rate_lemma1}
We have,
\begin{equation} \label{eq_rate_lemma1_claim}
S(q_n,z) \leq C(q_n,q_n)+C_1 ~\mbox{.}
\end{equation}
\end{lemma}

\begin{proof}
Write $S=S(q_n,z)$.  We will use that
\be
\tilde S=\sum_{k=0}^{q_n-1} \log |e^{2 \pi i k p_n/q_n} z-1|=\log |z^{q_n}-1|.
\ee
Let
\be
S_0=\sum_{k \in J_{q_n,q_n,z}, 0 \leq k \leq q_n-1}
\log |e^{2 \pi i k \alpha} z-1|, \quad
S_1=\sum_{k \notin J_{q_n,q_n,z}, 0 \leq k \leq q_n-1}
\log |e^{2 \pi i k \alpha} z-1|,
\ee
\be
\tilde S_0=\sum_{k \in J_{q_n,q_n,z}, 0 \leq k \leq q_n-1}
\log |e^{2 \pi i k p_n/q_n} z-1|, \quad
\tilde S_1=\sum_{k \notin J_{q_n,q_n,z}, 0 \leq k \leq q_n-1}
\log |e^{2 \pi i k p_n/q_n} z-1|,
\ee
so that $S=S_0+S_1$ and $\tilde S=\tilde S_0+\tilde S_1$. Letting 
\begin{equation}
s:=\# ~\mbox{of terms in }S_1=\# ~\mbox{of terms in } \tilde S_1 ~\mbox{,}
\end{equation}
obviously yields
\begin{equation}
S_1, \tilde S_1 \leq -s \log q_n+C ~\mbox{.}
\end{equation}
We distinguish between the following two cases:

First, assume that $|z^{q_n}-1| \geq \frac {1} {10}$.  Then, one has $| \tilde S_1+s \log q_n| \leq C$ and it follows that $S_1 \leq \tilde S_1+C$, so that $S \leq \tilde S+S_0-\tilde S_0+C$. Consequently, we obtain
\begin{equation}
S \leq S_0-\tilde S_0+C \leq C(q_n,q_n)+C ~\mbox{,}
\end{equation}
which is the claim of (\ref{eq_rate_lemma1_claim}) for $|z^{q_n}-1| \geq \frac {1} {10}$.

If, on the other hand, one has that $|z^{q_n}-1|<\frac {1} {10}$, then there exists a unique $0 \leq k_0 \leq q_n-1$ such that $z_*=e^{2 \pi i k_0 p_n/q_n} z$ is closest to $1$. Letting $z = \mathrm{e}^{2 \pi i \theta}$, definition of $z_*$ in particular entails
\begin{equation} \label{eq_rate_lemma1_0}
\vertiii{\theta + k_0 p_n/q_n} \leq \frac{1}{2 q_n} ~\mbox{, } \vertiii{ \theta + k p_n/q_n} \geq \frac{1}{2 q_n} ~\mbox{, } k \neq k_0 ~\mbox{.}
\end{equation}
From (\ref{eq_rate_lemma1_0}), we therefore conclude $|\tilde S_1+ (s-1) \log q_n-\log |z_*-1|| \leq C$, whence 
\begin{equation} \label{eq_rate_lemma1_1}
S_1 \leq \tilde S_1- \log q_n-\log |z_*-1|+C ~\mbox{.}
\end{equation}

Since $z^{q_n}=z_*^{q_n}$ and
\begin{equation} \label{eq_rate_lemma1_2}
\dfrac{\vert z_* - 1 \vert}{\vert z_*^{q_n} - 1 \vert} = \dfrac{1}{\vert \sum_{k=0}^{q_n - 1} z_*^{q_n} \vert} \geq \dfrac{1}{q_n} ~\mbox{,}
\end{equation}
we conclude from (\ref{eq_rate_lemma1_1}) that
\begin{equation}
S_1 \leq \tilde S_1-\log |z^{q_n}-1|+C ~\mbox{.}
\end{equation}
In particular, we have
\be
S \leq \tilde S+S_0-\tilde S_0-\log|z^{q_n}-1|+C=S_0-\tilde S_0+C \leq C(q_n,q_n)+C ~\mbox{,}
\ee
which establishes the claim of (\ref{eq_rate_lemma1_claim}) for the remaining case that $|z^{q_n}-1|<\frac {1} {10}$.
\end{proof}

Lemma \ref{rate_lemma1} reduces the proof of Theorem \ref{thm_rate_main} to analyzing the error caused by {\em{rational approximation}} of $\alpha$, the latter of which  is expressed by $C(q_n,q_n)$. Specifically, we claim:
\begin{lemma}
There exists $C_2>0$ (independent of $\alpha$) such that
\be
C(q_n,q_m) \leq C_2 |\alpha-\frac {p_n} {q_n}| \Lambda(q_m) ~\mbox{,}
\ee
where $\Lambda(q_m)=q_m \sum_{k=1}^m q_k \log \frac {2 q_k} {q_{k-1}}$.
\end{lemma}

\begin{proof}
Let $z=\mathrm{e}^{2 \pi i \theta}$, for $\theta \in [0,1)$. We first consider a trivial estimate for $C(q_n,q_m,0,r,z)$: For $k \in J_{q_n,q_m,z}$ with $0 \leq k \leq r-1$, simply write 
\begin{eqnarray} 
\vert \log |e^{2 \pi i k \alpha} z-1|-\log |e^{2 \pi i k p_n/q_n} z-1| \vert =: \left \vert \log \left \vert 1 + \eta \right \vert \right \vert  ~\mbox{.}
\end{eqnarray}
Observe that the lower bound in (\ref{eq_rate_defn}) combined with (\ref{eq_contifracbasic}) implies that $\vert \eta \vert \leq \pi/5$, which in turn yields:
\begin{equation} \label{eq_rate_lem2_0}
\vert \log |e^{2 \pi i k \alpha} z-1|-\log |e^{2 \pi i k p_n/q_n} z-1| \vert \leq \dfrac{ C k \vert \alpha - p_n/q_n \vert}{\vert z \mathrm{e}^{2 \pi i k p_n/q_n} - 1 \vert} ~\mbox{.}
\end{equation}

The denominator on the right hand side of (\ref{eq_rate_lem2_0}) is controlled by approximation by $q_m$-th roots of unity. To this end take $b \in \{0, \dots, q_{m-1}\}$ such that $\vert \vert \vert b \frac{p_m}{q_m} - \theta \vert \vert \vert$ is at {\em{minimum}}. In particular, for $k \in J_{q_n,q_m,z} \cap \{0,...,r-1\}$, the points $e^{2 \pi i (b+k) p_m/q_m}$ are distinct $q_m$-th roots of unity which are different from $1$ and thereby satisfy
\begin{equation} \label{eq_rate_lem2_1}
0 < C \leq \frac {|e^{2 \pi i k p_n/q_n} z-1|} {|e^{2 \pi i (b+k) p_m/q_m}-1|} ~\mbox{.}
\end{equation}
To verify (\ref{eq_rate_lem2_1}) notice that
\begin{equation}
\vertiii{ (\theta + k \frac{p_n}{q_n}) - (k+b) \frac{p_m}{q_m} } \leq \vertiii{\theta - b \frac{p_m}{q_m} } + \vertiii{k \frac{p_n}{q_n} - k \frac{p_m}{q_m} } < \frac{3}{q_m} ~\mbox{,}
\end{equation}
whence, taking\footnote{As common, for $x \in \mathbb{R}$, $\{ x \}:=x - \lfloor x \rfloor$ denote its fractional part.} $1 \leq l \leq q_m -1$ such that $\frac{l}{q_m} \leq \{ \theta + k \frac{p_n}{q_n} \} < \frac{l+1}{q_m}$, we arrive at (\ref{eq_rate_lem2_1}) since
\begin{equation}
\dfrac{ \vert \mathrm{e}^{2 \pi i k \frac{p_n}{q_n} } z - 1 \vert }{ \vert \mathrm{e}^{2 \pi i (k+b) \frac{p_m}{q_m} }- 1 \vert  } \geq C \dfrac{l}{l+3} \geq \frac{C}{4} ~\mbox{.}
\end{equation}
In consequence of (\ref{eq_rate_lem2_1}), we thus conclude
\be \label{rec}
C(q_n,q_m,0,r,z) \leq C (r-1) |\alpha-\frac {p_n} {q_n}| q_m \log(2 q_m) ~\mbox{.}
\ee

To improve this estimate, we reason as follows: take $q_t<q_m$ and introduce $s:=\lfloor \frac {r} {q_t} \rfloor$, $\tilde r:=r-s q_t$, $l_j:=\tilde r+j q_t$.
Then, one estimates:
\be \label{eq_rate_lem2_2a}
C(q_n,q_m,0,r,z) \leq C(q_n,q_t,0,\tilde r,z)+\sum_{j=0}^{s-1}
C(q_n,q_t,l_j,q_t,z)+\Delta ~,
\ee
where
\be
\Delta=\left |\sum_{k \in J_{q_n,q_m,z} \setminus J_{q_n,q_t,z}, 0
\leq k \leq r-1} \log |e^{2 \pi i k \alpha}-1|-\log |e^{2 \pi i p_n/q_n}-1|
\right |.
\ee

Since one has $\vertiii{ k (p_n/q_n) - k (p_t/q_t) } < \frac{2}{q_t}$ for all $1 \leq k \leq q_t < q_m$, there are at most 
\begin{equation} \label{eq_rate_lem2_3}
C (s+1)
\end{equation}
elements $k \in (J_{q_n,q_m,z} \setminus J_{q_n,q_t,z}) \cap \{0,...,r-1\}$.  Approximation by the $q_m$-th roots of unity as before thus yields
\be \label{eq_rate_lem2_3a}
\Delta \leq C (r-1) |\alpha-\frac {p_n} {q_n}| q_m \log \frac {2 q_m} {q_t}.
\ee

We now turn to the right hand side of (\ref{eq_rate_lem2_2a}) with the goal of estimating $C(q_n,q_t,l_j,q_t,z)$.  To this end, first bound $C(q_n,q_t,l_j,q_t,z)$ by a sum of the 
two terms
\be
(\mathrm{I}) := \left |\sum_{k \in J}
\log |e^{2 \pi i k \alpha} e^{2 \pi i l_j \alpha} z-1|-
\log |e^{2 \pi i k \alpha} e^{2 \pi i l_j p_n/q_n} z-1| \right |
\ee
and
\be
(\mathrm{II}):=\left |\sum_{k \in J}
\log |e^{2 \pi i k \alpha} e^{2 \pi i l_j p_n/q_n} z-1|-
\log |e^{2 \pi i k p_n/q_n} e^{2 \pi i l_j p_n/q_n} z-1| \right |,
\ee
where $J$ is the set of all $k \in \{0,...,q_t-1\}$ such that $l_j+k \in
J_{q_n,q_t,z}$.
Observe that (II) is of the form $C(q_n,q_t,z_j)$ where $z_j:=e^{2 \pi i l_j
p_n/q_n} z$.  

For (I), we claim the bound
\begin{equation} \label{eq_rate_lem2_4}
 (\mathrm{I}) \leq C q_t l_j |\alpha-\frac {p_n} {q_n}| ~\mbox{.}
\end{equation}
This upper bound is obtained by noting that for $k \in J$ the expression
\be
\log |e^{2 \pi i k \alpha} e^{2 \pi i l_j \alpha} z-1|-
\log |e^{2 \pi i k \alpha} e^{2 \pi i l_j p_n/q_n} z-1|
\ee
equals $l_j (\alpha-\frac {p_n} {q_n})$ times the derivative of the map
\be
\Phi: (0,1) \to \mathbb{R} ~\mbox{, } x \mapsto \log |e^{2 \pi i x}  -1|
\ee
at some $\theta_k \in (0,1)$ between
$\{ \theta + k \alpha + l_j p_n/q_n \}$ and $\{ \theta + k \alpha +  l_j \alpha\}$. The bound in (\ref{eq_rate_lem2_4}) is thus reduced to show that:
\begin{claim}
\be \label{eq_rate_lem2_5}
\left | \sum_{k \in J} D\Phi(\theta_k) \right | \leq C q_t.
\ee
\end{claim}
\begin{proof}
The proof of (\ref{eq_rate_lem2_5}) will crucially depend on the observation that $\Phi$ is strictly convex on $(0,1)$ and satisfies $D \Phi(x) = - D \Phi(1 - x)$. To this end, let $\tilde \theta \in [0,1)$ be such that
$e^{2 \pi i \tilde \theta}= e^{2 \pi i (l_j \alpha + \theta)}$ and take $b$ such that  $\mathrm{e}^{2 \pi i b (p_t/q_t)}$ is the $q_t$-th root of unity closest to $\mathrm{e}^{2 \pi i \tilde{\theta}}$. Set $\mathrm{e}^{2 \pi i ( b + k ) (p_t/q_t)} =: \mathrm{e}^{2 \pi i \theta_k^\prime}$ and let $\theta^\pm_k=\theta'_k\pm\frac {4} {q_t}$. Then for $k \in J$, one has $\theta_k \in (\theta^-_k,\theta^+_k)$. 

By convexity of $\Phi$, we estimate $D\Phi(\theta^-_k)< D \Phi(\theta_k) < D \Phi(\theta^+_k)$.  Notice that for $k \in J$, $\theta^-_k$ are all
distinct and form a subset $\Theta^-$ of $\Theta:=\{\frac {j} {q_t},\, 1 \leq j \leq q_t-1\}$.  Moreover, since $\# J \geq q_t-C$ with $C$ as in (\ref{eq_rate_lem2_3}), $\Theta
\setminus \Theta^-$ has at most $C$ elements.  Since
$D\Phi(x)=-D\Phi(1-x)$, we have
$\sum_{\theta \in \Theta} D\Phi(\theta)=0$, so that
$\sum_{\theta \in \Theta^-} D\Phi(\theta)=\sum_{\theta \in \Theta \setminus
\Theta^-} D \Phi(\theta)$.  It follows that $\sum_{k \in J} D \Phi(\theta_k)
\geq \sum_{\theta \in \Theta \setminus \Theta^-} D \Phi(\theta) \geq -C
q_t$.  A similar argument involving $\{\theta_k^+, 1 \leq k \leq q_t -1\}$ yields $\sum_{k \in J} D \Phi(\theta_k) \leq C
q_t$, which in summary verifies the claim of (\ref{eq_rate_lem2_5}).
\end{proof}

In summary, we can so far conclude that
\begin{eqnarray}
C(q_n,q_m,0,r,z) \leq &
C q_m^2 |\alpha-\frac {p_n} {q_n}| \log \frac {2 q_m} {q_t}+C q_t
\sum_{j=0}^{s-1} l_j |\alpha-\frac {p_n} {q_n}| +  \nonumber \\
& C(q_n,q_t,0,\tilde r,z)+\sum_{j=0}^{s-1} C(q_n,q_t,z_j) ~\mbox{.}
\end{eqnarray}
Taking into account that $s \leq \frac {q_m} {q_t}$ and $l_j \leq q_m$,
we get
\be \label{eq_rate_lem2_6}
C(q_n,q_m,0,r,z) \leq C q_m^2 |\alpha-\frac {p_n} {q_n}| \log \frac {2 q_m}
{q_t}+C(q_n,q_t,0,\tilde r,z)+\sum_{j=0}^{s-1} C(q_n,q_t,z_j).
\ee

To estimate further, we specify $r=q_m$ and $q_t=q_{m-1}$.  Then, $\tilde r=q_{m-2}$ and if we use (\ref{eq_rate_lem2_2a}) and (\ref{eq_rate_lem2_3a}), we obtain
\begin{align}
C(q_n,q_t,0,\tilde r,z) &\leq C(q_n,q_{m-2},z)+\left | \sum_{ \substack{  k \in
J_{q_n,q_{m-1},z} \setminus J_{q_n,q_{m-2},z} \\ 0 \leq k \leq q_{m-2}-1} } \log
|e^{2 \pi i k \alpha}-1|-\log |e^{2 \pi i k p_n/q_n}-1| \right |\\
\nonumber
&\leq
C(q_n,q_{m-2})+C (q_{m-2}-1)
|\alpha-\frac {p_n} {q_n}| q_{m-1} \log \frac {2 q_{m-1}} {q_{m-2}} ~,
\end{align}
so that (\ref{eq_rate_lem2_6}) becomes
\begin{eqnarray}
C(q_n,q_m) \leq & C_1^\prime q_m^2 |\alpha-\frac {p_n} {q_n}| \log \frac {2 q_m}
{q_{m-1}}+C_1^\prime q_{m-2} q_{m-1} |\alpha-\frac {p_n} {q_n}| \log \frac {2 q_{m-1}}
{q_{m-2}}+ \nonumber \\
& \lfloor \frac {q_m} {q_{m-1}} \rfloor C(q_n,q_{m-1})+C(q_n,q_{m-2}) ~\mbox{.}
\end{eqnarray}

Since $q_0=1$, we have $C(q_n,q_0)=0$.  Moreover, $\Lambda(q_1)=q_1^2 \log 2 q_1$, whence by (\ref {rec}), we see that $C(q_n,q_1) \leq C_2^\prime |\alpha-\frac {p_n} {q_n}| \Lambda(q_1)$.  

Finally, observe that
\begin{equation}
q_{m-2} q_{m-1} \log \left( \dfrac{2 q_{m-1}}{q_{m-2}} \right)= q_{m-1}^2 \dfrac{q_{m-2}}{q_{m-1}}  \log \left( \dfrac{2 q_{m-1}}{q_{m-2}} \right) \leq q_m^2 \log \left( \dfrac{2 q_{m}}{q_{m-1}}  \right) ~\mbox{,}
\end{equation}
and 
\begin{equation}
\Lambda(q_{m-2}) \leq \frac{q_{m-2}}{q_{m-1}} \Lambda(q_{m-1}) = \{ \frac{q_m}{q_{m-1}} \} \Lambda(q_{m-1}) ~\mbox{.}
\end{equation}
Thus, taking $C_2:=\max \{C_1^\prime,C_2^\prime\}$, induction in $m$ shows that
\begin{equation}
C(q_n,q_m) \leq 2 C_2 |\alpha-\frac {p_n} {q_n}| \Lambda(q_m) ~\mbox{,}
\end{equation}
for $0 \leq m \leq n$, thereby completing our proof.
\end{proof}

\begin{lemma} \label{rate_lemma3}
$\liminf_{n \to \infty} C(q_n,q_n) \leq C_3$.
\end{lemma}
\begin{proof}
Employing Lemma \ref{rate_lemma3}, it is enough to show that
\be
\liminf_{n \to \infty}
\frac {1} {q_{n+1}} \sum_{k=1}^n q_k \log 2 a_{k-1} \leq 3 \log 2.
\ee

Assume first that there exist infinitely many $n$ such that $a_{n+1} \geq
a_m$ for every $m \leq n$.  Then
\be
\frac {1} {q_{n+1}} \sum_{k=1}^n q_k \log 2 a_{k-1} \leq \frac
{\log 2 a_{n+1}} {a_{n+1}} \frac {1} {q_n} \sum_{k=1}^n q_k \leq 3 \ln 2.
\ee

Assume now that $a=\limsup a_n<\infty$.  Take $N$ such that $a_n \leq a$ for
every $n \geq N$.  If $n \geq N$ is such that $a_{n+1}=a$ then we have
\be
\frac {1} {q_{n+1}} \sum_{k=1}^n q_k \log 2 a_{k-1} \leq 3 \log
2+\frac {1} {q_{n+1}} \sum_{k=1}^N q_k \log 2 a_{k-1}=3 \log 2+O(\frac {1}
{q_{n+1}}).
\ee
\end{proof}

Finally, combining Lemma \ref{rate_lemma1} and \ref{rate_lemma3} we conclude that
\be
\liminf \sup_{|z|=1} S(q_n,z) \leq C_1 + C_3 ~\mbox{,}
\ee
thereby completing the proof of Theorem \ref{thm_rate_main}.\qed

\begin{proof}[Proof of Theorem \ref{counter}]
We will show that if $q_n > q_{n-1}^C,\; a_{n+1}=1$ and $q_{n+2}>Cq_{n+1}$ then $ \sup_{|z|=1} S(q_n,z) \geq (1-c)\log q_{n+1}.$
Using that for $z,w\in\mathbb{R}$ with $\cos (z-w) \geq 0,$ we have 
\begin{equation} \label{sinratio}
\left|\frac{\sin z}{\sin w} -1\right | \leq \left|2\frac{\sin(z-w)}{\sin w}\right |,
\end{equation}
we obtain for $z=e^{\frac {2\pi i}{q_{n+1}^2}}$, that 
\begin{eqnarray} \nonumber
|\log |e^{2\pi i k\alpha} z - 1| -\log |e^{2\pi i k\frac{p_{n+1}}{q_{n+1}}} z - 1||&=&|\log |\frac {\sin \pi (k\alpha+1/q_{n+1}^2)}{\sin \pi (k\frac{p_{n+1}}{q_{n+1}}+1/q_{n+1}^2)}|\leq \\\nonumber
&\leq&\frac{C\sin \pi k \Delta_{n+1}}{q_{n+1}\min (\sin \pi (k\frac{p_{n+1}}{q_{n+1}}+1/q_{n+1}^2),\sin \pi (k\alpha+1/q_{n+1}^2))}\nonumber ~\mbox{,}
\end{eqnarray}
and hence, for our choice of $z$, 
\begin{equation} 
\sum_{k=0}^{q_{n+1}-1}|\log |e^{2\pi i k\alpha} z - 1| -\log |e^{2\pi i k\frac{p_{n+1}}{q_{n+1}}} z - 1||\leq \frac {C q_{n+1}\log q_{n+1}}{q_{n+2}} ~\mbox{.}
\end{equation}

We have therefore
\begin{equation} 
\sum_{k=0}^{q_{n+1}-1} \log |e^{2\pi i k\alpha} z - 1| \geq \sum_{k=0}^{q_{n+1}-1}\log |e^{2\pi i k\frac{p_{n+1}}{q_{n+1}}} z - 1| -\frac {C q_{n+1}\log q_{n+1}}{q_{n+2}}=\log |z^{q_{n+1}}-1|-\frac {C q_{n+1}\log q_{n+1}}{q_{n+2}} ~\mbox{.}
\end{equation}

On the other hand, 
\begin{equation} 
\sum_{k=0}^{q_{n+1}-1} \log |e^{2\pi i k\alpha} z - 1| =
\sum_{k=0}^{q_{n-1}-1} \log |e^{2\pi i k\alpha} z - 1|+ 
\sum_{k=0}^{q_{n}-1} \log |e^{2\pi i (k+q_{n-1})\alpha} z - 1| ~\mbox{.}
\end{equation}

We also have \cite{AvilaJitomirskaya_2009}
\begin{equation} \label{1}
\sum_{k=0}^{q_{n-1}-1} \log |e^{2\pi i k\alpha} z - 1|< C\log q_{n-1} +\log \min_{0\leq k\leq q_{n-1}-1} |e^{2\pi i k\alpha} z - 1| ~\mbox{.}
\end{equation}

As a result, with our choice of $z$, the $\min $ in (\ref{1}) is achieved at $k=0,$ and, based on the relationship between $q_{n-1},  q_{n+1}, q_{n+2},$ we conclude
\begin{eqnarray} 
\sum_{k=0}^{q_{n}-1} \log |e^{2\pi i (k+q_{n-1})\alpha} z - 1| & \geq & \log |z^{q_{n+1}}-1|-\log |z-1| -\frac {C q_{n+1}\log q_{n+1}}{q_{n+2}}-C\log q_{n-1} \nonumber \\ 
                                    & \geq & (1-c)\log q_{n+1} ~\mbox{.}
\end{eqnarray}
\end{proof}

\begin{proof}[Proof of Theorem \ref{prop_prozero}] 
Decompose
\begin{equation} \label{eq_17}
f(x) = g(x) \prod_{j=1}^{n} \left(\mathrm{e}^{2 \pi i x} - \mathrm{e}^{2 \pi i x_j}\right) ~\mbox{,}
\end{equation}
where $\{x_j, 1 \leq j \leq n\}$ denote the zeros of $f$ on $\mathbb{T}$ counting multiplicity and $g$ is zero free and analytic in a neighborhood of $\mathbb{T}$. Then, since $\log \abs{g}$ is harmonic in a neighborhood of $\mathbb{T}$, the zero free part of (\ref{eq_17}) is easily dealt with as a result of the following:
\begin{lemma} \label{lem_roc1}
Let $\alpha \in \mathbb{T}$ be a fixed irrational number and $h$ a harmonic function in a neighborhood of $\mathbb{T}$. Then for some $C>0$, 
\begin{equation}
\left \vert \frac{1}{q_n} \sum_{j=0}^{q_n-1} h(x + j \alpha) - \int_\mathbb{T} h(x) \ud \mu(x) \right \vert \leq \frac{C}{q_n} ~\mbox{, } 
\end{equation}
for all $n \in \mathbb{N}$ and uniformly in $x \in \mathbb{T}$.
\end{lemma}
The proof of Lemma \ref{lem_roc1} is fairly standard and will be given in Appendix \ref{c}. It thus remains to deal with the product in (\ref{eq_17}) which is precisely what is achieved in Theorem \ref{thm_rate_main}. 
\end{proof}

\section{Spectral consequences of the global theory} \label{sec_avilasglobal}
This section is not specific to extended Harper's model, but considers an arbitrary quasi-periodic Jacobi operator $H_\theta$ of the form (\ref{eq_hamiltonian}) with analytic sampling functions $c(\theta) \not \equiv 0$ and $v(\theta)$.

Several results of this section were first obtained in \cite{JitomirskayaMarx_2012} where certain aspects of the GT, which had originally been developed in \cite{global} for {\em{Schr\"odinger}} operators  ($c \equiv 1$), were extended to the {\em{Jacobi}} case ($c \not \equiv 1$). To keep this paper as self-contained as possible, the intention of this section is to embed our main result, Theorem \ref{thm_point}, into this framework and to discuss its spectral consequences. In particular, we will thereby reduce Theorem \ref{thm_ehmspectral} to Theorem \ref{thm_point}. For a more detailed presentation of the dynamical aspects of the spectral theory of quasi-periodic Jacobi operators, including some extensions to long-range operators, we refer the reader to the recent survey article \cite{JitomirskayaMarx_ETDS_2016_review}.

We start by recalling some definitions. Following, $M_2(\mathbb{C})$ denotes the 2$\times$2 complex matrices, and $\Vert . \Vert$ is any fixed matrix norm. Given $\alpha \in \mathbb{T}$ irrational and $D:\mathbb{T} \to M_2(\mathbb{C})$ measurable with $\log_+\Vert D(.) \Vert \in L^1(\mathbb{T})$, a quasi-periodic {\em{cocycle}} $(\alpha, D)$ is a dynamical system on $\mathbb{T} \times \mathbb{C}^2$ defined by $(\alpha,D)(\theta,v):=(\theta + \alpha, D(\theta) v)$. If $D$ is analytic, $(\alpha,D)$ is called an {\em{analytic cocycle}}. An analytic cocycle $(\alpha,D)$ where $\det D(\theta_0) = 0$ for some $\theta_0 \in \mathbb{T}$  is called {\em{singular}}, and {\em{non-singular}} otherwise. 

The averaged asymptotics of any cocycle $(\alpha,D)$ is quantified by its (top) {\em{Lyapunov exponent}},
\begin{equation}
L(\alpha, D) := \lim_{n \to \infty} \dfrac{1}{n} \int_{\mathbb{T}} \log \Vert D(\theta+(n-1) \alpha) \dots D(\theta) \Vert \ud \mu(\theta) ~\mbox{,}
\end{equation}
which is well-defined by subadditvity with values in $[-\infty, + \infty)$.

In view of quasi-periodic analytic Jacobi operators, the relevant analytic cocycle is induced by 
\begin{equation} \label{eq_defnjacobico}
A^E(\theta):= \begin{pmatrix} E - v(\theta) & - \widetilde{c}(\theta - \alpha) \\ c(\theta) & 0 \end{pmatrix} ~\mbox{,}
\end{equation}
where the spectral parameter $E$ ranges in $\mathbb{R}$. Here, for $\epsilon \in \mathbb{R}$ and $\theta \in \mathbb{T}$, we define $\widetilde{c}(\theta+i\epsilon):=\overline{c(\theta - i \epsilon)}$ as the reflection of $c$ along the real axis. Morally, $\widetilde{c}$ analytically ``re-interpretes'' $\overline{c}$ appearing in (\ref{eq_hamiltonian}), which agrees with $\tilde{c}$ on $\mathbb{T}$. 

Iterates of $(\alpha, A^E)$ relate to solutions of the finite difference equation $H_\theta \psi = E \psi$ over $\mathbb{C}^\mathbb{Z}$, cf (\ref{eq_deftransfer}). If $c$ has zeros on $\mathbb{T}$, Jacobi cocycles $(\alpha, A^E)$ provide important examples for singular cocycles since $\det A^E(\theta) = c(\theta) \widetilde{c}(\theta - \alpha)$. With this in mind, one calls a quasi-periodic Jacobi operator {\em{singular}} if $c$ has zeros on $\mathbb{T}$, and {\em{non-singular}} otherwise.

The GT stratifies the energy axis according to the behavior of {\em{complexified Lyapunov exponent of a quasi-periodic Jacobi operator}}, defined by
\begin{equation} \label{eq_defcomplexle}
L(E;\epsilon):= L(\alpha, A_\epsilon^E) - \int_\mathbb{T} \log\vert c(\theta) \vert ~\ud \mu(\theta) ~\mbox{,}
\end{equation}
for real $\epsilon$ in a neighborhood of $\epsilon = 0$. Here, for fixed $E \in \mathbb{R}$, $L(\alpha, A_\epsilon^E)$ is the Lyapunov exponent obtained by phase-complexifying the Jacobi cocycle, 
\begin{equation} A_\epsilon^E(\theta):= A^E(\theta+i\epsilon) ~\mbox{.}
\end{equation}
As we shall elaborate, the GT relates the complexified Lyapunov exponent of a given quasi-periodic analytic Jacobi operator to its spectral properties. We also note that by letting $\epsilon = 0$, the complexified Lyapunox exponent reduces to what is usually called the Lyapunov exponent of a Jacobi operator; for simplicity, we denote the latter by $L(E):=L(E; 0)$. In view of Theorem \ref{thm_global} mentioned below, we recall that $L(E) \geq 0$ for all $E \in \mathbb{R}$.

\begin{remark} \label{rem_defcomplexLE}
For later purposes, we emphasize that in the definition of the complexified Lyapunov exponent (\ref{eq_defcomplexle}), we complexified the Jacobi cocycle $(\alpha, A^E)$ and {\em{not}} the measurable cocycle $(\alpha, B^E)$. The latter generates solutions to the finite difference equation and is defined below in (\ref{eq_deftransfer}). In particular, the logarithmic integral on the right hand side of (\ref{eq_defcomplexle}) carries no $\epsilon$-dependence. Indeed, as explained in \cite{JitomirskayaMarx_2013_erratum}, $L(\alpha, B^E)$ would not even be an even function in $\epsilon$ (see also Appendix \ref{app_thm_accel}); evenness in $\epsilon$ is crucial for the partition of the spectrum into subcritical, supercritical, and critical energies introduced below. Moreover, there is an important dynamical reason underlying the definition of the complexified LE, which will be explored in Sec. \ref{sec_afk_reductions}, see the comment following (\ref{eq_defcomplexle_1}). The latter plays a role in the spectral theoretic implications of the GT, which are discussed below and in Sec. \ref{sec_afk} - \ref{sec_almredimpliesac}.
\end{remark}

It is well known from Kotani theory that the set 
\begin{equation}
\mathcal{Z}:= \{E \in \mathbb{R}~:~ L(E)=0\}
\end{equation}
forms an essential support of the ac spectrum of $H_\theta$. One of the main achievements of the GT, however, is that it refines Kotani theory by explicitly separating contributions from purely singular (sc+pp) spectrum from those of purely ac spectrum. 

The GT relies on the properties of the complexified LE, which we summarize in Theorem \ref{thm_global}. Following, we denote by $\Sigma$ the spectrum of $H_{\theta}$, which is well known to be independent of $\theta$.
\begin{theorem} \label{thm_global}
Fixing $E \in \mathbb{R}$, $L(E;\epsilon)$ is a non-negative, even, piecewise linear, and convex function in $\epsilon$ with right derivatives satisfying
\begin{equation}
\omega(E;\epsilon) = \dfrac{1}{2 \pi} \lim_{h \to 0+} \dfrac{L(\alpha, A_{\epsilon+h}^E) - L(\alpha, A_\epsilon^E)}{h} \in \frac{1}{2} \mathbb{Z} ~\mbox{.}
\end{equation}
Moreover, for every $E \in \mathbb{R}$ with $L(E)>0$, $E \in \Sigma$ if and only if $\omega(E;\epsilon)$ has a jump discontinuity at $\epsilon=0$, or equivalently, $\omega(E;0)>0$.
\end{theorem}
\begin{remark}
For certain applications it is useful to know that for {\em{non-singular}} Jacobi operators, one has in fact that $\omega(E;\epsilon) \in \mathbb{Z}$ for all $\epsilon$ in any neighborhood of $\epsilon=0$ where $c(. +i \epsilon)$ does not vanish, see Theorem 1 in \cite{JitomirskayaMarx_2013_erratum}. This played an important role in the computation of the complexified Lyapunov exponent for extended Harper's model.
\end{remark}
$\omega(E;\epsilon)$ is called the {\em{acceleration}} and was first introduced for Schr\"odinger operators in \cite{global}; correspondingly, the fact that $\omega(E;\epsilon) \in \frac{1}{2} \mathbb{Z}$ is known as ``{\em{quantization of the acceleration}}.'' Likewise, Theorem \ref{thm_global} first appeared in \cite{global} for the special case of Schr\"odinger cocycles. In its present formulation, Theorem \ref{thm_global} includes results from \cite{JitomirskayaMarx_2012, JitomirskayaMarx_2013_erratum, AvilaJitomirskayaSadel_2013, Marx_2014}. For convenience of the reader, we assemble these results in Appendix \ref{app_thm_accel} and also provide simplified proofs of certain aspects. 

To discuss the stratification of the spectrum implied by Theorem \ref{thm_global}, we first distinguish between non-singular and singular Jacobi operators. We mention that some of the below-mentioned spectral consequences of the GT were in fact developed earlier or in parallel to the GT; important contributions were made in \cite{AvilaFayadKrikorian_2011, AvilaKrikorian_2006, AvilaJitomirskaya_2010}. For further context of the historical developments leading to the GT, including a more comprehensive list of references, we refer the reader to survey article \cite{JitomirskayaMarx_ETDS_2016_review}.

\subsection{Non-singular Jacobi operators} \label{subsec_nonsingjac}
Taking into account Theorem \ref{thm_global}, we partition the set $\mathcal{Z}$ into {\em{subcritical}} energies, where $\omega(E;\epsilon)$ does {\em{not}} exhibit a jump discontinuity at $\epsilon=0$ (correspondingly, $\omega(E; 0) = 0$), and {\em{critical}} energies with, correspondingly, $\omega(E; 0) > 0$. Any $E \in  \Sigma$ where $L(E)>0$ is called {\em{supercritical}}. We remark, that this terminology was inspired by the spectral properties of the almost Mathieu operator \cite{global, JitomirskayaMarx_2012}. Identifying subcritical and critical energies in $\mathcal{Z}$ yields above mentioned resolution of $\mathcal{Z}$ which explicitly identifies contributions from singular and ac spectrum:
 
\begin{itemize}
\item {\bf{Critical behavior}} is associated with {\em{singular}} (sc + pp) spectrum, as a consequence of:
\begin{theorem} \label{thm_coroafk}
Given a {\em{non-singular}} quasi-periodic, analytic Jacobi operator with irrational $\alpha,$ the set of critical energies has zero Lebesgue measure.
\end{theorem}
Theorem \ref{thm_coroafk} was first obtained for quasi-periodic, analytic Schr\"odinger operators in \cite{AvilaFayadKrikorian_2011}. In Sec. \ref{sec_afk} we extend this result to the Jacobi case, thereby proving Theorem \ref{thm_coroafk}. 

\vspace{0.2 cm}

\item {\bf{Subcritical behavior}} identifies the contribution from ac spectrum as a consequence of:
\begin{theorem} \label{thm_subcritical}
Let $H_\theta$ be a non-singular quasi-periodic, analytic Jacobi operator with irrational frequency $\alpha$. Then, for $\mu$-a.e. $\theta$, all its spectral measures are purely ac on the set of subcritical energies.
\end{theorem}
\begin{remark}
Theorem \ref{thm_subcritical} is known for Schr\"odinger operators; its proof for Jacobi operators will be the subject of both Sec. \ref{sec_afk} and Sec. \ref{sec_almredimpliesac}. In essence, Theorem \ref{thm_subcritical} relies on a general dynamical result known as {\em{almost reducibility theorem}} (ART) which shows equivalence between subcritical behavior of analytic $SL(2, \mathbb{R})$-cocycles and a certain dynamical property known as {\em{almost reducibility}}, see Def. \ref{def_almostred} in Sec. \ref{sec_almredimpliesac}.  A proof of ART is announced in \cite{global}, to appear in \cite{Avila_prep_ARC_2}; the latter extends an earlier result which proves ART for exponentially Liouvillean $\alpha$ \cite{Avila_prep_ARC_1}. For Jacobi operators, the relevant analytic $SL(2, \mathbb{R})$ cocycles will be given in (\ref{eq_defApr}) of Sec. \ref{sec_afk_reductions}. Given ART, the missing link to Theorem \ref{thm_subcritical} is to prove that almost reducibility implies purely ac spectrum. For Schr\"odinger operators this was established in \cite{Avila_prep_ARC_1} for $\mu$-a.e. $\theta$, and, using a much more delicate argument, for {\em{all}} $\theta \in \mathbb{T}$ in \cite{Avila_prep_ARC_2}. Since the statement for $\mu$-a.e. $\theta$ is enough for the spectral theory of extended Harper's model (Theorem \ref{thm_ehmspectral}), we will limit our proof for Jacobi operators to this a.e. statement which is the subject of Sec. \ref{sec_almredimpliesac}.
\end{remark}
\end{itemize}

\subsection{Singular Jacobi operators} \label{susec_singjac}
Like for non-singular Jacobi operators, all $E \in  \Sigma$ where $L(E)>0$ are called {\em{supercritical}}. Even though Theorem \ref{thm_global} holds irrespective of whether the Jacobi operator is singular or non-singular, dividing the set $\mathcal{Z}$ into subcritical and critical behavior as above does not provide additional insight. Indeed, by a well known argument \cite{Dombrowsky_1978} (see also \cite{JitomirskayaMarx_2012}, Proposition 7.1 therein), one has:
\begin{prop} \label{fact_singular}
Let $H_\theta$ be a {\em{singular}} quasi-periodic analytic Jacobi operator with irrational frequency $\alpha$. Then, for all $\theta \in \mathbb{T}$, the ac spectrum of $H_\theta$ is empty.
\end{prop}

In summary, combining Sec. \ref{subsec_nonsingjac} and \ref{susec_singjac}, the GT yields a full characterization of the spectral properties of {\em{both}} singular and non-singular Jacobi operators, provided one can establish the content of the following conjecture, which we call the {\em{critical energy conjecture}}:
\begin{conj}[Critical energy conjecture (CEC)] \label{conj_cec}
Let $\alpha$ be irrational and $H_\theta$ be a quasi-periodic Jacobi operator with analytic sampling functions.
\begin{itemize}
\item[(i)] If the Jacobi operator is {\em{non-singular}}, the spectrum on the set of critical energies is purely sc for $\mu$-a.e. $\theta$. 
\item[(ii)] If the Jacobi operator is {\em{singular}}, the spectrum on the set $\mathcal{Z}$ is purely sc for $\mu$-a.e. $\theta$.
\end{itemize}
\end{conj}
\begin{remark} \label{rem_CEC}
The CEC yields a sought-after {\em{direct}} criterion for detecting presence of sc spectrum for quasi-periodic Jacobi operators with analytic sampling functions. Even though the CEC was at least implicit in \cite{JitomirskayaMarx_2012, JitomirskayaMarx_2013_erratum}, in the present form the CEC appears first in this article. We also mention that it can be considered a special case of a problem posed by Damanik in \cite{Damanik_SimonFest_Kotani}, asking to prove or disprove that for ergodic Schr\"odinger operators, the set of zero LE does not contain any eigenvalues.
\end{remark} 

\section{Applications to extended Harper's model} \label{sec_avilasglobal_applicehm}
In \cite{JitomirskayaMarx_2012, JitomirskayaMarx_2013_erratum} we explicitly computed the complexified Lyapunov exponent for extended Harper's model, thereby identifying subcritical, critical, and supercritical energies for all values of $\lambda$ and all irrational $\alpha$. Theorem \ref{thm_complexLEEHM} summarizes these results and the arising phase diagram in the sense of the GT is depicted in Fig. \ref{figure_phased}. Theorem \ref{thm_complexLEEHM} in particular shows that respective type of behavior (i.e., subritical, supercritical, or critical) only depends on $\lambda$, i.e., is the same everywhere on the spectrum and is independent of $\alpha$. 
\begin{theorem}[Corollary 5.1. in \cite{JitomirskayaMarx_2012} and Sec. 4.5 in \cite{Marx_thesis}] \label{thm_complexLEEHM} 
For $\alpha$ irrational, all energies in the spectrum of extended Harper's model are
\begin{itemize}
\item[(i)] supercritical for all $\lambda \in I^\circ \cup \{ \lambda_1 + \lambda_3 = 0 ~,~ 0 < \lambda_2 < 1\}$,
\item[(ii)] subcritical for all $\lambda \in II^\circ \cup \{ \lambda_1 + \lambda_3 = 0 ~,~ \lambda_2 > 1\}$,
\item[(iii)] subcritical for all $\lambda \in III^\circ$ if $\lambda_1 \neq \lambda_3$
\item[(iv)] critical for all $\lambda \in III^\circ$ if $\lambda_1 = \lambda_3$
\item[(v)] critical for all $\lambda \in L_I \cup L_{II} \cup L_{III}$
\end{itemize}
\end{theorem}
\begin{figure}[ht] 
\centering
\subfigure[$\lambda_1 \neq \lambda_3$]{
\includegraphics[width=0.4\textwidth]{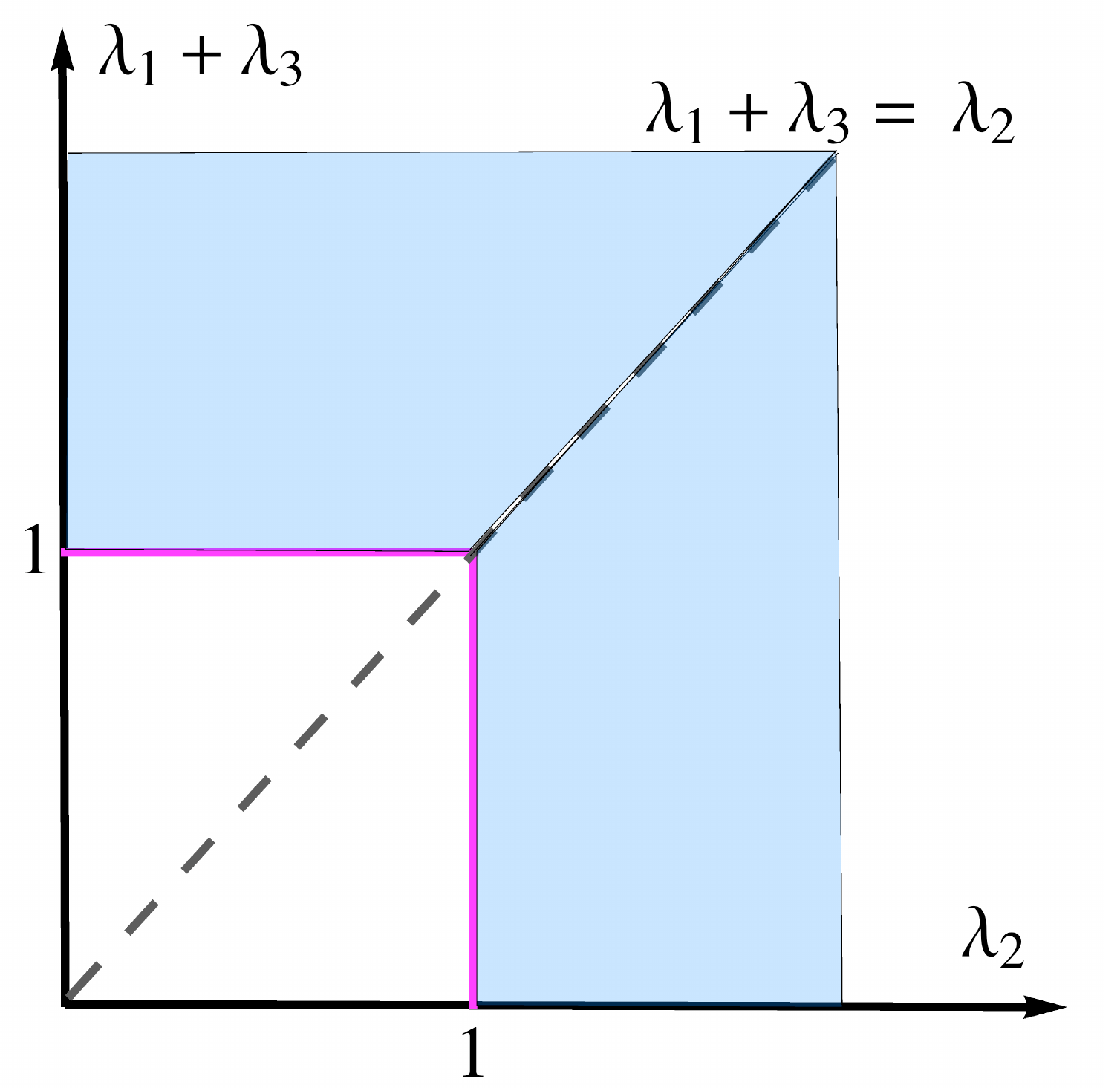} \label{figure_phased_aniso}
}
\subfigure[$\lambda_1 = \lambda_3$]{
\includegraphics[width=0.4\textwidth]{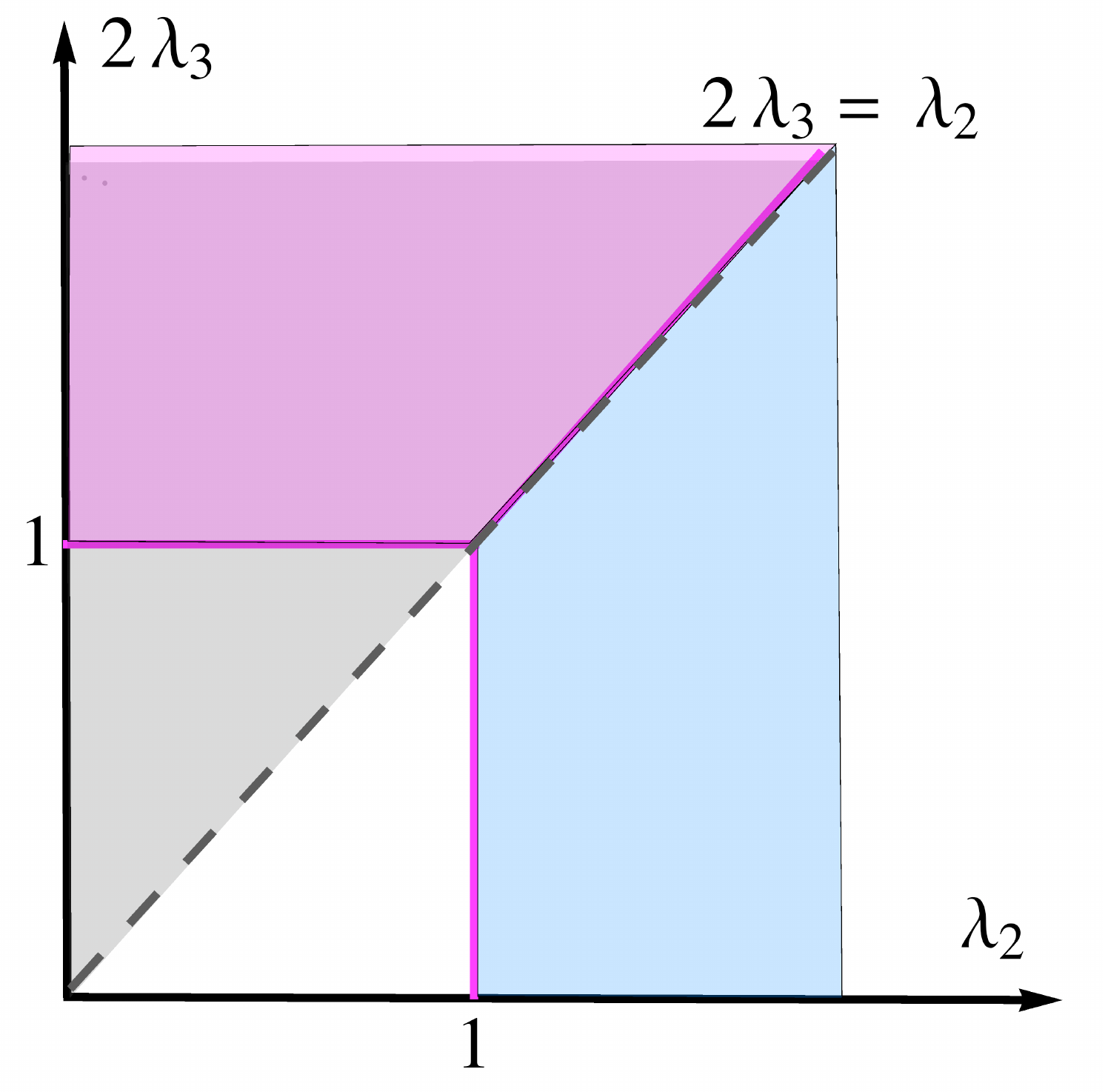} \label{figure_phased_iso}
}
\caption{Phase diagram for extended Harper's model in the sense of GT. Based on Proposition \ref{obs_cfun}, the dashed grey line and shaded grey area (panel (b)) indicate singularity of the underlying Jacobi operator.  Subcritical behavior is shown in light blue. Areas in red indicate critical behavior if extended Harper's model is non-singular, and zero LE, if extended Harper's model is singular (cf Proposition \ref{fact_singular}). The remaining region, corresponding to $0 \leq \lambda_1 + \lambda_3 < 1$ and $0 \leq \lambda_2 <1$ shows supercritical behavior. } 
\label{figure_phased}
\end{figure}

Notice that Theorem \ref{thm_complexLEEHM} exhibits a symmetry-induced transition in $III^\circ$ from subcritical behavior, if $\lambda_1 \neq \lambda_3$, to critical behavior, if $\lambda_1 = \lambda_3$; as a consequence of Theorem \ref{thm_point}, the latter results in the spectral collapse from ac to sc spectrum given in Theorem \ref{thm_ehmspectral}.

Moreover, the presence of singularities of extended Harper's model is quantified by the following proposition, which is easily verified by direct computation:
\begin{prop} \label{obs_cfun}
Letting $z=\theta + i \epsilon$, $\theta \in \mathbb{T}$, $c_{\lambda}(z)$ has at most {\em{two}} zeros. Necessary conditions for real roots are $\lambda_{1} = \lambda_{3}$ or $\lambda_{1} + \lambda_{3} = \lambda_{2}$. Moreover,
\begin{itemize}
\item[(a)] for $\lambda_{1} = \lambda_{3}$, $c_{\lambda}(z)$ has real roots if and only if $2\lambda_{3} \geq \lambda_{2}$, determined by
\begin{equation} \label{eq_conditioncfun}
2 \lambda_3 \cos(2 \pi (\theta + \frac{\alpha}{2} )) = - \lambda_2 ~\mbox{,}
\end{equation}
and giving rise to a double root at $\theta = \frac{1}{2} - \frac{\alpha}{2}$ if $\lambda_2 = 2 \lambda_3$.
\item[(b)] for  $\lambda_1 \neq \lambda_3$, $c_\lambda(\theta)$ has only one simple real root at $\theta = \pm \frac{1}{2} - \frac{\alpha}{2}$ if $\lambda_1 + \lambda_3 = \lambda_2$.
\end{itemize}
\end{prop}

Combining Theorem \ref{thm_complexLEEHM} with Proposition \ref{obs_cfun}, the content of Sec. \ref{sec_avilasglobal} reduces the proof of Theorem \ref{thm_ehmspectral} to excluding point-spectrum in the self-dual regime, as claimed by Theorem \ref{thm_point}. In particular, Theorem \ref{thm_point} establishes the CEC for the special case of extended Harper's model. 

\section{A dynamical formulation of Aubry-Andr\'e duality} \label{sec_Aubry}

Following, we will assume that $\lambda_2 > 0$, in which case Aubry-Andr\'e duality is expressed by the map $\sigma(\lambda)$ defined in (\ref{eq_sigma}). If $\lambda_2 =0$, the theorems of Sec. \ref{sec_Aubry} and \ref{sec_selfdual} may be adapted to still hold true. Since the underlying ideas are analogous, we postpone the details to Appendix \ref{app_zeronn}. 

First, recall that the solutions to the time-independent Schr\"odinger equation $H_{\theta;\lambda,\alpha} \psi = E \psi$ over $\mathbb{C}^\mathbb{Z}$ can be generated iteratively using the transfer matrix
\begin{eqnarray} \label{eq_deftransfer}
B_{\lambda}^{E}(\theta) := \dfrac{1}{c_{\lambda}(\theta)} \begin{pmatrix} E - v(\theta) & -\overline{c_{\lambda}(\theta-\alpha)} \\ c_{\lambda}(\theta) & 0 \end{pmatrix} ~\mbox{.}
\end{eqnarray}
Since $c_{\lambda}$ and $v$ are analytic on $\mathbb{T}$, $B_{\lambda}^{E}(\theta)$ is well-defined except for the possibly finitely many $\theta \in \mathbb{T}$ where $c_\lambda(\theta) = 0$ (quantified in Proposition \ref{obs_cfun}). Given $\lambda$, let $\mathfrak{Z}(\lambda):=\{\theta \in \mathbb{T}: c_{\lambda}(\theta)=0\}$ and set $\mathbb{T}_0(\lambda):= \mathbb{T} \setminus \cup_{\theta \in \mathfrak{Z}(\lambda)} \mathcal{O}(z)$, where $\mathcal{O}(\theta):=\{\theta + n \alpha(\mathrm{mod}1), n \in \mathbb{Z}\}$. Clearly, $\mu(\mathbb{T}) = \mu(\mathbb{T}_0(\lambda)) = 1$.

Thus, fixing $\lambda$, for all $\theta \in \mathbb{T}_0(\lambda)$ (and hence $\mu$-a.e. on $\mathbb{T}$), solutions $\psi \in \mathbb{C}^\mathbb{Z}$ of $H_{\theta;\lambda,\alpha} \psi = E \psi$ are generated by iterating the measurable cocycle $(\alpha, B_\lambda^E)$:
\begin{eqnarray} 
& \begin{pmatrix} \psi_{n} \\ \psi_{n-1} \end{pmatrix} = B_{\lambda; n}^{E}(\alpha,\theta) \begin{pmatrix} \psi_{0} \\ \psi_{-1} \end{pmatrix}  ~\mbox{,} \\
& B_{\lambda; n}^{E}(\alpha,\theta) := B_{\lambda}(\theta + \alpha (n-1)) \dots  B_{\lambda}(\theta) ~\mbox{,} \\
& B_{\lambda; -n}^{E}(\alpha,\theta) := B_{\lambda; n}^{E}(\alpha,\theta - n \alpha)^{-1} ~\mbox{,} ~n\geq 1 ~\mbox{.} \label{eq_iterates}
\end{eqnarray}

Suppose now that for some $\lambda$ and $\theta$, $H_{\theta;\lambda, \alpha}$ has an eigenvalue $E \in \mathbb{R}$ with respective eigenvector $(u_n)$. Then, considering its Fourier transform, 
\begin{equation}
u(x) : = \sum_{n\in\mathbb{Z}} u_n \mathrm{e}^{2 \pi i n x} \in L^2(\mathbb{T})\setminus \{0\} ~\mbox{,}
\end{equation} 
and letting
\begin{equation}
M_\theta(x) = \left(\begin{array}{c c} u(x) & u(-x) \\ \mathrm{e}^{-2 \pi i \theta} u(x - \alpha) & \mathrm{e}^{2 \pi i \theta} u(-(x-\alpha)) \end{array}\right)  ~\mbox{,}
\end{equation}
Aubry-Andr\'e duality can be formulated as the $L^2$-semiconjugacy:
\begin{equation} \label{eq_semiconj}
B_{\sigma(\lambda)}^{E/\lambda_2}(x) M_\theta(x) = M_\theta(x+\alpha) R_\theta, ~R_\theta = \begin{pmatrix} \mathrm{e}^{2 \pi i \theta} & 0 \\ 0 & \mathrm{e}^{-2 \pi i \theta} \end{pmatrix} ~\mbox{.}
\end{equation} 

For all non-$\alpha$-rational phases (see Definition \ref{def_alpharational}), the semi-conjugacy of (\ref{eq_semiconj}) is in fact an $L^2$-conjugacy:
\begin{prop} \label{prop_det}
Let $\theta$ {\em{not}} be $\alpha$-rational. Then, for a.e. $x \in \mathbb{T}$, $\det M_\theta(x) \neq 0$. Moreover, for some $b > 0$, one has
\begin{equation}
\abs{\det M_\theta(x)} = \dfrac{b}{\abs{c(x - \alpha)}} ~\mbox{.}
\end{equation}
\end{prop}
The statement is known for analytic Schr\"odinger operators where it played a significant role in a {\em{quantitative}} version of the Aubry-Andr\'e duality \cite{AvilaJitomirskaya_2010}. Since the proof of Proposition \ref{prop_det} only requires slight modifications of the Schr\"odinger case, we defer it to Appendix \ref{app_aubry}.

In summary we have thus arrived at the following characterization of solutions of dual points in parameter space:
\begin{prop} \label{coro_bddsol}
For given irrational $\alpha$, suppose $\lambda$ and $\theta$ are such that $H_{\theta; \lambda, \alpha}$ has an eigenvalue $E \in \mathbb{R}$. If $\theta$ is {\em{not}} $\alpha$-rational, the cocycle $(\alpha, B_{\sigma(\lambda)}^{E/\lambda_2})$ is $L^2$-conjugate to the complex rotation $(\alpha, R_\theta)$. In particular, 
if the eigenfunction associated with $E$ is in $l^1(\mathbb{Z})$, one has
\begin{equation} \label{eq_growthtransferm}
\sup_{x \in \mathbb{T}} \Vert  B_{\sigma(\lambda);n}^{E/\lambda_2}(x) \Vert = \mathcal{O}(1) ~\mbox{.}
\end{equation}
\end{prop}

As mentioned earlier, the analogues of Propositions \ref{prop_det} and \ref{coro_bddsol} are known for analytic Schr\"odinger operators \cite{AvilaJitomirskaya_2010}, see Theorem 2.5 therein. 

To conclude, we apply Proposition \ref{coro_bddsol} to the interior of region $\mathrm{II}$. As usual, $\alpha$ is called {\em{Diophantine}} if 
\begin{equation} \label{eq_diophcond}
\vert \sin(2 \pi n \alpha) \vert > \dfrac{\kappa}{\vert n \vert^r} ~\mbox{, } n \in \mathbb{Z} \setminus \{0\} ~\mbox{,}
\end{equation}
for some $r>1$ and $\kappa>0$. We make use of the following result:
\begin{theorem}[Theorem 1 in \cite{JitomirskayaKosloverSchulteis_2005}] \label{thm_jks}
Let $\alpha$ be Diophantine and fix $\lambda \in \mathrm{I}^\circ$. For a full measure set of phases, $H_{\theta;\lambda, \alpha}$ is purely point with exponentially localized eigenfunctions. 
\end{theorem}
Theorem \ref{thm_jks} and Proposition \ref{coro_bddsol} consequently imply:
\begin{theorem} \label{thm_dualregime}
Let $\alpha$ Diophantine and $\lambda \in II^\circ$. For a.e. $x \in \mathbb{T}$, the spectrum of $H_{x;\lambda, \alpha}$ is purely absolutely continuous.
\end{theorem}
\begin{proof}
Given $\lambda \in II^\circ$, let $\Omega$ be the full measure set of phases $\theta \in \mathbb{T}$ for which Theorem \ref{thm_jks} asserts localization of the {\em{dual}} operator $H_{\theta; \sigma(\lambda), \alpha}$. Since the $\alpha$-rational phases are only a countable set, we may assume them to be removed from $\Omega$. Let
\begin{equation} \label{eq_defsignma0}
\Sigma_0:=\cup_{\theta \in \Omega} \sigma_{\mathrm{pt}}(H_{\theta;\sigma(\lambda), \alpha}) ~\mbox{.}
\end{equation}

By a standard argument based on subordinacy theory (or alternatively, using \cite{LastSimon_1999}), (\ref{eq_growthtransferm}) already implies pure ac-spectrum of $H_{x;\lambda, \alpha}$ on $\Sigma_0$, for all $x \in \mathbb{T}$. Thus the theorem follows if we can show that for $\mu$-a.e. $x \in \mathbb{T}$, $\mathbb{R} \setminus \Sigma_0$ does not support any spectrum of $H_{x;\lambda, \alpha}$.

To see this, denote by 
\begin{equation} \label{eq_dosmeasure}
n(\lambda, \alpha; . ):=\int_\mathbb{T} \nu(x, \lambda, \alpha; .) \ud \mu(x) ~\mbox{,}
\end{equation}
the density of states measure for $H_{x;\lambda, \alpha}$, where $\nu(\theta, \lambda, \alpha; .)$ is the spectral measure of $H_{x;\lambda, \alpha}$ and $\delta_0 \in l^2(\mathbb{Z})$. Invariance of the density of states under duality implies
\begin{equation}
n(\lambda, \alpha; \mathbb{R} \setminus \Sigma_0) = \int_{\mathbb{T}} \nu( \theta, \sigma(\lambda), \alpha; \lambda_2^{-1} \left( \mathbb{R} \setminus \Sigma_0\right) ) \ud \mu(\theta) = 0 ~\mbox{,}
\end{equation}
where the last equality follows by definition of $\Sigma_0$. Thus, for a.e. $x \in \mathbb{T}$, $\nu(x, \lambda, \alpha; \mathbb{R} \setminus \Sigma_0)=0$, which proves above claim.
\end{proof}
\begin{remark}
Given ART, the content of Theorem \ref{thm_dualregime} extends to {\em{all}} phases and {\em{all irrational}} frequencies.
\end{remark}

\section{Absence of point spectrum in the self-dual regime} \label{sec_selfdual}

We will now explore the formulation of Aubry-Andr\'e duality given in the previous section to prove absence of point spectrum for $\lambda \in \mathcal{SD}$. The proof of Theorem \ref{thm_point}  is done by contradiction, leading to the set-up of Section \ref{sec_Aubry}. 

To give a preview of what is to come for the self-dual extended Harper's model, we start with the special case of the critical almost Mathieu operator. Recall from Sec. \ref{sec_intro} that the latter arises from extended Harper's model by letting $\lambda_1 = \lambda_3=0$ and $\lambda_2=1$.

\subsection{Warm-up: The critical almost Mathieu operator} \label{sec_selfdual_criticalamo}
We aim to prove Theorem \ref{thm_point} in the special case of the critical almost Mathieu operator:
\begin{theorem} \label{thm_AMO}
For all irrational $\alpha$, the critical almost Mathieu operator has empty point spectrum for all phases $\theta$ which are not $\alpha$-rational.
\end{theorem}
\begin{remark}
As pointed out also in Remark \ref{rem_criticalAMO_sc}, Theorem \ref{thm_AMO} has so far only appeared in the preprint \cite{Avila_preprint_2008_2}, which was not intended for publication.
\end{remark}

Since it is known from \cite{global} (see also \cite{JitomirskayaMarx_2012}, for an alternative proof) that all energies in the spectrum of the critical almost Mathieu operator are critical in the sense of the GT, Theorem \ref{thm_AMO} immediately implies Theorem \ref{AM}.

Since the critical almost Mathieu operator amounts to extended Harper's model with $\lambda=(1,0,1)$, the transfer matrix in (\ref{eq_deftransfer}) simplifies to
\begin{equation} \label{eq_AMO_0}
B^E(x) = \begin{pmatrix} E - 2 \cos(2 \pi x) & -1 \\ 1 & 0 \end{pmatrix} ~\mbox{.}
\end{equation}
Notice also that $(1,0,1)$ is a fixed point of $\sigma$, whence the transfer matrix of the critical almost Mathieu operator is invariant under duality.

\begin{proof}[Proof of Theorem \ref{thm_AMO}]
Assume that the critical almost Mathieu operator had an eigenvalue $E$ for some phase $\theta$ which is not $\alpha$-rational. Then, Proposition \ref{prop_det} yields the $L^2$-conjugacy,
\begin{equation} \label{eq_AMO_1}
 B^E(x)  = M_\theta(x+ \alpha) R_\theta M_\theta(x)^{-1} ~\mbox{.}
\end{equation}

Inspired by (\ref{eq_AMO_1}), we compare the cocycle dynamics before and after the coordinate change, introducing
\begin{equation} \label{eq_AMO_2}
\Psi^{(n)}(x):= \mathrm{tr} \{ B_n^E(x) - R_\theta^n \} = \mathrm{tr} \{ B_n^E(x) \} - 2 \cos(2 \pi n \theta) ~\mbox{.}
\end{equation}
Here, as before, we denote $B_n^E(x):= B^E(x + (n-1) \alpha) \dots B^E(x)$.

$B^E(x)$ only involves trigonometric polynomials of degree 1, whence $\Psi^{(n)}$ is a trigonometric polynomial of degree $n$. The simple form of $B^E(x)$ allows to immediately write down its boundary Fourier coefficients,
\begin{equation}
\widehat{\Psi^{(n)}}(\pm n) = (-1)^n \prod_{k=0}^{n-1} \mathrm{e}^{\pm 2 \pi i k \alpha} = (-1)^n \mathrm{e}^{ \pm \pi i \alpha n (n-1) } ~\mbox{,}
\end{equation}
which in particular implies
\begin{equation} \label{eq_AMO_3}
\vert \widehat{\Psi^{(n)}}(\pm n) \vert = 1 ~\mbox{.}
\end{equation}

To contrast this, using (\ref{eq_AMO_1}), we estimate 
\begin{eqnarray}
\vert  \Psi^{(n)}(x) \vert & = & \vert \mathrm{tr} \left\{ \left[ M_\theta(x+n\alpha) - M_\theta(x) \right] R_\theta^n M_\theta(x)^{-1} \right\} \vert \nonumber \\
                                                & \leq & 2 \Vert M_\theta(x+ n \alpha) - M_\theta(x) \Vert \cdot \Vert M_\theta(x) \Vert ~\mbox{.} \label{eq_AMO_4}
\end{eqnarray}
We mention that (\ref{eq_AMO_4}) uses cyclicity of the trace and the straightforward bounds, $\mathrm{tr}(A) \leq 2 \norm{A}$ and $\norm{A^{-1}} = \frac{\norm{A}}{\vert \det(A) \vert}$ for $A \in GL(2, \mathbb{C})$. 

Recalling that $M_\theta \in L^2(\mathbb{T},SL(2,\mathbb{C}))$, Cauchy-Schwarz yields
\begin{equation} \label{eq_AMO_5}
\Vert \Psi^{(n)} \Vert_{L^1(\mathbb{T})} \leq \Vert \norm{M_\theta(. + \alpha n) - M_\theta(.)}\Vert_{L^2(\mathbb{T})} \Vert \norm{M_\theta(.)} \Vert_{L^2(\mathbb{T})} ~\mbox{.} 
\end{equation}

Finally, since $\vertiii{ q_n \alpha} \to 0$, (\ref{eq_AMO_5}) implies that $\Vert \Psi^{(q_n)} \Vert_{L^1(\mathbb{T})} = o(1)$ as $n \to \infty$, which contradicts (\ref{eq_AMO_3}).
\end{proof}

\subsection{Including next nearest neighbor interaction}

Before turning to the proof of Theorem \ref{thm_point}, we comment on the exclusion of the zero-measure set of phases in its statement. First, notice that given $\alpha$, consideration of the set of $\alpha$-rational phases is a priori excluded for all $\lambda$ because our strategy relies on Proposition \ref{prop_det}. 

For the same reason, this a priori exclusion of phases has already been encountered in Sec. \ref{sec_selfdual_criticalamo} for the critical almost Mathieu operator. In fact, our proof shows that for $\lambda_1 \neq \lambda_3$, empty point spectrum for the self-dual extended Harper's model holds for {\em{all}} non $\alpha$-rational phases. 

As opposed to the critical almost Mathieu operator, one can however claim that the exclusion of $\alpha$-rational phases is in general necessary for extended Harper's model: For $\lambda_1 = \lambda_3$, presence of real zeros of the sampling function $c_\lambda(x)$, generating off-diagonal elements of the Jacobi operator, allows for phases where the operator has a finite decoupled block, and thus eigenvalues.
\begin{prop} \label{prop_ehmexcludephase}
Fix $\alpha$ irrational and let $\lambda_1 = \lambda_3$. There exists a dense set of $\lambda \in \mathrm{III}^\circ$ and a corresponding $\alpha$-resonant phase $\theta = \theta(\lambda)$ such that $\sigma_{\mathrm{pt}}(H_{\theta; \lambda, \alpha}) \neq \emptyset$. 
\end{prop}
\begin{proof}
By Proposition \ref{obs_cfun} (a), whenever $\lambda_1 = \lambda_3$ and $2\lambda_3 > \lambda_2$, $c_\lambda(\theta)$ has two distinct real roots $\theta_1, \theta_2$ determined by (\ref{eq_conditioncfun}). Thus, if $\theta_1, \theta_2$ are such that for some $n \in \mathbb{Z}$ one has $\vert \theta_1 - \theta_2 \vert = n \alpha$, the Jacobi operator will have a finite decoupled block of size $(\vert n \vert -1)$. Using (\ref{eq_conditioncfun}), this happens if and only if $\theta_1 = \theta_1(\lambda)$ is $\alpha$-rational.

Since for given $\alpha$, the set of $\alpha$-rational phases is dense in $\mathbb{T}$, (\ref{eq_conditioncfun}) implies that for any fixed $\lambda_3$ there exists a dense set of $\lambda_2$ in $\{2 \lambda_3 > \lambda_2\}$ which allow $\theta_1=\theta_1(\lambda)$ to be $\alpha$-rational.
\end{proof}

\subsection{Proof of Theorem \ref{thm_point}} \label{sec_proofofthmpoint} 

Assume the claim was false, i.e. for some non $\alpha$-rational $\theta$, the operator $H_{\theta; \lambda, \alpha}$ had an eigenvalue $E$. For $n \in \mathbb{N}$, write $d_{\sigma(\lambda)}^{(n)}(x) := \prod_{j=0}^{n-1} \abs{c_{\sigma(\lambda)}(x+j\alpha)}$ and introduce
\begin{eqnarray} \label{eq_psi}
\Psi_{\sigma(\lambda)}^{(n)}(x) & := &\mathrm{tr}\left\{d_{\sigma(\lambda)}^{(n)}(x) \left(B_{\sigma(\lambda);n}^{E/\lambda_2}(x) - R_\theta^n\right)\right\}  \\
& = & \mathrm{tr}\left(d_{\sigma(\lambda)}^{(n)}(x) B_{\sigma(\lambda);n}^{E/\lambda_2}(x) \right)- 2 d_{\sigma(\lambda)}^{(n)}(x) \cos(2 \pi n \theta) ~\mbox{,}
\end{eqnarray}
in analogy to (\ref{eq_AMO_2}). Then, Proposition \ref{prop_det} implies that for a.e. $x \in \mathbb{T}$, one has
\begin{eqnarray}
\abs{\Psi_{\sigma(\lambda)}^{(n)}(x)}&\leq &\dfrac{2 \abs{d_{\sigma(\lambda)}^{(n)}}}{\abs{\det M_\theta(x)}} \norm{M_\theta(x + \alpha n) - M_\theta(x)} \norm{M_\theta(x)} \\
& \leq & \dfrac{ 2 \Vert c \Vert_\mathbb{T} \abs{d_{\sigma(\lambda)}^{(n)}(x)}}{b} \norm{M_\theta(x + \alpha n) - M_\theta(x)} \norm{M_\theta(x)} ~\mbox{.}
\end{eqnarray}

The appearance of $d_{\sigma(\lambda)}^{(n)}(x)$ complicates matters enough to require the results of Section \ref{sec_rate}. Indeed the growth of the quasi-periodic product $d_{\sigma(\lambda)}^{(n)}(x)$ is controlled by Theorem \ref{prop_prozero}, which guarantees that there exists $C>0$ a subsequence $q_{n_l}$ such that
\begin{equation} \label{eq_16}
\abs{\Psi_{\sigma(\lambda)}^{(q_{n_l})}(x)} \leq C \mathrm{e}^{q_{n_l} I(\sigma(\lambda))} \norm{M_\theta(x + \alpha q_{n_l}) - M_\theta(x)} \norm{M_\theta(x)} ~\mbox{,}
\end{equation}
for a.e. $x \in \mathbb{T}$. Here, we let 
\begin{equation}
I(\lambda) : = \int_\mathbb{T} \log\abs{c_\lambda(x)} \ud \mu(x) ~\mbox{.}
\end{equation}
Set
\begin{equation} \label{eq_vee}
a \vee b := \max\{a,b\} ~\mbox{for $a,b \in \mathbb{R}$} ~\mbox{.}
\end{equation} 
In \cite{JitomirskayaKosloverSchulteis_2005}, the integral $I(\lambda)$ is explicitly computed, which, for $\lambda \in \mathcal{SD}$, gives
\begin{equation} \label{eq_integral}
I(\lambda) = \begin{cases} \log \abs{\lambda_{3} \vee \lambda_1} & \mbox{, if } \lambda \in \mathrm{III} ~\mbox{,} \\
\log \left \vert \dfrac{2\lambda_{1}\lambda_{3}}{1- \sqrt{1 - 4\lambda_{1}\lambda_{3}} } \right \vert & \mbox{, if} ~\lambda \in  \mathrm{L}_\mathrm{II} ~\mbox{and} ~ \lambda_{1},\lambda_{3} \neq 0 ~\mbox{,}\\
0 & \mbox{, if} ~\lambda \in \mathrm{L}_\mathrm{II} ~\mbox{,} ~ \lambda_{1} ~\mbox{or}  ~\lambda_{3} = 0 ~\mbox{.} \end{cases}
\end{equation}

Application of Cauchy-Schwarz in (\ref{eq_16}) finally yields
\begin{equation} \label{eq_5}
\Vert \Psi_{\sigma(\lambda)}^{(q_{n_l})} \Vert_{L^1(\mathbb{T})} \leq C \mathrm{e}^{q_{n_l} I(\sigma(\lambda))} \Vert \norm{M_\theta(. + \alpha q_{n_l}) - M_\theta(.)}\Vert_{L^2(\mathbb{T})} \Vert \norm{M_\theta(.)} \Vert_{L^2(\mathbb{T})} ~\mbox{.} 
\end{equation}
as $l \to  \infty$. In particular, since $\vert \vert \vert q_{n_l} \alpha \vert \vert \vert \to 0$, (\ref{eq_5}) implies
\begin{equation} \label{eq_12}
\Vert \Psi_{\sigma(\lambda)}^{(q_{n_l})} \Vert_{L^1(\mathbb{T})} \leq C_l  \mathrm{e}^{q_{n_l} I(\sigma(\lambda))} ~\mbox{, } C_l = o(1) ~\mbox{, as $l \to  \infty$.}
\end{equation}

For later purposes, we note that Theorem \ref{prop_prozero} and (\ref{eq_12}) also holds along the sequence $(q_{n_l}) \cup (2 q_{n_l}) \cup (3 q_{n_l})$; here, given two sequences $(x_n)$ and $(y_n)$, we define their concatenation by $(x_n) \cup (y_n):=(x_1, y_1, x_2, y_2, \dots)$. 
 
On the other hand, notice that $\Psi_{\sigma(\lambda)}^{(n)}$ is a trigonometric polynomial of degree $n$. Similar to the critical almost Mathieu operator, we will explicitly compute the boundary Fourier-coefficients $\widehat{\Psi_{\sigma(\lambda)}^{(n)}}(\pm n)$ and show that their decay rate contradicts (\ref{eq_12}). To simplify notation, set $\phi_\lambda^{(n)}:= \mathrm{tr}\left(  d_\lambda^{(n)} B_{\lambda;n}^{E/\lambda_2} \right)$. Then,
\begin{equation} \label{eq_9}
\widehat{\Psi_{\sigma(\lambda)}^{(n)}}(\pm n) = \widehat{\phi_{\sigma(\lambda)}^{(n)}}(\pm n) - 2 \cos(2 \pi n \theta) \cdot \begin{cases}
\left(\dfrac{\lambda_1}{\lambda_2}\right)^n \mathrm{e}^{\pi i \alpha n^2} & \mbox{, for } +n ~\mbox{,} \\  \left(\dfrac{\lambda_3}{\lambda_2}\right)^n \mathrm{e}^{- \pi i \alpha n^2} & \mbox{, for } -n ~\mbox{.} \end{cases} 
\end{equation}

It is well known that $\phi_\lambda^{(n)}$ is related to finite cut offs of the original Jacobi operator (\ref{eq_hamiltonian}). Indeed, let $\Pi_{[0,n]}$ be the orthogonal projection in $l^2(\mathbb{Z})$ onto $\mathrm{Span}\{\delta_k ~\mbox{, } 0 \leq k \leq n\}$ and set 
\begin{eqnarray} \label{eq_polyncutoff}
P_\lambda^{(n)}(E; x):=\det \left(E- \Pi_{[0,n-1]} H_{x; \lambda,\alpha} \Pi_{[0,n-1]}\right) ~\mbox{,} ~n \geq 1 ~\mbox{,} \\
P_\lambda^{(0)}(E; x):= 1~\mbox{,} ~P_\lambda^{(-1)}(E;x):=0 ~\mbox{.}
\end{eqnarray}
Then, for $x \in \mathbb{T}_0(\lambda)$, one has
\begin{equation}
d_{\lambda}^{(n)}(x) B_{\lambda;n}^E(x) = \begin{pmatrix} P_\lambda^{(n)}(E; x) & - \overline{c_\lambda(x-\alpha)} P_\lambda^{(n-1)}(E;x+\alpha) \\ c_\lambda(x + (n-1)\alpha) P_\lambda^{(n-1)}(E; x) & -c_\lambda(x+(n-1)\alpha) \overline{c_\lambda(x-\alpha)} P_\lambda^{(n-2)}(E; x + \alpha)  \end{pmatrix} ~\mbox{.}
\end{equation}

In particular, this allows to express $\widehat{\phi_{\sigma(\lambda)}^{(n)}}(\pm n)$ as 
\begin{equation} \label{eq_6}
\widehat{\phi_{\sigma(\lambda)}^{(n)}}(\pm n) = \widehat{P_{\sigma(\lambda)}^{(n)}}(\pm n) - \dfrac{\lambda_1 \lambda_3}{\lambda_2^2} \mathrm{e}^{\pm 2 \pi i (2n-3) \alpha} \widehat{P_{\sigma(\lambda)}^{(n-2)}}(\pm(n-2)) ~\mbox{.}
\end{equation}
The problem is thus reduced to computing $\widehat{P_{\sigma(\lambda)}^{(n)}}(\pm n)$. A first simplifcation is achieved by the following Lemma:
\begin{lemma} \label{lem_1}
Let $\tilde{\lambda}_1 = \frac{\lambda_1}{\lambda_2} \mathrm{e}^{i \pi \alpha}$ and $\tilde{\lambda}_3= \frac{\lambda_3}{\lambda_2} \mathrm{e}^{-i \pi \alpha}$, then
\begin{equation}
(-1)^n \widehat{P_{\sigma(\lambda)}^{(n)}}(\pm n) = \mathrm{e}^{\pm \pi i \alpha n (n-1)} \det(T_n) ~\mbox{,}
\end{equation}
where $T_n$ is a tridiagonal $n \times n$-matrix defined by
\begin{equation}
T_n := \begin{pmatrix} 1& \tilde{\lambda}_1 & & & \\ \tilde{\lambda}_3 & 1 & \tilde{\lambda}_1 & & \\ & \tilde{\lambda}_3 & 1& \tilde{\lambda}_1 & \\ & & \ddots & \ddots & \ddots     \end{pmatrix}
\end{equation}
\end{lemma}
\begin{proof}
We show the argument for the boundary coefficient $+n$; $-n$ is dealt with analogously. The claim becomes obvious when rewriting $\widehat{P_{\sigma(\lambda)}^{(n)}}(n)$ in terms of $\tilde{\lambda}_{1}$ and $\tilde{\lambda}_3$, since then
\begin{equation}
\widehat{P_{\sigma(\lambda)}^{(n)}}(n) = \det \begin{pmatrix} 1 & \tilde{\lambda}_1 & & & \\ \tilde{\lambda}_3 \mathrm{e}^{2 \pi i \alpha} & \mathrm{e}^{2 \pi i \alpha} & \tilde{\lambda}_1 \mathrm{e}^{2 \pi i \alpha} & & \\ & \tilde{\lambda}_3 \mathrm{e}^{4 \pi i \alpha} & \mathrm{e}^{4 \pi i \alpha} & \tilde{\lambda}_1 \mathrm{e}^{4 \pi i \alpha} & \\ & & \ddots & \ddots & \ddots        \end{pmatrix} ~\mbox{.}
\end{equation} 
\end{proof}

Setting $t_n:=\det(T_n)$ and employing Lemma \ref{lem_1}, (\ref{eq_6}) yields
\begin{equation}
\widehat{\phi_{\sigma(\lambda)}^{(n)}}(\pm n) = (-1)^n \mathrm{e}^{\pm \pi i \alpha n (n-1)} \left\{ t_n - \dfrac{\lambda_1 \lambda_3}{\lambda_2^2}  t_{n-2} \right\} ~\mbox{.}
\end{equation}
The simple form of the matrices $T_n$ allows to compute $t_n$. Expanding $T_n$ with respect to its last row, $(t_n)$ satisfies the following second order finite difference equation 
\begin{equation} \label{eq_7}
t_n = t_{n-1} - \gamma t_{n-2} ~\mbox{,} ~ n\geq 1 ~\mbox{,}
\end{equation}
subject to the initial conditions $t_0 = 1$ and $t_{-1} = 0$. Here, for ease of notation, we write $\gamma:=\dfrac{\lambda_1 \lambda_3}{\lambda_2^2}$. 

Solving (\ref{eq_7}), we obtain
\begin{equation}
t_n =  \begin{cases} \frac{1}{\lambda_{+} - \lambda_{-}} (\lambda_{+}^{n+1} - \lambda_{-}^{n+1}) & \mbox{, if} ~\gamma \neq \frac{1}{4} ~\mbox{,} \\
(n+1) \left(\frac{1}{2}\right)^n & \mbox{, if} ~\gamma=\frac{1}{4}  ~\mbox{,} \end{cases} 
\end{equation}
where 
\begin{equation} \label{eq_10}
\lambda_{\pm} = \frac{1}{2} \left( 1 \pm \sqrt{1 - 4 \gamma} \right) ~\mbox{.}
\end{equation}
Finally this gives rise to the following closed expression for $\widehat{\phi_{\sigma(\lambda)}^{(n)}}(\pm n)$,
\begin{equation} \label{eq_8}
(-1)^n \widehat{\phi_{\sigma(\lambda)}^{(n)}}(\pm n) = \mathrm{e}^{\pm \pi i \alpha n (n-1)} \begin{cases} \frac{1}{\lambda_{+} - \lambda_{-}} \left[ \left(\lambda_{+}^{n+1} - \lambda_{-}^{n+1} \right)
- \gamma \left( \lambda_{+}^{n-1} - \lambda_{-}^{n-1} \right)   \right] & \mbox{, if} ~\gamma \neq \frac{1}{4} ~\mbox{,} \\
\left(\frac{1}{2}\right)^{n-1} & \mbox{, if} ~\gamma=\frac{1}{4}  ~\mbox{.} \end{cases} 
\end{equation}

Equations (\ref{eq_9}) and (\ref{eq_8}) allow to analyze the sequences $(\widehat{\Psi_{\sigma(\lambda)}^{(n)}}(\pm n))_{n \in \mathbb{N}}$. In view of that, we set $(m_l):=(q_{n_l}) \cup (2 q_{n_l}) \cup (3 q_{n_l})$. Without loss of generality, we may assume $\lambda_1 \vee \lambda_3 = \lambda_1$ \footnote{If $\lambda_1 \vee \lambda_3 = \lambda_3$, consider $\Psi_{\sigma(\lambda)}^{(n)}(-n)$ instead of $\Psi_{\sigma(\lambda)}^{(n)}(n)$ in the proof of Proposition \ref{prop_lb}.}.

Using (\ref{eq_6}) and (\ref{eq_8}),  one obtains
\begin{eqnarray}  \label{eq_11}
\left \vert  \widehat{\Psi_{\sigma(\lambda)}^{(n)}}(n) \right \vert =  2 \left \vert \dfrac{\lambda_1}{\lambda_2}\right\vert^n \left \vert - \cos(2 \pi n \theta) 
+ (-1)^n \dfrac{\mathrm{e}^{-i \pi \alpha n}}{\lambda_{+} - \lambda_{-}} \dfrac{\lambda_{+}^n}{2 \left(\frac{\lambda_1}{\lambda_2}\right)^n} \lambda_{+} \right. \nonumber \\ 
\quad \left. \times \left[ \left(1 - \left(\dfrac{\lambda_{-}}{\lambda_+} \right)^{n+1}   \right) - \dfrac{\gamma}{\lambda_{+}^2} \left(1 - \left(\dfrac{\lambda_-}{\lambda_+}\right)^{n-1}          \right)             \right]       \right\vert  ~\mbox{, if $\gamma \neq \frac{1}{4}$ ,}
\end{eqnarray}
and
\begin{eqnarray} \label{eq_13}
\left \vert  \widehat{\Psi_{\sigma(\lambda)}^{(n)}}(n) \right \vert & = & 2 \left \vert \dfrac{\lambda_1}{\lambda_2}\right\vert^n \left\vert (-1)^n \mathrm{e}^{-i \pi \alpha n} \dfrac{\left(\frac{1}{2}\right)^{n-1}}{2 \left(\frac{\lambda_1}{\lambda_2}\right)^n} - \cos(2 \pi n \theta) \right\vert  ~\mbox{, if $\gamma = \frac{1}{4}$ .} 
\end{eqnarray}

\begin{prop} \label{prop_lb}
Let $\alpha$ irrational and $\lambda \in \mathcal{SD}$. For a.e. $\theta$, 
\begin{eqnarray} 
\limsup_{l \to \infty} \mathrm{e}^{- m_l I(\sigma(\lambda))} \left \vert  \widehat{\Psi_{\sigma(\lambda)}^{(m_l)}}(m_l) \right \vert > 0 ~\mbox{.} \label{eq_19a}
\end{eqnarray}
\end{prop}
\begin{remark}
The proof below shows that Proposition \ref{prop_lb} holds {\em{for all}} $\theta$ if $\lambda_1 \neq \lambda_3$.
\end{remark}

\begin{proof}
We consider separately the two situations, $\lambda \in \mathrm{III}$ and $\lambda \in \mathrm{L}_{\mathrm{II}}$. 

In both cases, the following observation will be of use: As shown above, the expression for $\widehat{\Psi_{\sigma(\lambda)}^{(k)}}(k)$ contains a term of the form $\mathrm{e}^{-i \pi \alpha k} (-1)^{k}$. As we are only interested in asymptotic behavior (following indicated by ``$\sim$''), employing (\ref{eq_contifracbasic}) yields
\begin{equation} \label{eq_asyexprimp}
\mathrm{e}^{-i \pi \alpha j q_n} (-1)^{j q_n} \sim (-1)^{j (p_n + q_n)} ~\mbox{, $j \in \mathbb{N}$ ,}
\end{equation}
which, for fixed $j \in \mathbb{N}$, produces a constant sign upon passing to a subsequence of $(q_n)$ where $(p_n + q_n)$ has constant parity. From here on, we shall thus assume $(q_{n_l})$ to be a fixed subsequence of $(q_n)$ such that the conclusion of Theorem \ref{prop_prozero} holds and that $(p_{n_l} + q_{n_l})$ has {\em{constant parity}}. Following, denote by $\mathfrak{p}$ this (constant) parity of $(p_{n_l} + q_{n_l})$.

\begin{description}
\item[Case I, $\lambda \in \mathrm{III}$]  Since $\lambda_1 \vee \lambda_3 = \lambda_1$, (\ref{eq_integral}) implies $I(\sigma(\lambda)) = \log\left( \frac{\lambda_1}{\lambda_2}\right)$. 
Suggested by (\ref{eq_8}), we distinguish the following three cases for $\gamma$:
\begin{description}
\item[(a) $0 \leq \gamma < \frac{1}{4}$]
In this case $\lambda_{\pm}$ in (\ref{eq_10}) are real positive and distinct. Moreover, $\lambda_1 + \lambda_3 \geq \lambda_2$ (and $\lambda_1 \vee \lambda_3 = \lambda_1$) implies that $\lambda_{+} \leq \frac{\lambda_1}{\lambda_2}$ with equality if and only if $\lambda_1 + \lambda_3 = \lambda_2$.

Upon use of (\ref{eq_asyexprimp}), for any fixed $j \in \mathbb{N}$, (\ref{eq_11}) reduces to 
\begin{equation} \label{eq_11aa}
\left \vert  \widehat{\Psi_{\sigma(\lambda)}^{(j q_{n_l})}}(j q_{n_l}) \right \vert \sim 2 \left\vert \dfrac{\lambda_1}{\lambda_2} \right\vert^{j q_{n_l}} \left\vert -\cos(2 \pi j q_{n_l} \theta) + (-1)^{j (p_{n_l} + q_{n_l})} A \right\vert ~\mbox{,}
\end{equation}
where A=0, if $\lambda_1 + \lambda_3 > \lambda_2$, and $A=\frac{1}{2}$, if $\lambda_1 + \lambda_3 = \lambda_2$.

We first consider the situation when $\lambda_1 + \lambda_3 = \lambda_2$, which by (\ref{eq_11aa}) depends on $\mathfrak{p}$. 

For odd $\mathfrak{p}$, the claim of the theorem would follow directly for $\widehat{\Psi_{\sigma(\lambda)}^{(q_{n_l})}}(q_{n_l})$ if one could ensure that
\begin{equation}
\limsup_{l \to \infty} \vert - \cos(2 \pi q_{n_l} \theta) - \frac{1}{2} \vert > 0 ~\mbox{,}
\end{equation}
which, however, will not be true for general $\theta$. 

Making use of (\ref{eq_11aa}) for $j=3$, this may easily be mended, replacing $q_{n_l}$ by $3 q_{n_l}$ whenever $l$ is such that $\cos(2 \pi q_{n_l} \theta) \approx -\frac{1}{2}$, in which case it is guaranteed that $\cos(2 \pi (3 q_{n_l}) \theta) \not\approx -\frac{1}{2}$. Referring to (\ref{eq_11aa}), the same strategy also works if $\mathfrak{p}$ is even.

For $\lambda_1 + \lambda_3 > \lambda_2$, a similar argument can be used to conclude the claim of the theorem; we mention that based on (\ref{eq_11aa}) with $A=0$, the argument is independent of $\mathfrak{p}$ and it is enough to consider the sequence $(q_{n_l}) \cup (2 q_{n_l})$.

\item[(b) $\gamma = \frac{1}{4}$] From $\lambda_1 + \lambda_3 \geq \lambda_2$, we conclude that $\frac{\lambda_1}{\lambda_2} \geq \frac{1}{2}$ and $\frac{\lambda_3}{\lambda_2} \leq \frac{1}{2}$, where equality holds if and only if $\lambda_1 = \lambda_3$. 
Referrring to (\ref{eq_13}), if $\lambda_1/\lambda_2 > 1/2$, the claim (\ref{eq_19a}) follows for $(q_{n_l}) \cup (2 q_{n_l})$, thereby taking care of instances $l$ when $\cos(2 \pi q_{n_l} \theta) \approx 0$.

If $\lambda_1 /\lambda_2 = 1/2$, one has
\begin{equation} \label{eq_13b}
\left \vert  \widehat{\Psi_{\sigma(\lambda)}^{(j q_{n_l})}}(j q_{n_l}) \right \vert = 2 \left(\frac{\lambda_1}{\lambda_2}\right)^{j q_{n_l}} \left\vert -\cos(2 \pi j q_{n_l} \theta) + (-1)^{j (p_{n_l} + q_{n_l})} \right\vert ~\mbox{.}
\end{equation}

For even $\mathfrak{p}$, the sign in (\ref{eq_13b}) is constant in $j$. We note however that above strategy of replacing $q_{n_l}$ by $j q_{n_l}$ does not work for {\em{any}} $j$ since the {\em{expanding map}} of degree $j$, $E_j: \mathbb{T} \to \mathbb{T}$, $E_j(x) = j x (\mathrm{mod} ~1)$, has a fixed point at zero.

We address this problem in Sec. \ref{sec_rationalapprox} where Proposition \ref{lem_almostunique} shows that at least for $\mu$-a.e. $\theta$ one has
\begin{equation} \label{eq_13a}
\limsup_{l \to \infty} \dfrac{ \left \vert  \widehat{\Psi_{\sigma(\lambda)}^{(q_{n_l})}}(q_{n_l}) \right \vert   }{2 \left(\frac{\lambda_1}{\lambda_2}\right)^{q_{n_l}}} \gtrsim \limsup_{l \to \infty} \left\vert \cos(2 \pi q_{n_l} \theta) - 1 \right\vert > 0 ~\mbox{.}
\end{equation}

The case when $\mathfrak{p}$ odd is reduced to a problem analogous to (\ref{eq_13a}) by considering $\left \vert  \widehat{\Psi_{\sigma(\lambda)}^{(2 q_{n_l})}}(2 q_{n_l}) \right \vert$ instead of  $\left \vert  \widehat{\Psi_{\sigma(\lambda)}^{(q_{n_l})}}(q_{n_l}) \right \vert$.

\item[(c) $\gamma > \frac{1}{4}$] Then, $\lambda_+ = \overline{\lambda_-}$ and $\abs{\lambda_+} = \sqrt{\gamma}$. In particular, $\abs{\lambda_+} \leq \frac{\lambda_1}{\lambda_2}$ with equality if and only if $\lambda_1 = \lambda_3$. Hence, using (\ref{eq_11}) for $\lambda_1 \neq \lambda_3$, the claim follows for $(q_{n_l}) \cup (2 q_{n_l})$ and {\em{every}} $\theta$.
  
If $\lambda_1 = \lambda_3$, the right hand side of (\ref{eq_11}) additionally depends on $\phi:=\frac{1}{2 \pi} \arg(\lambda_+)=\frac{1}{2 \pi } \arctan(\sqrt{4\gamma-1})$. Referring to (\ref{eq_11}), we set 
\begin{equation}
\left \vert  \widehat{\Psi_{\sigma(\lambda)}^{(n)}}(n) \right \vert  =: 2 \left \vert \dfrac{\lambda_1}{\lambda_2}\right\vert^n \left \vert - \cos(2 \pi n \theta) + (-1)^n \dfrac{\mathrm{e}^{-i \pi \alpha n}}{2} A_n \right\vert ~\mbox{.} 
\end{equation}

A computation verifies that $A_n$ is purely real with $A_n = 2 \cos(2 \pi \phi n)$. Therefore, 
\begin{equation} \label{eq_13c}
\dfrac{ \left \vert  \widehat{\Psi_{\sigma(\lambda)}^{(j q_{n_l} )}}(j q_{n_l}) \right \vert   }{2 \left(\frac{\lambda_1}{\lambda_2}\right)^{j q_{n_l}}} \gtrsim \left\vert -\cos(2 \pi j q_{n_l} \theta) + (-1)^{j (p_{n_l} + q_{n_l})} \cos(2 \pi j q_{n_l} \phi)  \right\vert ~\mbox{.}
\end{equation}
As the right hand side of (\ref{eq_13c}) now requires control of two cosines oscillating at, in general, unrelated frequencies, the simple argument relying on properties of the expanding map will not be of use.

For even $\mathfrak{p}$, the sign on the right hand side of (\ref{eq_13c}) is independent of $j$, whence the claim reduces to
\begin{equation} \label{eq_13d}
\limsup_{l\to\infty} \left\vert -\cos(2 \pi q_{n_l} \theta) + \cos(2 \pi q_{n_l} \phi) \right\vert \stackrel{?}{>} 0 ~\mbox{.}
\end{equation}
Even though (\ref{eq_13d}) will not be true for all $\theta$, the problem may again be formulated in a form that allows application of Proposition \ref{lem_almostunique}, thus implying (\ref{eq_13d}) for $\mu$-a.e. $\theta$.

To this end, first assume by possibly passing to an appropriate subsequence, that both $(q_{n_l} \theta)$ and $(q_{n_l} \phi)$ converge. Then, if (\ref{eq_13d}) fails, the set \begin{equation}
\Omega_\phi:=\{\theta \in \mathbb{T}: q_{n_l} (\theta \pm \phi) \to 0   \} ~\mbox{,}
\end{equation}
will be non-empty, which however is of $\mu$-measure zero by Proposition \ref{lem_almostunique}.

The case when $\mathfrak{p}$ is odd leads to the same type of problem as (\ref{eq_13d}), when replacing $(q_{n_l})$ by $(2 q_{n_l})$.

\end{description}
\end{description}

\begin{description}
\item[Case II, $\lambda \in \mathrm{L}_{\mathrm{II}}$] In particular then, $\lambda_2 = 1$. First, notice that for $\gamma = 0$, (\ref{eq_integral}) implies that $I(\lambda) = 0$. Moreover, from Lemma \ref{lem_1}, we have
\begin{equation}
\left \vert  \widehat{\Psi_{\sigma(\lambda)}^{(n)}}(n) \right \vert = \left \vert  \widehat{\phi_{\sigma(\lambda)}^{(n)}}(n) \right \vert = 1 ~\mbox{,}
\end{equation}
which, in summary, already implies (\ref{eq_19a}).

For $\gamma \neq 0 $, rewriting (\ref{eq_integral}) in terms of the relevant parameter $\gamma$ yields
\begin{equation}
I(\lambda) = \log \left \vert \dfrac{\gamma}{\lambda_{-}} \right \vert ~\mbox{.}
\end{equation}

We again distinguish three cases.
\begin{description}
\item[(a) $0 < \gamma < \frac{1}{4}$]
Making use of (\ref{eq_11}), 
\begin{equation} \label{eq_15}
\left\vert  \dfrac{ \widehat{\Psi_{\sigma(\lambda)}^{(q_n)}}(q_n)}{(\gamma/\lambda_-)^{q_n}   }  \right\vert \geq \left\vert -2 \cos(2 \pi q_n \theta) \left(\dfrac{ \lambda_- \lambda_1  }{ \gamma  }\right)^{q_n} + (-1)^{q_n} \mathrm{e}^{- i \pi \alpha q_n}   \right\vert ~\mbox{.}
\end{equation}
We note that $\frac{\lambda_- \lambda_1}{\gamma} \geq 0$ since $\lambda_1 \vee \lambda_3 = \lambda_1$ and $\lambda_1 + \lambda_3 \geq 0$.

For $\frac{ \lambda_- \lambda_1   }{ \gamma  } \neq 1$, (\ref{eq_15}) immediately implies
\begin{equation}
\limsup_{n\to\infty} \left\vert  \dfrac{ \widehat{\Psi_{\sigma(\lambda)}^{(q_n)}}(q_n)}{(\gamma/\lambda_-)^{q_n}   }  \right\vert \geq 1 ~\mbox{.}
\end{equation}

The case $ \frac{ \lambda_- \lambda_1   }{ \gamma  } = 1$ implies $\lambda_1 + \lambda_3 =1$, thus has already been dealt with in region III.

\item[(b) $\gamma = \frac{1}{4}$]
Here, $I(\lambda) = \log \left( \frac{1}{2} \right)$ by (\ref{eq_integral}). Thus, using (\ref{eq_13}), the claim follows immediately for $\lambda_1 > 1/2$, 
\begin{equation}
\limsup_{n\to\infty} \left\vert  \dfrac{ \widehat{\Psi_{\sigma(\lambda)}^{(q_n)}}(q_n)}{ \left( \frac{1}{2} \right)^{q_n -1}}  \right\vert \geq 1 ~\mbox{.}
\end{equation}

If, on the other hand, $\lambda_1=1/2$, one obtains the same expression as in (\ref{eq_13b}), whence can proceed as then.
\item[(c) $\gamma > \frac{1}{4}$] 
Since $\abs{\lambda_{\pm}} = \sqrt{\gamma}$,
\begin{equation}
\left \vert \frac{\lambda_1 \lambda_-}{\gamma} \right \vert = \left( \frac{\lambda_1}{\lambda_3}  \right)^{1/2} ~\mbox{,}
\end{equation}
whence for $\lambda_1 \neq \lambda_3$ 
\begin{equation}
\limsup_{n \to \infty} \left \vert \frac{\lambda_-}{\gamma} \right \vert^{q_n} \left \vert  \widehat{\Psi_{\sigma(\lambda)}^{(q_n)}}(q_n) \right \vert \geq 1 ~\mbox{.}
\end{equation}
Notice that for $\lambda_1 = \lambda_3$, $2 \lambda_3 \leq 1$ implies $\gamma \leq \frac{1}{4}$.
\end{description}
\end{description}
\end{proof}

Since Proposition \ref{prop_lb} contradicts the asymptotics of $\Vert \Psi_{\sigma(\lambda)}^{(n_l)} \Vert_{L^1(\mathbb{T})}$ given in (\ref{eq_12}), we have proven Theorem \ref{thm_point}.

\subsection{Almost uniqueness in rational approximation} \label{sec_rationalapprox}

An important ingredient in the proof of Proposition \ref{prop_lb} (specifically, (\ref{eq_13a}) in Case I(b) and (\ref{eq_13d}) in Case I(c) of Sec. \ref{sec_proofofthmpoint})  were conclusions of the form:
\begin{equation} \label{eq_motivrational}
\mu \left( \left\{ \theta \in \mathbb{T}: q_{n_l} \theta \to 0 \right\} \right) = 0 ~\mbox{.}
\end{equation}
Here, $(q_{n_l})$ was a certain subsequence of the sequence of denominators $(q_n)$ in the continued fraction expansion of $\alpha$, which in particular implies that $q_{n_l} \alpha \to 0$. The purpose of this section is to prove statements of the form (\ref{eq_motivrational}).

To this end, let $\theta \in \mathbb{T}$ be irrational. We call a sequence $(k_n)$ of natural numbers a {\em{sequence of denominators approximating $\theta$}} if $\vert \vert \vert k_n \theta \vert \vert \vert \to 0$, as $n \to \infty$. Necessarily, $\theta \in \mathbb{R}\setminus \mathbb{Q}$ implies $k_n \to \infty$. Given $(k_n)$, let $\Omega(k_n)$ be the set of $\theta \in \mathbb{T}$ such that $(k_n)$ forms a sequence of denominators approximating $\theta$. 

The following proposition asserts ``almost - uniqueness'' of the approximated number for a given sequence of denominators:
\begin{prop} \label{lem_almostunique}
Let $(k_n)$ be a sequence in $\mathbb{N}$, then
\begin{equation}
\mu(\Omega(k_n))= 0 ~\mbox{.}
\end{equation}
\end{prop}
\begin{remark}
\begin{itemize}
\item[(i)] Considering the degree $N$ expanding map $E_N: \mathbb{T} \to \mathbb{T}$, $x \mapsto n ~x (\mathrm{mod} ~1)$, one concludes
\begin{equation} \label{eq_invexp}
E_N(\Omega(k_n)) \subseteq \Omega(k_n) ~\mbox{.}
\end{equation}
In particular, $\theta_0 \in \Omega(k_n)$ implies the same holds true for any $\theta$ with $\left( \mathbb{Z} \theta_0 + \theta \right) \cap \mathbb{Z} \neq \emptyset$. Notice however that $\Omega$ is {\em{not}} invariant under $E_N$.
\item[(ii)] It is easy to see that $\Omega(k_n)$ is in general {\em{uncountable}}. Indeed, suppose $k_n=10^{l_n}$ with $l_n \in \mathbb{N}$ such that $l_{n+1} - l_n \geq n+1$. Any $\theta \in [0,1)$ whose decimal expansion $0.a_1 a_2 \dots$ satisfies $a_{l_{n+j}} = 0$, $1\leq j \leq n$ and all $n \in \mathbb{N}$, yields $\vertiii{ k_n \theta } \leq 10^{-n} \to 0$. Obviously, the set of such $\theta$ is uncountable.
\end{itemize}
\end{remark}
\begin{proof}\footnote{An alternative argument would be to observe that $\Omega(k_n)$ is a proper subgroup of $\mathbb{T}$, whence $\mu(\Omega(k_n)) = 0$ by problem 14 in Sec. 1 of Katzelson's classic book on Harmonic Analysis \cite{Katznelson_book2004}.} .

Set $A_k^\epsilon =\{\theta : \vert\vert \vert k\theta \vert\vert \vert <\epsilon.\} $ For any $k\in \mathbb{N},$ $\vert A_k^\epsilon\vert \leq 2\epsilon$. Since for every $\epsilon,$ $\Omega (k_n)\subset A_{k_{n(\epsilon})}^\epsilon,$ the result follows. 
\end{proof}

\section{The theorem of Avila, Fayad, and Krikorian for Jacobi operators} \label{sec_afk}

Purpose of this section is to prove Theorem \ref{thm_coroafk} for a {\em{non-singular}} quasi-periodic, analytic Jacobi operator. For Schr\"odinger operators ($c \equiv 1$), the theorem is an immediate consequence of \cite{AvilaFayadKrikorian_2011}, where the following dichotomy is proven for Lebesgue a.e. $E \in \mathbb{R}$: either the $SL(2,\mathbb{R})$-cocycle $(\alpha, A^E)$ satisfies $L(\alpha, A^E) > 0$ or it is analytically conjugate to a real, not necessarily constant rotation. In this section, we comment on extending this statement to non-singular Jacobi operators. 

\subsection{Reductions} \label{sec_afk_reductions}
We first recall some definitions. Let $\mathcal{Y}$ stand for either $L^p$, $ p \geq 1$, or $\mathcal{C}^{\omega}$ (analytic category), and $\mathcal{M}$ for one of $M_2(\mathbb{C})$, $SL(2,\mathbb{C})$ or $SL(2,\mathbb{R})$. Fixing $\alpha$ irrational, two cocycles $(\alpha, A)$ and $(\alpha, D)$, $A, D \in \mathcal{Y}(\mathbb{T},\mathcal{M})$, are $\mathcal{Y}$-{\em{conjugate}} over $\mathcal{M}$ if for some $C \in \mathcal{Y}(\mathbb{R}/2\mathbb{Z},\mathcal{M})$ \footnote{In view of the case $\mathcal{Y}=\mathcal{C}^\omega$, we only require the mediating change of coordinates $C(x)$ in (\ref{eq_defconj}) to be two- instead of one-periodic.} with 
$\log\abs{\det{C}} \in L^1(\mathbb{R}/2\mathbb{Z})$ one has
\begin{equation} \label{eq_defconj}
C(. + \alpha)^{-1} A(.) C(.) = D(.) ~\mbox{, in $\mathcal{Y}$.}
\end{equation}
Clearly, $L(\alpha,D) = L(\alpha,A)$.

\begin{definition} \label{def_rotred}
For $\mathcal{M}=SL(2,\mathbb{R})$, we call $(\alpha,A)$ {\em{$\mathcal{Y}$-reducible}} if it is $\mathcal{Y}$-conjugate over $SL(2, \mathbb{R})$ to a real, {\em{not necessarily constant}} rotation.
\end{definition}

The proof in \cite{AvilaFayadKrikorian_2011} relies on Theorem 1.3 therein, which is not specific to Schr\"odinger cocycles. The strategy is based on a KAM scheme which requires an analytic $SL(2, \mathbb{R})$-cocycle which is homotopic to the identity. We emphasize that the techniques used in \cite{AvilaFayadKrikorian_2011} rely on {\em{real}}-analyticity.

In spite of $A^E(\theta)$ in (\ref{eq_defnjacobico}) in general being $M_2(\mathbb{C})$-valued, for any {\em{non}}-singular quasi-periodic, analytic Jacobi operator one has the following analytic conjugacy
\begin{eqnarray} \label{eq_conjcomplexreal}
C(\theta+\alpha)^{-1} \begin{pmatrix} E - v(\theta) & - \abs{c}(\theta-\alpha) \\ \abs{c}(\theta) & 0 \end{pmatrix} C(\theta)  &=& \begin{pmatrix} E - v(\theta) & -\widetilde{c}(\theta - \alpha) \\ c(\theta) & 0  \end{pmatrix} ~\mbox{,} \nonumber \\
C(\theta) & := & \begin{pmatrix} 1 & 0 \\ 0 & \sqrt{\frac{\widetilde{c}(\theta- \alpha)}{c(\theta-\alpha)}} \end{pmatrix} ~\mbox{,}
\end{eqnarray}
which reduces the problem to a quasi-periodic, analytic Jacobi operator where $c(\theta)$ is {\em{real and positive}} \footnote{The conjugacy (\ref{eq_conjcomplexreal}) is a dynamical formulation of the well-known fact that any Jacobi operator $H_{c,v}$ with underlying sequences $c=(c_n)$ and $v=(v_n)$ is unitarily equivalent to $H_{\vert c \vert, v}$, see e.g. \cite{Teschl_book_2000}, (1.57) and Lemma 1.6, therein.}. 

In the same spirit as $\widetilde{c}$ analytically ``re-interprets'' $\overline{c}$, morally, the function $\abs{c}(\theta):= \sqrt{\widetilde{c}(\theta) c(\theta)} \in \mathcal{C}^\omega(\mathbb{T})$ analytically ``re-interprets'' $\vert c(\theta) \vert$. We note that since $\inf_{\theta \in \mathbb{T}} \vert c(\theta) \vert > 0$, the branch of the square-root appearing in (\ref{eq_conjcomplexreal}) and in the definition of $\abs{c}(\theta)$ can be chosen so that both $\abs{c}(\theta)$ and $\sqrt{\frac{\widetilde{c}(\theta- \alpha)}{c(\theta-\alpha)}}$ are still 1-periodic and holomorphic in a neighborhood of $\mathbb{R}$ (apply e.g. Fact 1 in \cite{JitomirskayaMarx_2013_erratum}).

In particular, we may apply the arguments of \cite{AvilaFayadKrikorian_2011} to the {\em{analytically normalized real  Jacobi-cocycle}} $(\alpha, (A^E)^\sharp)$ defined by
\begin{equation} \label{eq_defApr}
(A^E)^{\sharp}(\theta) := \dfrac{A^E(\theta)}{\sqrt{det A^E(\theta)}} \in \mathcal{C}^\omega(\mathbb{T}, SL(2,\mathbb{R})) ~\mbox{.}
\end{equation}
Note that in the neighborhood of $\epsilon = 0$ where $c(\theta + i \epsilon) \neq 0$ (and thus where (\ref{eq_defApr}) is well-defined), one has
\begin{equation} \label{eq_defcomplexle_1}
L(\alpha, (A_\epsilon^E)^{\sharp}) = L(\alpha, A_\epsilon^E) - \int_\mathbb{T} \log \vert c(\theta) \vert ~\ud \mu(\theta) = L(E; \epsilon) ~\mbox{,}
\end{equation}
This was the dynamical reason, mentioned in the end of Remark \ref{rem_defcomplexLE}, which underlies the definition of the complexified Lyapunov exponent in (\ref{eq_defcomplexle}).

To apply the arguments of \cite{AvilaFayadKrikorian_2011} to the normalized Jacobi cocycle, first notice that $(\alpha, (A^E)^\sharp)$ is homotopic  to the identity in $\mathcal{C}^\omega(\mathbb{T}, SL(2,\mathbb{R}))$: To see this, just consider
\begin{eqnarray}
& H_t(\theta) = \dfrac{1}{\sqrt{c(\theta) c(\theta - t \alpha)}} \begin{pmatrix} t (E - v(\theta))  & -c(\theta - t \alpha) \\ c(\theta) & 0  \end{pmatrix}, ~t \in [0,1] ~\mbox{,}
\end{eqnarray}
which establishes a homotopy of $(\alpha, (A^E)^\sharp)$ to the constant (real) rotation by $\pi/2$ and hence to the identity matrix.

Based on Theorem 1.3 in \cite{AvilaFayadKrikorian_2011}, the authors then argue (Lemma 1.4 and its proof on p.4 of \cite{AvilaFayadKrikorian_2011}) that if $(\alpha, A)$ is $L^2$-reducible, it is  already so analytically. 

Hence, it is left to establish $L^2$-reducibility of $(\alpha, (A^E)^\sharp)$ for Lebesgue a.e. $E$ where $L(\alpha, (A^E)^\sharp) = L(E)= 0$. As in the Schr\"odinger case \cite{DeiftSimon_1983}, this is a consequence of Kotani theory. Assuming a more dynamical point of view, we extend the result in Sec. \ref{m5_app_kotanidyn} below.

In summary, we arrive at the following extension of the result in \cite{AvilaFayadKrikorian_2011} to non-singular Jacobi operators: 
\begin{theorem} \label{m5_thm_afk_jacobi}
Consider a  non-singular quasi-periodic, analytic Jacobi-operator with irrational frequency $\alpha$. For Lebesgue a.e. $E \in \mathbb{R}$: either $L(E) > 0$ or the cocycle $(\alpha, (A^E)^\sharp)$ is analytically reducible to a real, not necessarily constant rotation. 
\end{theorem}
By (\ref{eq_defcomplexle_1}), analytic reducibility of $(\alpha, (A^E)^\sharp)$ implies subcritical behavior, whence Theorem \ref{m5_thm_afk_jacobi} proves Theorem \ref{thm_coroafk}. 

\subsection{A dynamical formulation of Kotani theory} \label{m5_app_kotanidyn}

Following, we consider a fixed non-singular quasi-periodic, analytic Jacobi-operator $H_\theta$ with irrational frequency $\alpha$; in particular, for all $\theta \in \mathbb{T}$, one has
\begin{equation} \label{m5_eq_kotani_assumptions}
0 < m \leq c(\theta) \leq M < +\infty ~\mbox{.}
\end{equation}
The previous section reduced the proof of Theorem \ref{m5_thm_afk_jacobi} to the following claim:
\begin{theorem} \label{m5_thm_kotani}
For Lebesgue a.e. $E \in \mathcal{Z} = \{E^\prime: L(\alpha, (A^{E^\prime})^\sharp) = 0\}$, the cocycle $(\alpha,(A^E)^\sharp)$ is $L^2$-reducible.
\end{theorem}
Recall that Kotani theory shows that $\mathcal{Z}$ forms an essential support of the ac spectrum. Thus, Theorem \ref{m5_thm_kotani} makes rigorous the heuristics that extended states are described in terms of two Bloch waves, $\mathrm{e}^{\pm 2 \pi i \phi(n)}$, propagating in opposite directions.
\begin{remark}
\item[(i)] Theorem \ref{m5_thm_kotani} is a dynamical formulation of a known result for ergodic Schr\"odinger operators proven in \cite{DeiftSimon_1983}, see Sec. 7 therein. Below mentioned proof carries over to ergodic situation as well.
\item[(ii)] Theorem \ref{m5_thm_kotani} can be deduced from the general theory of monotonic cocycles, which has recently been developed in \cite{AvilaKrikorian_2013}. For Jacobi-cocycles the result may however easily be obtained directly, which is what is done below.
\end{remark}

In order to relate iterates of $(\alpha, (A^E)^\sharp)$ to solutions of $H_\theta \psi = E \psi$ induced by $(\alpha, B^E)$, observe that
\begin{equation} \label{m5_eq_kotanicocyl}
B^E(\theta) = \dfrac{\sqrt{c( \theta - \alpha)}}{\sqrt{c(\theta)}} (A^E)^\sharp(\theta) ~\mbox{,}
\end{equation}
which establishes a conjugacy over $M_2(\mathbb{C})$ between $(\alpha,B^E)$ and $(\alpha,(A^E)^\sharp)$.

To prepare the proof of Theorem \ref{m5_thm_kotani}, we first recall some basic facts. As common, let $\mathbb{H}^\pm:=\{z \in \mathbb{C}: \mathrm{sgn} \im(z) = \pm1\}$. 

For $z \in \mathbb{H}^+$,  one defines the {\em{$m$-functions}}
\begin{equation} \label{m4_eq_defnmfunctions}
m_{+}(\theta, z):=-\dfrac{\psi_{+}(1, \theta, z)}{c(\theta) \psi_{+}(0,\theta, z)} ~\mbox{,} ~m_{-}(\theta,z):=-\dfrac{\psi_{-}(-1,\theta, z)}{c(\theta-\alpha)\psi_{-}(0,\theta, z)} ~\mbox{,}
\end{equation}
where $\psi_\pm(. , \theta, z)$ satisfies $H_\theta \psi_\pm(\theta,z) = z \psi_{\pm}(\theta,z)$ with $\psi_\pm(0,\theta,z) = 1$. We note that the solutions $\psi_\pm(. , \theta, z)$ decay  exponentially at respectively $\pm \infty$, are unique, and non-zero for all $n \in \mathbb{Z}$. In particular, for any $k \in \mathbb{Z}$ one has the covariance relations
\begin{equation} \label{m5_eq_covarianceblochw}
\psi_{\pm}(n,\theta+k \alpha,z) = a_\pm(k,\theta,z) \psi_\pm(n+k, \theta,z) ~\mbox{, } \forall n \in \mathbb{Z} ~\mbox{,}
\end{equation}
for some measurable functions $a_\pm(k,\theta,z)$. 

Observe that (\ref{m5_eq_covarianceblochw}) allows to express the solutions $\psi_\pm(.,\theta,z)$ in terms of $m$-functions, 
\begin{eqnarray} \label{m5_eq_limitingsol+}
 & & \quad \psi_+(n, \theta, z) := \begin{cases} (-1)^n \prod_{j=0}^{n-1} c(\theta + j \alpha) m_{+}(\theta + j \alpha, z) & \mbox{, $n > 0$, } \\
                                                                    1 & \mbox{, $n = 0$,} \\
                                                                    (-1)^n \prod_{j=n}^{-1} c(\theta + j \alpha )^{-1} m_{+}(\theta + j \alpha, z)^{-1} & \mbox{, $n < 0$.}
      \end{cases}
\end{eqnarray}
and
\begin{eqnarray} \label{m5_eq_limitingsol-}
 & \psi_-(n, \theta, E)  := \begin{cases}(-1)^n \prod_{j=n+1}^{0} c(\theta+(j-1)\alpha ) m_{-}(\theta+j\alpha, z) & \mbox{, $n < 0$, } \\
                                                                    1 & \mbox{, $n = 0$,} \\
                                                                    (-1)^n \prod_{j=1}^{n} c(\theta + (j-1)\alpha)^{-1} m_{-}(\theta + j\alpha, z)^{-1} & \mbox{, $n > 0$.}
\end{cases} \nonumber \\
\end{eqnarray}

The definitions of the $m$-functions given in (\ref{m4_eq_defnmfunctions}) originate from expressions for the Green's functions of the half-line operators associated with $H_\theta$. In particular, for $z = E+i\epsilon$ and $E \in \mathcal{Z}$, Kotani theory analyzes their boundary values as $\epsilon \to 0+$ :
\begin{theorem}[see e.g. Lemma 5.18 in \cite{Teschl_book_2000}] \label{m5_thm_kotanilemma}
For $\mu$-a.e. $\theta$ and Lebesgue a.e. $E \in \mathcal{Z}$, the limits $m_\pm(\theta, E+ i 0)$ exist and satisfy $\im \left( m_\pm(\theta,E+i0) \right)> 0$ and 
\begin{eqnarray} \label{eq_kotanikeyresult}
\int_\mathbb{T} \left( \dfrac{1}{c(\theta) \im \left(m_+(\theta, E + i 0) \right)} \right) \ud \mu(\theta) < \infty ~\mbox{, } \nonumber \\
\int_\mathbb{T} \left( \dfrac{1}{c( \theta - \alpha ) \im \left(m_-(\theta, E + i 0) \right)} \right) \ud \mu(\theta) < \infty
\end{eqnarray}
Moreover, one has
\begin{eqnarray} \label{m5_eq_kotanirel}
c(\theta) \im \left( m_+(\theta, E + i 0) \right) =  c(\theta - \alpha) \im \left( m_-(\theta, E+i0) \right) ~\mbox{,} \nonumber \\
 \quad \quad \re \left\{ E - v(\theta) + c(\theta)^2 m_+(\theta, E+ i 0) + c(\theta - \alpha)^2 m_-(\theta, E+i 0) \right\} = 0  ~\mbox{.}
\end{eqnarray}
\end{theorem}

Following, it is convenient to use the natural action of a given cocycle $(\alpha, D)$ on $\mathbb{T} \times \overline{\mathbb{C}}$ by identifying $v=(\begin{smallmatrix} v_1 \\ v_2 \end{smallmatrix}) \in \mathbb{C}^2\setminus \{0\}$ with $z= \frac{v_1}{v_2} \in \overline{\mathbb{C}}$ , in which case $D(x) \cdot z = \frac{a(x) z + b(x)}{c(x)z + d(x)}$, for $D(x) =(\begin{smallmatrix} a(x) & b(x) \\  c(x) & d(x) \end{smallmatrix})$.

Thus, letting
\begin{eqnarray}
s_+(\theta,z) := & \psi_+(1, \theta - \alpha, z) = & - c(\theta- \alpha) m_+(\theta-\alpha, z)  ~\mbox{, } \\ s_-(\theta,z):= & \psi_-(1, \theta - \alpha, z) = & \dfrac{-1}{c(\theta-\alpha) m_-(\theta, z)} ~\mbox{,}
\end{eqnarray}
(\ref{m5_eq_covarianceblochw}) and (\ref{m5_eq_kotanicocyl}) imply that $s_\pm(\theta,z) \in \mathbb{C}\setminus\{0\}$ are {\em{invariant sections}} for $(\alpha, (A^z)^\sharp)$, i.e.
\begin{equation}
(A^z)^\sharp(\theta) \cdot s_\pm(\theta, z) = s_\pm(\theta+\alpha, z) ~\mbox{.}
\end{equation}
We mention that, since $\psi_\pm$ exhibit exponential decay (uniformly in $\theta$) at respectively $\pm \infty$, the $(\alpha, (A^E)^\sharp)$-invariant splitting just recovers the fact that $(\alpha, (A^z)^\sharp)$ is uniformly hyperbolic for $z \in \mathbb{H}^+$ (\cite{JohnsonMoser_1982}; see also \cite{Marx_2014} for an appropriate generalization to singular operators).

\begin{proof}[Proof of Theorem \ref{m5_thm_kotani}]
Let $E \in \mathcal{Z}$ be fixed. For $\mu$-a.e. $\theta$, Theorem \ref{m5_thm_kotanilemma} allows to extend the solutions $\psi_\pm(. ,\theta,z)$ to $z=E + i 0$ using, respectively, (\ref{m5_eq_limitingsol+}) and (\ref{m5_eq_kotanirel}). The resulting random sequences $\psi_\pm(. , \theta, E+i0)$ relate according to
\begin{eqnarray} \label{m5_eq_blochwaves_1}
\re ( \psi_-(. , \theta, E+i0) ) = \re ( \psi_+(. , \theta, E+i0) ) ~\mbox{,} \nonumber \\
\im ( \psi_-(. , \theta, E+i0) ) = - \frac{c(\theta - \alpha)}{c(\theta)} \im ( \psi_+(. , \theta, E+i0) ) ~\mbox{,}
\end{eqnarray}
for all $(\theta, E)$ where they are defined. To see this, observe that by (\ref{m5_eq_kotanirel}), (\ref{m5_eq_blochwaves_1}) is satisfied at $n = -1$. Since (\ref{m5_eq_blochwaves_1}) also holds true trivially at $n=0$, it holds for all $n \in \mathbb{Z}$. 

Thus, rewriting (\ref{m5_eq_blochwaves_1}) in terms of $s_\pm(\theta, E+i0)$, we conclude that
\begin{eqnarray} 
\re s_-(\theta, E+i0) = \re s_+(\theta, E+i0) ~\mbox{, } \nonumber \\ 
\im s_-(\theta, E+i0) =  - \frac{c(\theta - 2 \alpha)}{c(\theta-\alpha)} \im s_+(\theta, E+i0)  ~\mbox{,}  \label{m5_eq_blochwaves} \\
\{\im s_\pm(. , E+i0)\}^{-1} \in L^1(\mathbb{T}) \label{eq_l2cond} ~\mbox{.}
\end{eqnarray}
For Schr\"odinger operators, (\ref{m5_eq_blochwaves}) recovers that $s_\pm(\theta,E+i0)$ and hence $\psi_\pm(\theta,E+i0)$ are merely complex conjugates, the latter of which was key for the proof presented in \cite{DeiftSimon_1983}.

Even though this is not the case in general for Jacobi operators, since $(A^E)^\sharp$ is {\em{real}}, $\overline{s_+(., E+i0)}$ automatically yields an invariant section as well. Hence letting $C(\theta,E)$ be the matrix with column vectors $(\begin{smallmatrix} s_+(\theta, E+i0) \\ 1 \end{smallmatrix})$ and $(\begin{smallmatrix} \overline{s_+(\theta, E+i0)} \\ 1 \end{smallmatrix})$, $C^\sharp:=C/\sqrt{\det(C)}$ mediates a conjugacy over $SL(2,\mathbb{C})$ to a {\em{complex}} rotation, which is $L^2$ by (\ref{eq_l2cond}), (\ref{m5_eq_kotanicocyl}), and (\ref{m5_eq_kotani_assumptions}).  Finally, since the columns of $C$ are complex conjugates, $D = C^\sharp (\begin{smallmatrix} 1 & i \\ 1 & -i \end{smallmatrix}) \in SL(2,\mathbb{R})$ sets up a conjugacy over $SL(2, \mathbb{R})$ to a {\em{real}}, not necessarily constant, rotation.
\end{proof}

\section{Almost reducibility implies absolute continuity} \label{sec_almredimpliesac}

We consider a non-singular Jacobi operator. Theorem \ref{thm_subcritical} identifies the set of subcritical energies as a support of the ac spectrum which, in addition, carries no singular spectrum. As mentioned earlier, this result relies on the almost reducibility theorem (ART). ART originated from a series of works on quasi-periodic Schr\"odinger cocycles \cite{AvilaJitomirskaya_2010, Avila_2008, Avila_prep_ARC_1, Avila_prep_ARC_2} which sought to characterize the cocycle dynamics on the set of zero Lyapunov exponent. In this quest, the relevant dynamical framework turned out to be notion of {\em{almost reducibility}}: 
\begin{definition} \label{def_almostred}
A cocycle $(\alpha, A)$ with $A \in \mathcal{C}^\omega(\mathbb{T}, \mathrm{SL}(2, \mathbb{R}))$ is called almost reducible if the closure of its conjugacy class contains a {\em{constant}} rotation, i.e., if for some sequence $B_n \in \mathcal{C}^\omega(\mathbb{T}, \mathrm{PSL}(2,\mathbb{R}))$, $B_n(x + \alpha)^{-1} A(x) B_n(x) \to R$ in $\mathcal{C}^\omega$-topology for some constant rotation $R$.
\end{definition}
For Schr\"odinger operators almost reducibility was first proven for analytic potentials dual to long-range operators which exhibit localization \cite{AvilaJitomirskaya_2010}. In particular, almost reducibility was shown to occur for all energies in the spectrum for the subcritical almost Mathieu operator ($v(\theta) = 2 \lambda \cos(2\pi \theta)$ with $\vert \lambda \vert < 1$). The latter was then proven to imply pure ac spectrum. With the development of the GT it was thus natural to conjecture that, in general, subcritical behavior implies almost reducibility (the reverse implication holds trivially).

ART verifies this conjecture, establishing the equivalence of almost reducibly and subcriticality. The remaining spectral theoretic step to Theorem \ref{thm_subcritical} is to show that almost reducibility implies pure ac spectrum. For Schr\"odinger operators this was first proven in \cite{ajdry} for Diophantine $\alpha$, using an argument that essentially dates back to Eliasson \cite{eli}. Later, in \cite{Avila_prep_ARC_1}, this result was extended to all irrational $\alpha$ and $\mu$-a.e. $\theta$. A proof for {\em{all}} phases is much more delicate and is to appear in \cite{Avila_prep_ARC_2}. 

In this section we give a proof of the ``a.e. phase statement'' valid for any non-singular, quasi-periodic Jacobi operator; the statement for a.e. phase is sufficient for the conclusions in Theorem \ref{thm_ehmspectral}. Rather than adapting the argument for Schr\"odinger operators given in \cite{Avila_prep_ARC_1}, we take a slightly different route which shortens the original proof for the Schr\"odinger case. Using the same terminology as in Sec. \ref{sec_afk_reductions}, we thus claim:

\begin{theorem}[``almost reducibly implies absolute continuity''] \label{thm_arimpliesac}
Consider a non-singular, analytic Jacobi operator $H_\theta$ with $\alpha$ irrational such that the set 
\begin{equation}
\Sigma_{ar} := \{E \in \mathbb{R} : (\alpha, (A^E)^\sharp) ~\mbox{is almost reducible} \}
\end{equation}
is non-empty. Then, for $\mu$-a.e. $\theta \in \mathbb{T}$, all spectral measures are purely ac on $\Sigma_{ar}$.
\end{theorem}

The key ingredient in the proof of Theorem \ref{thm_arimpliesac} is that almost reducibility for an analytic $\mathrm{SL}(2, \mathbb{R})$-cocycle $(\alpha, A)$ already implies $\mathcal{C}^\omega$-reducibility at least if its rotation number $\rho(\alpha, A)$ satisfies a certain Diophantine condition; the latter is made precise in Theorem \ref{thm_arimpliesred}. To formulate it, given $\epsilon >0$, $0 < \nu < \frac{1}{2}$, and $\tau > 0$, denote by $Q_\alpha(\tau, \nu, \epsilon) \subseteq \mathbb{T}$ the set of all $\rho$ such that for all $n \in \mathbb{N}$,
\begin{equation}
\vert\vert\vert 2 \rho q_n \vert\vert\vert > \epsilon \max \{ q_{n+1}^{-\nu} , q_n^{-\tau}  \} ~\mbox{.}
\end{equation}

Here, we recall that for an analytic $\mathrm{SL}(2,\mathbb{R})$-cocycle $(\alpha,D)$ which is homotopic to the identity, its {\it{fibered rotation number}} $\rho(\alpha, D)$ is defined as follows: Let $\tilde{F}: \mathbb{T}^\nu \times \mathbb{R} \to  \mathbb{T}^\nu \times \mathbb{R}$ be a continuous lift of the map $(\theta, v) \mapsto (\theta + \alpha, \frac{D(\theta) v}{\Vert D(\theta) v \Vert})$ on $\mathbb{T}^\nu \times S^1$. Naturally, any such lift $\tilde F$ can be written in the form $\tilde{F}(\theta, x) = (\theta + \alpha, x + f(\theta,x))$, for some continuous $f$ satisfying $f(\theta, x+1) = f(\theta, x)$. The fibered rotation number $\rho(\alpha,D)$ is then defined by the limit,
\begin{equation}
\rho(\alpha,D):= \lim_{n \to \pm \infty} \frac{1}{n} \sum_{k=0}^{n-1} f(\tilde{F}^k(\theta,x) ~(\mathrm{mod} 1)) \in \mathbb{T} ~\mbox{,}
\end{equation}
which is independent of the lift and converges uniformly in $(\theta,x)$ to a constant with continuous dependence on the cocycle \cite{JohnsonMoser_1982, Herman_1983, DelyonSouillard_1983}. For our applications it will be important to note that the fibered rotation number is in general not preserved under conjugacies. In fact, conjugacy may change the fibered rotation number by an element of $\mathbb{Z} \oplus \alpha \mathbb{Z}$, if the change of coordinates is not isotopic to a constant. In what follows, we will denote $\rho(\alpha, (A^E)^\sharp) =: \rho(\alpha, E)$ to simplify notation.

The key ingredient in the proof of Theorem \ref{thm_arimpliesac} is given by the following theorem, Theorem \ref{thm_arimpliesred}, which results from a combination of Theorem 1.3 in \cite{AvilaFayadKrikorian_2011} and Theorem 1.4 in \cite{Avila_prep_ARC_1}. To keep this paper as self-contained as possible, we include its proof below. We also mention that Theorem \ref{thm_arimpliesred} is in fact stated in \cite{Avila_prep_ARC_1} as Corollary 1.5, however without explicitly quantifying the set of non-resonant rotation numbers, $Q_\alpha(\tau, \nu, \epsilon)$. 
\begin{theorem} \label{thm_arimpliesred}
Suppose $(\alpha, A)$ is almost reducible. If $\rho(\alpha, A) \in Q_\alpha(\tau, \nu, \epsilon)$, for some $\epsilon > 0$, $0 < \nu < \frac{1}{2}$, and $\tau > 0$, then $(\alpha, A)$ is $\mathcal{C}^\omega$-reducible.
 \end{theorem}
 \begin{proof}
Since $(\alpha, A)$ is almost reducible and non-uniformly hyperbolic, Theorem 1.4 of \cite{Avila_prep_ARC_1} implies that the elements of the sequence $B_n$ in Definition \ref{def_almostred} can be chosen such that, for each $n \in \mathbb{N}$, one has  that $B_n \in \mathcal{C}^\omega(\mathbb{T}, \mathrm{SL}(2, \mathbb{R}))$ and $B_n$ is {\em{homotopic to a constant}}. As mentioned above, conjugacies mediated by a change of coordinates which are homotopic to a constant preserve the rotation number, thus we conclude that for each $n \in \mathbb{N}$, the matrices
\begin{equation}
\tilde{A_n}(x):= B_n(x + \alpha)^{-1} A(x) B_n(x)
\end{equation}
satisfy 
\begin{equation} \label{eq_rotationnumberpreserv}
\rho(\alpha, \tilde{A_n}) = \rho(\alpha, A) \in Q_\alpha(\tau, \nu, \epsilon) ~\mbox{.}
\end{equation}

On the other hand, Theorem 1.3 of \cite{AvilaFayadKrikorian_2011} guarantees that there exists $\eta = \eta(\tau, \nu, \epsilon)$ such that for every analytic $\mathrm{SL}(2, \mathbb{R})$-cocycle $(\alpha, C)$ with $\rho(\alpha, C) \in Q_\alpha(\tau, \nu, \epsilon)$ which is $\eta$-close (in the analytic category) to a (not necessarily constant) rotation, one can conclude that $(\alpha, C)$ is in fact $\mathcal{C}^\omega$-reducible. Thus, taking $n \in \mathbb{N}$ such that $\tilde{A_n}$ is $\eta$-close to the (not necessarily constant) rotation $R$ originating from almost reducibility, (\ref{eq_rotationnumberpreserv}) and Theorem 1.3 of \cite{AvilaFayadKrikorian_2011} implies that $(\alpha, \tilde{A_n})$, and hence $(\alpha, A)$, is $\mathcal{C}^\omega$-reducible.
\end{proof}

Equipped with Theorem \ref{thm_arimpliesred}, we are ready to prove Theorem \ref{thm_arimpliesac}.
\begin{proof}[Proof of Theorem \ref{thm_arimpliesac}]
Fix some $0 < \tau$ and $0 < \nu < \frac{1}{2}$. Suppose that for some $\epsilon > 0$, $E \in \Sigma_{ar}$ is such that $\rho(\alpha, E) \in Q_\alpha(\tau, \nu, \epsilon)$. Then, by Theorem \ref{thm_arimpliesred}, $(\alpha, (A^E)^\sharp)$ is $\mathcal{C}^\omega$-reducible, which, using (\ref{m5_eq_kotanicocyl}), implies that all solutions of $H_\theta \psi = E \psi$ are bounded uniformly in $\theta$. Thus the set
\begin{equation}
\Sigma_b:= \{ E \in \Sigma_{ar} ~:~ \rho(\alpha, E) \in Q_\alpha(\tau, \nu, \epsilon) ~\mbox{, for some $\epsilon > 0$} \} ~\mbox{,}
\end{equation}
supports only absolutely continuous spectrum, for all $\theta \in \mathbb{T}$.

On the other hand, note that $\mu \left( \mathbb{T} \setminus Q_\alpha(\tau, \nu, \epsilon) \right) \leq \epsilon \sum_{n \in \mathbb{N}} \max \{ q_{n+1}^{-\nu} , q_n^{-\tau}\}$, whence $\Omega = \cap_{\epsilon > 0} (\mathbb{T} \setminus Q_\alpha(\tau, \nu, \epsilon))$ is a set of zero $\mu$-measure. Since $\rho(\alpha, E) = 1 - 2 N(\alpha, E)$ where $N(\alpha, E) = n( (-\infty, E])$ is the integrated density of states and $\Sigma_{ar} \setminus \Sigma_b \subseteq \rho^{-1}(\alpha, .) \left( \Omega \right)$, we conclude that $n( \Sigma_{ar} \setminus \Sigma_b) = 0$. From the definition of the latter in (\ref{eq_dosmeasure}), this already implies the claim. Here, we made use of continuity of the density of states measure and the following general fact:
\begin{fact} \label{fact_measure}
Let $\mu$ be a {\em{continuous}} Borel probability measure\footnote{Note that without the hypothesis of continuity of $\mu$ the statement becomes radically false; indeed, if $\mu$ has atoms, the measure $\mu \circ F_\mu^{-1}$ is not even absolutely continuous w.r.t. to $\mu_L$. To see this explicitly, take $\mu = \frac{1}{2} (\delta_{1/2} + \mu_L)$ on $[0,1]$. Then, the set $S = \{\frac{3}{4} \}$ is of zero Lebesgue measure nevertheless, $(\mu \circ F_\mu^{-1})(S) = \frac{1}{2} > 0$.}
 on $\mathbb{R}$ and $F_\mu$ its cumulative distribution. Then,
\begin{equation} \label{eq_factmeasure}
\mu \circ F_\mu^{-1} = \mu_L ~\mbox{.}
\end{equation}
Here, $\mu_L$ denotes the Lebesgue measure on $[0,1]$.
\end{fact}
Fact \ref{fact_measure} follows immediately by verifying (\ref{eq_factmeasure}) for half-open intervals $(a,b] \subseteq [0,1]$. 
\end{proof}

\appendix
\section{Proof of Lemma \ref {lem_roc1}}\label{c}
\begin{proof}
Denote by $\hat{h}_k$ the $k$-th Fourier coefficient of $h$. For $n \in \mathbb{N}$, we decompose
\begin{equation}
h= h_n^{(1)} + h_n^{(2)} =: \sum_{\abs{k} \leq q_n} \hat{h}_k \mathrm{e}^{2 \pi i k x} + \sum_{\abs{k} > q_n} \hat{h}_k \mathrm{e}^{2 \pi i k x} ~\mbox{.}
\end{equation}
Since, 
\begin{equation}
\left \vert \frac{1}{q_n} \sum_{j=0}^{q_n -1} h_n^{(2)}(x+ j \alpha) \right \vert \leq \sum_{\abs{k} > q_n} \abs{\hat{h}_k} ~\mbox{,}
\end{equation}
and $h$ is harmonic, we obtain
\begin{eqnarray}
\left \vert \frac{1}{q_n} \sum_{j=0}^{q_n-1} h(x + j \alpha) - \hat{h}_0 \right \vert & \leq & \left \vert \frac{1}{q_n} \sum_{j=0}^{q_n-1} h_n^{(1)}(x + j \alpha) - \hat{h}_0 \right \vert  + \mathcal{O}(\frac{1}{q_n})  \nonumber \\
& = & \left \vert \frac{1}{q_n} \sum_{0 < \abs{k} \leq q_n} \hat{h}_k \mathrm{e}^{2 \pi i x k} \dfrac{1 - \mathrm{e}^{2 \pi i k q_n \alpha}}{1 - \mathrm{e}^{2 \pi i k \alpha}} \right \vert + \mathcal{O}(\frac{1}{q_n}) ~\mbox{.}
\end{eqnarray}
The basic estimates (\ref{eq_contifracbasic}) imply for $\abs{k} < q_{n+1}$
\begin{eqnarray}
\abs{1 - \mathrm{e}^{2 \pi i k \alpha}} & \gtrsim & \frac{1}{q_{n+1}} ~\mbox{,} \\
\abs{1 - \mathrm{e}^{2 \pi i k q_n \alpha}} & \lesssim & \frac{1}{q_{n+1}} \abs{k} ~\mbox{.}
\end{eqnarray}
Thus we finally conclude
\begin{equation}
 \left \vert \frac{1}{q_n} \sum_{j=0}^{q_n-1} h_n^{(1)}(x + j \alpha) - \hat{h}_0 \right \vert \lesssim \frac{1}{q_n} \sum_{k \in \mathbb{Z}} \abs{\hat{h}_k} \abs{k} ~\mbox{,}
\end{equation}
where the right hand side is summable based on harmonicity of $h$.
\end{proof}

\section{Comments on Theorem \ref{thm_global}} \label{app_thm_accel}
As mentioned, Theorem \ref{thm_global} combines results from various articles, specifically the papers \cite{JitomirskayaMarx_2012, JitomirskayaMarx_2013_erratum, AvilaJitomirskayaSadel_2013, Marx_2014}. Since certain aspects have meanwhile been simplified, the purpose of this section is to assemble these results in a more streamlined form. In this spirit, when referring to a particular result in the literature, we will quote its latest, most general, available formulation. For an account of some of the underlying historical developments, we refer the interested reader to the survey article \cite{JitomirskayaMarx_ETDS_2016_review}.

\begin{proof}[Proof of Theorem \ref{thm_global}]
Fix $E \in \mathbb{R}$. Convexity in $\epsilon$ of $L(E; \epsilon)$ is equivalent to proving convexity of 
\begin{equation}
L(\alpha, A_\epsilon^E) = \lim_{n \to \infty} \dfrac{1}{n} \int_{\mathbb{T}} \log \Vert A^E(\theta + i \epsilon +(n-1) \alpha) \dots A^E(\theta + i \epsilon) \Vert ~\ud \mu(\theta) ~\mbox{,}
\end{equation}
which clearly is implied by showing that for each {\em{fixed}} $n \in \mathbb{N}$, 
\begin{equation} \label{eq_thmaccel_conv}
\int_{\mathbb{T}} \log \Vert A^E(\theta + i \epsilon +(n-1) \alpha) \dots A^E(\theta + i \epsilon) \Vert ~\ud \mu(\theta) 
\end{equation}
is convex in $\epsilon$. Since analyticity of the cocycle implies that the integrand of (\ref{eq_thmaccel_conv}) is subharmonic, the convexity in question is as an immediate consequence of the following general fact about averages of subharmonic functions, which is usually attributed to Hardy:
\begin{theorem}[``Hardy's convexity theorem,'' see e.g. Theorem 1.6 in \cite{Duren_HpSpacesBook}] \label{thm_hardy}
For $\delta > 0$, let $u$ be a subharmonic function on the strip $\{ x+ i \epsilon ~\vert~ x \in \mathbb{T}, \vert \epsilon \vert \leq \delta\}$. Consider the averages,
\begin{equation*}
\langle u \rangle(\epsilon) := \int_\mathbb{T} u(x + i \epsilon) ~\ud x \mbox{, $\vert \epsilon \vert \leq \delta$.}
\end{equation*}
Then, either $\langle u \rangle(\epsilon) = -\infty$ for all $\vert \epsilon \vert \leq \delta$, or $\epsilon \mapsto \langle u \rangle(\epsilon)$ is convex.
\end{theorem}

Quantization of the acceleration, i.e. $\omega(E;\epsilon) \in \frac{1}{2} \mathbb{Z}$, follows from Theorem 1.4 of \cite{AvilaJitomirskayaSadel_2013} where the respective result is proven in general for all (possibly singular) analytic cocycles.

To see that $\epsilon \mapsto L(E; \epsilon)$ is even, we use that $(\alpha, A^E)$ is {\em{measurably}} conjugate to the analytic cocycle $(\alpha, \widetilde{A}^E)$ where
\begin{equation} \label{eq_altcocycle}
\widetilde{A}^E(\theta) := \begin{pmatrix} E - v(z) & - \widetilde{c}(\theta - \alpha) c(\theta - \alpha) \\ 1 & 0 \end{pmatrix} ~\mbox{.}
\end{equation}
Here, the measurable conjugacy is given by
\begin{equation} \label{eq_conjaltcocycl}
M(\theta + \alpha)^{-1}(\theta) \widetilde{A}^E(\theta) M(\theta) = A^E(\theta) ~\mbox{, } M(\theta) = \begin{pmatrix} 1 & 0 \\ 0 & c(\theta - \alpha)^{-1} \end{pmatrix} ~\mbox{.}
\end{equation}
We mention that the conjugacy in (\ref{eq_conjaltcocycl}) played an important role in \cite{Marx_2014}.

The crucial observation for our purposes is that $\widetilde{A}^E$ is {\em{real-symmetric}} and analytic, whence, using the reflection principle, $L(\alpha, \widetilde{A}^E_\epsilon)$ is even in $\epsilon$. Since measurable conjugacies preserve the Lyapunov exponent, we conclude that $L(\alpha, A^E_\epsilon)$, and hence $L(E; \epsilon)$, is an even function in $\epsilon$.

Naturally, evenness and convexity of $\epsilon \mapsto L(E; \epsilon)$ necessitates that it monotonically increases on the non-negative real axis. In particular, $L(E)=L(E; 0) \geq 0$, implies that $L(E; \epsilon) \geq 0$ for all $\epsilon$. In summary, we conclude that $L(E; \epsilon)$ is a non-negative piece-wise linear and convex function in $\epsilon$, as claimed.

Finally, it was proven in \cite{Marx_2014} that for every (possibly singular) quasi-periodic Jacobi operator, $E \not \in \Sigma$ if and only if $(\alpha, A^E)$ induces a {\em{dominated splitting}}. We recall that an analytic cocycle $(\alpha, D)$ is said to induce a dominated splitting if there exists a continuous (in $\theta$), nontrivial splitting of $\mathbb{C}^2=E_\theta^{(1)}\oplus E_\theta^{(2)}$ and $N\in\mathbb{N}$ such that for $1 \leq j \leq 2$ and each $\theta \in \mathbb{T}$, one has $D^{(N)}(\theta; \alpha) E_\theta^{(j)}\subseteq E_{\theta+N\alpha}^{(j)}$ and $\frac{\Vert D^{(N)}(\theta; \alpha) v_1\Vert}{\Vert v_1\Vert}>\frac{\Vert D^{(N)}(\theta; \alpha)v_2\Vert}{\Vert v_2\Vert}$, for all $v_j\in E_\theta^{(j)}\setminus\{0\}$. Here, as earlier, $D^{(N)}(\theta ; \alpha) = \prod_{j = N-1}^{0} D(\theta + j \alpha)$ denotes the iterates of the cocycle on the fibers.

Moreover, it is a consequence of \cite{AvilaJitomirskayaSadel_2013} (see Theorem 1.2, therein) that $(\alpha, A^E)$ induces a dominated splitting if and only if $L(E) > 0$ and the acceleration is locally zero in a neighborhood of $\epsilon = 0$. 

Thus, combining these two dynamical results, we conclude that for every $E \in \mathbb{R}$ with $L(E) > 0$, $E \in \Sigma$ if and only if $\omega(E; 0) > 0$, or equivalently, $\epsilon \mapsto \omega(E; \epsilon)$ has a jump-discontinuity at $\epsilon = 0$.
\end{proof}

\section{Proof of Proposition \ref{prop_det}} \label{app_aubry}

For every  $x \in \mathbb{T}_0(\sigma(\lambda))$, (\ref{eq_semiconj}) yields
\begin{equation} \label{eq_1}
\abs{\det M_\theta(x)} \abs{c(x - \alpha)} = \abs{\det M_\theta(x + \alpha)} \abs{c(x)} ~\mbox{,}
\end{equation}
which by ergodicity of irrational rotations already implies $\abs{\det M_\theta(x)} \abs{c(x - \alpha)} = b$ a.e. for some $b \geq 0$. Since $c(x) \neq 0$ on $\mathbb{T}_0(\sigma(\lambda))$, we conclude $b >0$ if and only if $\det M_\theta(x) \neq 0$ a.e. We mention that by (\ref{eq_1})  the set $\{x \in \mathbb{T}_0(\sigma(\lambda)): \det M_\theta(x) = 0\}$ is invariant under rotations whence it can only be of $\mu$-measure zero or one.

Seeking a contradiction, suppose that $\det M_\theta(x) = 0$ a.e., then there exists $\phi(x)$ such that for a.e. $x$
\begin{equation} \label{eq_2}
\begin{pmatrix} u(x) \\ \mathrm{e}^{-2\pi i \theta} u(x-\alpha) \end{pmatrix} = \phi(x) \begin{pmatrix} u(-x) \\  \mathrm{e}^{2 \pi i \theta} u(-(x-\alpha))  \end{pmatrix} ~\mbox{.}
\end{equation}
In particular, $\phi(x) = \frac{u(x)}{u(-x)} \in \overline{\mathbb{C}}$ is a non-identically vanishing, measurable function on $\mathbb{T}$.
(\ref{eq_2}) implies
\begin{equation} \label{eq_3}
\phi(x + \alpha) = \mathrm{e}^{- 4 \pi i \theta} \phi(x)  ~\mbox{, a.e.}
\end{equation}
By ergodicity, $\abs{\phi(x)} = b^\prime$ for some $b^\prime \neq 0$, in particular, $\phi \in L^1(\mathbb{T})$. 

Writing $\phi(x) = \sum_{n \in \mathbb{Z}} \hat{\phi}_n \mathrm{e}^{2 \pi i n x}$, we conclude from (\ref{eq_3}) 
\begin{equation} \label{eq_4}
\hat{\phi}_n \left( \mathrm{e}^{2 \pi i n \alpha + 4 \pi i \theta} - 1 \right) = 0 ~\mbox{,} ~\forall n \in \mathbb{Z} ~\mbox{.}
\end{equation}
Since by hypotheses, we excluded all $\theta$ which are $\alpha$-rational,  (\ref{eq_4}) implies $\phi \equiv 0$ - a contradiction. 

\section{Zero nearest neighbor coupling} \label{app_zeronn}

In this section we present the necessary adaptations for the case $\lambda_2 = 0$, i.e. $\lambda = (\lambda_1, 0, \lambda_3)$ and $\lambda_1 + \lambda_3 \geq 1$. In this situation, the duality map $\sigma$ as given in (\ref{eq_sigma}) needs to be redefined appropriately. 

To this end, let us assume, similarly to Sec. \ref{sec_Aubry}, that $(u_n)$ is an $\mathit{l}^2$-eigenvector of $H_{\theta; \lambda, \alpha}$. Denoting by $u(x)$ its Fourier transform, we compute:
\begin{equation}
u(x-\alpha) \mathrm{e}^{-2 \pi i \theta} \overline{c_{ \sigma(\lambda)}(x-\alpha)} + u(x+\alpha) \mathrm{e}^{2 \pi i \theta} c_{\sigma(\lambda)}(x) = E u(x) ~\mbox{,}
\end{equation}
where we redefine the duality map $\sigma$ according to
\begin{equation} \label{eq_dualityapp1}
\sigma \left(\lambda_1, 0, \lambda_3\right)  : = (\lambda_3, 1, \lambda_1) ~\mbox{.}
\end{equation}

In particular, (\ref{eq_dualityapp1}) implies that the formulation of Aubry-duality given in (\ref{eq_semiconj}) carries over when replacing $B_{\sigma(\lambda)}^E$  by 
\begin{equation} \label{eq_modcocl}
\widetilde{B}_{\sigma(\lambda)}^E(x) : = \frac{1}{c_{\sigma(\lambda)}} \begin{pmatrix} E & -\overline{c_{ \sigma(\lambda)}(x-\alpha)} \\ c_{\sigma(\lambda)}(x) & 0 \end{pmatrix} ~\mbox{.}
\end{equation}

Notice that the determinant of the cocyle is unaffected by the adaptations of this section, i.e.
\begin{equation}
\det \widetilde{B}_{ \sigma(\lambda)}^E(x) = \dfrac{\overline{c_{ \sigma(\lambda)}(x-\alpha)}}{c_{ \sigma(\lambda)}(x)} ~\mbox{,}
\end{equation}
whence Proposition \ref{prop_det} and thus its corollary, Proposition \ref{coro_bddsol}, carry over literally. By the same reasoning as in Sec. \ref{sec_selfdual}, letting (cf. (\ref{eq_psi}))
\begin{eqnarray}
\Psi_{\sigma(\lambda)}^{(n)}(x) & := &\mathrm{tr}\left\{d_{ \sigma(\lambda)}^{(n)}(x) \left(\widetilde{B}_{\sigma(\lambda);n}^E(x) - R_\theta^n\right)\right\}  \\
& = & \mathrm{tr}\left(d_{\sigma(\lambda)}^{(n)}(x) \widetilde{B}_{ \sigma(\lambda);n}^E(x) \right)- 2 d_{ \sigma(\lambda)}^{(n)}(x) \cos(2 \pi n \theta) ~\mbox{,}
\end{eqnarray}
we obtain by (\ref{eq_integral}), for $\lambda_1 + \lambda_3 \geq 1$ and all irrational $\alpha$,
\begin{equation}
\Vert \Psi_{\sigma(\lambda)}^{(m_l)} \Vert_{L^1(\mathbb{T})} \leq C_l \abs{\lambda_1 \vee \lambda_3}^{m_l} ~\mbox{, $C_{l} = o(1)$ , }
\end{equation}
where $(m_l):=(q_{n_l}) \cup (2 q_{n_l})$ and, as earlier, $(q_{n_l})$ is the subsequence of $(q_n)$ provided by Theorem \ref{prop_prozero}. We claim:
\begin{theorem} \label{thm_zerolambda2}
Let $\alpha$ be irrational and $\lambda=(\lambda_1, 0, \lambda_3)$ with $\lambda_1 + \lambda_3 \geq 1$. For a.e. $\theta$, (\ref{eq_19a}) holds. In particular, for all irrational $\alpha$, $H_{\theta; \lambda, \alpha}$ has empty point spectrum for a.e. $\theta \in \mathbb{T}$ which are non-$\alpha$-rational.
\end{theorem}
\begin{remark}
As in the case $\lambda_2 \neq 0$,  (\ref{eq_19a}) holds for {\em{all}} non-$\alpha$-rational $\theta$ if $\lambda_1 \neq \lambda_3$.
\end{remark}

\begin{proof}
We follow the line of argument presented in Sec. \ref{sec_selfdual}, in particular, without loss of generality we assume that $\lambda_1  \vee \lambda_3 = \lambda_1$. 

For $n \in \mathbb{N}$ one computes,
\begin{equation} \label{eq_zero_psi}
\widehat{\Psi_{\sigma(\lambda)}^{(n)}}(\pm n) = \widehat{\phi_{ \sigma(\lambda)}^{(n)}}(\pm n) - 2 \cos(2 \pi n \theta) \cdot \begin{cases}
\lambda_1^n \mathrm{e}^{\pi i \alpha n^2} & \mbox{, for } +n ~\mbox{,} \\  \lambda_3^n \mathrm{e}^{- \pi i \alpha n^2} & \mbox{, for } -n ~\mbox{,} \end{cases} 
\end{equation}
where $\phi_{\sigma(\lambda)}^{(n)}:= \mathrm{tr}\left(d_{\sigma(\lambda)}^{(n)} \widetilde{B}_{\sigma(\lambda);n}^{E} \right)$. 

The form of $\widetilde{B}_{\sigma(\lambda)}^E$ implies that $\phi_{\sigma(\lambda)}^{(n)}$ relates to cutoffs of the following Jacobi matrix 
\begin{equation}
\hat H_{x; \lambda, \alpha} = \begin{pmatrix}  0 & c_\lambda(x) & & \\ \overline{c_\lambda(x)}  & 0 & c_\lambda(x+\alpha) & \\ & \overline{c_\lambda(x+\alpha)} & 0 & c_\lambda(x+ 2\alpha) & \\ & & \ddots & \ddots & \ddots  \end{pmatrix} ~\mbox{,}
\end{equation}
therefore setting $Q_\lambda^{(n)}(E;x):=\det \left(E- \Pi_{[0,n-1]} H_{x; \lambda,\alpha} \Pi_{[0,n-1]}\right)$, $Q_\lambda^{(0)}(E; x):= 1$, $Q_\lambda^{(-1)}(E;x):=0$, we obtain, as in (\ref{eq_6}),
\begin{equation} \label{eq_zero_phi}
\widehat{\phi_{ \sigma(\lambda)}^{(n)}}(\pm n) = \widehat{Q_{ \sigma(\lambda)}^{(n)}}(\pm n) - \lambda_1 \lambda_3 \mathrm{e}^{\pm 2 \pi i \alpha(2n-3)} \widehat{Q_{ \sigma(\lambda)}^{(n-2)}}(\pm (n-2)) ~\mbox{.}
\end{equation}

Using an analogue of Lemma \ref{lem_1}, $\widehat{Q_{ \sigma(\lambda)}^{(n)}}(\pm n)$ is readily computed which gives
\begin{equation}
\widehat{Q_{\sigma(\lambda)}^{(n)}}(\pm n) = \mathrm{e}^{\pm \pi i \alpha n (n-1)} s_n ~\mbox{,}
\end{equation}
where 
\begin{equation}
s_n:=\det \begin{pmatrix} 0& \widetilde \lambda_1 & & & \\ \widetilde \lambda_3 & 0 & \widetilde \lambda_1 & & \\ & \widetilde \lambda_3 & 0& \widetilde \lambda_1 & \\ & & \ddots & \ddots & \ddots     \end{pmatrix} 
\end{equation}
and $\widetilde \lambda_1:= \lambda_1 \mathrm{e}^{\pi i \alpha}$, $\widetilde \lambda_3:= \lambda_3 \mathrm{e}^{- \pi i \alpha}$.

Expanding the determinant, $s_n$ is seen to satisfy the recursion relation
\begin{equation}
 s_n = - \lambda_1 \lambda_3 s_{n-2} ~\mbox{, } n \in \mathbb{N} ~\mbox{,}
\end{equation}
where we define $s_0:=1$ and $s_{-1}:=0$. Thus, we conclude 
\begin{equation}
s_{2n+1} = 0 ~\mbox{, } s_{2n} =(-1)^n (\lambda_1 \lambda_3)^n ~\mbox{, } n \in \mathbb{N} ~\mbox{.}
\end{equation}

Using (\ref{eq_zero_phi}), one obtains
\begin{equation}
\widehat{\phi_{ \sigma(\lambda)}^{(n)}}(\pm n)= \begin{cases} 
0 & \mbox{, if } n ~\mbox{odd ,} \\
2 \mathrm{e}^{\pm \pi i \alpha n (n-1)} (-1)^{n/2} \left( \lambda_1 \lambda_3 \right)^{n/2} & \mbox{, if } n ~\mbox{even ,}
\end{cases}
\end{equation}
which in turn yields
\begin{equation} \label{eq_zero_psi_1}
\mathrm{e}^{- i \pi \alpha n^2} \dfrac{\widehat{\Psi_{ \sigma(\lambda)}^{(n)}}(n)}{ 2 \lambda_1^n } = \begin{cases}
 -\cos(2 \pi n \theta) & \mbox{, if } n ~\mbox{odd ,} \\
 \mathrm{e}^{-i \pi \alpha n}(-1)^{n/2} \left(\dfrac{\lambda_3}{\lambda_1}\right)^{n/2} -   \cos(2 \pi n \theta) & \mbox{, if } n ~\mbox{even ,}
\end{cases}
\end{equation}
by (\ref{eq_zero_psi}), and similarly for $\widehat{\Psi_{ \sigma(\lambda)}^{(n)}}(-n)$.
As suggested by (\ref{eq_zero_psi_1}), we distinguish the cases $\lambda_1 \neq \lambda_3$ and $\lambda_1 = \lambda_3$.

If $\lambda_1 \neq \lambda_3$, 
\begin{equation} \label{eq_zero_psi_2}
\left\vert  \dfrac{\widehat{\Psi_{\sigma(\lambda)}^{(q_{n_l})}}(q_{n_l})}{ 2 \lambda_1^{q_{n_l}} }     \right\vert \gtrsim \vert \cos(2 \pi q_{n_l} \theta) \vert ~\mbox{,}
\end{equation}
which implies (\ref{eq_19a}) for all $\theta$; here, we use analogous arguments to those of the proof of Proposition \ref{prop_lb} (see Case I (a), therein).

To obtain (\ref{eq_19a}) for the case that $\lambda_1 = \lambda_3$, we first note that, possibly passing to an appropriate subsequence, one may assume the parity of $q_{n_l}$ to be {\em{constant}} in $l$. Employing (\ref{eq_zero_psi_1}), the situation when $q_{n_l}$ is odd for all $l$ reduces to a problem of the form (\ref{eq_zero_psi_2}), whence it suffices to consider $q_{n_l}$ {\em{even}} for all $l$.

Then,
\begin{equation}
\mathrm{e}^{-i \pi \alpha q_n}  (-1)^{q_n/2} \sim (-1)^{p_n + (q_n/2)} ~\mbox{,}
\end{equation}
in analogy to (\ref{eq_asyexprimp}). As before, without loss, one may also assume the parity of both $q_{n_l}/2$ and $(p_{n_l} + (q_{n_l}/2))$ to be {\em{constant}} in $l$.

By  (\ref{eq_zero_psi_1}), 
\begin{equation}
\left\vert  \dfrac{\widehat{\Psi_{\sigma(\lambda)}^{(q_{n_l})}}(q_{n_l})}{ 2 \lambda_1^{q_{n_l}} }  \right\vert \sim \left\vert \cos(2 \pi q_{n_l} \theta) \pm 1 \right\vert ~\mbox{,}
\end{equation}
where the $+$ ($-$) sign applies for, respectively, $(p_{n_l} + (q_{n_l}/2))$ odd (even).

In particular, (\ref{eq_19a}) follows using analogous arguments as in the proof of Proposition \ref{prop_lb} (see Case I (b), therein); as then, the origin of the ``$a.e.$'' statement in Theorem \ref{thm_zerolambda2} is application of Proposition \ref{lem_almostunique}.
\end{proof}

\bibliographystyle{amsplain}

\end{document}